\documentclass{revtex-stylesheet}

\begin{document}
\title{Differentially private quantum sensor networks}

\author{Daniel J. Spencer}
\email{djspence@umd.edu}
\affiliation{Joint Center for Quantum Information and Computer Science, NIST/University of Maryland, College Park, MD 20742, USA}
\affiliation{Joint Quantum Institute, NIST/University of Maryland, College Park, MD 20742, USA}
\affiliation{Department of Physics, University of Maryland, College Park, MD 20742, USA}

\author{Kaiyan Shi}
\affiliation{Joint Center for Quantum Information and Computer Science, NIST/University of Maryland, College Park, MD 20742, USA}
\affiliation{Department of Computer Science, University of Maryland, College Park, MD 20742, USA}

\author{Emil T. Khabiboulline}
\affiliation{Joint Center for Quantum Information and Computer Science, NIST/University of Maryland, College Park, MD 20742, USA}
\affiliation{Joint Quantum Institute, NIST/University of Maryland, College Park, MD 20742, USA}

\author{Gorjan Alagic}
\affiliation{Joint Center for Quantum Information and Computer Science, NIST/University of Maryland, College Park, MD 20742, USA}
\affiliation{Department of Computer Science, University of Maryland, College Park, MD 20742, USA}

\author{Alexey V. Gorshkov}
\affiliation{Joint Center for Quantum Information and Computer Science, NIST/University of Maryland, College Park, MD 20742, USA}
\affiliation{Joint Quantum Institute, NIST/University of Maryland, College Park, MD 20742, USA}

\begin{abstract}
    Quantum sensing is a promising technology capable of demonstrating clear advantage over comparable classical techniques for precise measurement. One application of quantum sensing is in function estimation, which can be done using a network of entangled quantum sensors, allowing for measurements with greater optimal sensitivity than unentangled sensing protocols. In cases where quantum sensor networks will be used to measure data that should remain private (e.g., biomedical data), it is imperative that these protocols include a privacy mechanism to hide sensitive information. In this work, we show that entangled sensor networks are vulnerable to certain privacy-violating attacks. To mitigate these attacks, we introduce secure sensing protocols endowed with differential privacy. We reconcile differential privacy with retaining Heisenberg-limited scaling, and introduce several protocols achieving varying balances between the two. We show that our main protocol, an $n$-node network sensing protocol that injects noise directly into the sensing Hamiltonian, exhibits a tradeoff between the desirable $\bigOh{1/n^2}$ Heisenberg scaling of the mean-squared error of the function estimate and the level of privacy attainable. Under assumptions on the network (a common source of randomness and a constant fraction of honest parties), we show that this protocol is locally implementable and achieves $(\bigOh{1}, \delta)$-differential privacy for arbitrarily small $\delta$ while retaining Heisenberg scaling of the mean-squared error. We prove that our protocols are resilient to attacks by broad classes of classical and quantum adversaries, and find advantages in the privacy-utility tradeoff when using quantum techniques.
\end{abstract}

\maketitle

\section{Introduction}\label{sec:intro}
Quantum sensing~\cite{degen2017} is a promising technology with applications across various disciplines that rely on accurate measurements. It has been shown~\cite{giovannetti2004,giovannetti2006,eldredge2018,belliardo2020,qian2021,bringewatt2021,ehrenberg2023} that entanglement provides an advantage over unentangled protocols in networks of quantum sensors for precision measurement. In quantum sensing, a common problem of interest is \emph{function estimation}, where $q(\boldsymbol{\theta})$ is a function of $n$ parameters $\boldsymbol{\theta} = (\theta_1, \ldots, \theta_n)$, each coupled to a quantum sensor, which can be a qubit, a boson, or any other multi-level quantum system. The protocols proposed in, for example, Refs.~\cite{eldredge2018,qian2021,bringewatt2021,ehrenberg2023}, solve the function estimation problem using a network of entangled qubit sensors. Physically, this function could be a magnetic or electric field to be interpolated at arbitrary positions in space. Examples of such sensor networks have been proposed for a variety of applications in geophysics~\cite{nabighian2005,wright2004}, biomedical imaging~\cite{aslam2023,boto2018,jensen2018,jensen2019,deans2016,xu2006}, dark matter searches~\cite{brady2022}, and the enhancement of atomic clock stability~\cite{guo2020}.

Often in sensing, the figure of merit is taken to be the mean-squared error $\varepsilon_{\mathrm{MSE}}$ of the function estimator, in particular how $\varepsilon_{\mathrm{MSE}}$ scales with $n$, the number of sensors. The above-mentioned entangled sensing protocols achieve the desirable \emph{Heisenberg limit} on $\varepsilon_{\mathrm{MSE}}$, which is an improvement over the best achievable bound with any unentangled strategy, called the \emph{standard quantum limit}. In this work, we restrict ourselves to linear functions of the form $q(\boldsymbol{\theta}) = \defsum{i=1}{n}{\alpha_i\theta_i}$. In particular, we consider the \emph{average} function, that is, we take $\alpha_i = 1/n$ for all $i = 1, \ldots, n$ such that $q(\boldsymbol{\theta}) = \frac{1}{n}\defsum{i=1}{n}{\theta_i}$. In this scenario, the standard quantum limit on $\varepsilon_{\mathrm{MSE}}$ is $\bigOh{1/n}$ while the Heisenberg limit, achievable with access to entangled resources, scales as $\bigOh{1/n^2}$, a quadratic enhancement~\cite{eldredge2018}. The average of the parameters is a common measure of interest that allows for a clear demonstration of Heisenberg scaling advantage and allows us to focus on the privacy aspect of our protocols.
    
The regimes in which quantum sensor networks could be used include the measurement of sensitive data that clients may wish to keep private. For example, in a network of distributed biomedical quantum sensors that can measure the presence of an infectious disease, it may be desirable to determine the average rate of infection while hiding which specific individuals have it. As another example, consider a network of entangled sensors distributed among sites within a country tasked with measuring geophysical signals that could indicate the presence of a valuable natural resource such as oil or strategic minerals. If each site is owned by a different company, and a client (e.g., the government of the country) wants to determine some function of the estimated amount of the resource at each site, then the network should have a way to estimate the function without leaking information about the amount of the resource at each site, information that could be used for some nefarious financial or military purpose.

The security of sensor networks has been studied in several prior works~\cite{huang2019,shettell2022,kasai2022,moore2023,moore2024,bugalho2024,kasai2024,hassani2025,farokhi2026precision}. In this work, we focus on a specific type of privacy-breaching attack called a \emph{differencing attack} against the network. An example of a differencing attack in the quantum sensing setting involves an adversary running the sensing protocol once to get an estimate of the linear function and then running the protocol again, this time excluding a single parameter (the "target"). Using these two results, the adversary can learn the parameter value of the target node. While more creative and complicated attacks can be conceived of, this simple attack exposes any single node in the network, which is catastrophic for any user that wishes to keep their data private in such a sensing application.

To address such attacks, we introduce privacy mechanisms based on \emph{differential privacy} to hide the parameter values of all nodes in the network at the cost of losing some accuracy in the estimate of our function. Differential privacy is a well-studied mechanism~\cite{dwork2006a,dwork2006b} that lends itself naturally to scenarios where privacy may be important in a sensing protocol. Its application is especially important as we approach an era in which quantum sensing is realized commercially. Classically, differential privacy can be combined with additive homomorphic encryption to achieve complementary notions of security in sensor networks~\cite{shi2011privacy-preserving}. In this work, we augment quantum sensor networks, where entanglement plays a key role in both security and sensing. In particular, we aim to maintain quantum-enhanced scaling in the precision of estimating functions like summations, while direct application of classical techniques on classical data does no better than the standard quantum limit.

We introduce several differentially private quantum sensor network protocols that achieve varying levels of optimality, according to a set of three criteria that we define below. We consider two adversarial models, one assuming a trusted central node in the network (the \emph{centralized network setting}) and one assuming no such trusted party (the \emph{decentralized network setting}), and we develop differentially private protocols to protect against attacks in each setting. Removing a trusted curator can encourage contributors who do not want to reveal their raw data. We highlight how the level of trust in the network has significant implications for the accuracy of function estimation. 

It is useful to point out that while there is fundamental measurement noise in the original sensing protocol, it is not sufficient for differential privacy and can average out as the number of rounds of sensing increases. Entanglement bestows anonymity, in that the output of the computation can be revealed while the inputs remain unknown. However, this notion of security does not yield differential privacy, since knowledge about the output itself can reveal whether an input contributed to the computation. As such, to introduce privacy to quantum sensor networks of the type we consider in this paper, we require an explicit mechanism. We find that directly applying the Laplace mechanism, commonly used in classical differential privacy, suffices in the centralized network setting. However, we find that when we consider the decentralized network setting, a \emph{quantum} mechanism actually performs better, which we achieve by directly modifying the Hamiltonian.

The analysis of our quantum mechanism accounts for both classical and quantum adversaries, and we show that this mechanism is secure against both. Our analysis makes use of quantum information-theoretic tools that prescribe a smaller amount of injected noise, compared to classical results, and thus improved utility for the same privacy level. We show that, under additional assumptions about the network, our local mechanism can become optimal according to the criteria that we define below in~\cref{def:optimality-conditions}.

\begin{table*}[tb]
    \centering
    \begin{tblr}{
    colspec={X[c,m]X[c,m]Q[c,m]X[c,m]Q[c,m]X[c,m]X[c,m]X[c,m]},
    hlines,
    vlines,
    row{1}={font=\bfseries}
    }
        Protocol & Assumptions & Noise source & Scaling of $\varepsilon_{\mathrm{MSE}}$ & $\varepsilon$ & $\delta$ & Global or local noise implementation & All criteria of~\cref{def:optimality-conditions}?\\
        Global Laplace mechanism (\cref{sec:centralized-network}) & Trusted curator & Laplace & {\color{ForestGreen}$\bigOh{1/n^2}$} & {\color{ForestGreen}$\bigTheta{1}$} & {\color{ForestGreen}$0$} & {\color{Red}Global} & \xmark\\
        Noisy Hamiltonian protocol (\cref{subsec:noisy-hamiltonian-protocol}) & Worst-case: one honest node & Laplace & {\color{Orange}$\bigOh{1/n^{\alpha}}$} & {\color{Orange}$\Theta(n^{(\alpha-1)/2})$} & {\color{ForestGreen}0} & {\color{ForestGreen}Local} & \xmark\\
        Honest-fraction noisy Hamiltonian protocol (\cref{subsec:nhp-honest-fraction}) & Honest fraction, CSR, average-type queries on sufficiently large subset & Gaussian & {\color{ForestGreen}$\bigOh{1/n^2}$} & {\color{ForestGreen}$\bigTheta{1}$} & {\color{ForestGreen}Arbitrarily small} & {\color{ForestGreen}Local} & \cmark
    \end{tblr}
    \caption{Summary of the optimality conditions achieved by each of the differentially private quantum sensing protocols we introduce in this work. Here, $1 \leq \alpha \leq 2$ is a parameter that can be chosen by the user. The only fully optimal protocol (with additional assumptions), by the criteria in~\cref{def:optimality-conditions}, is the noisy Hamiltonian protocol where the network has a common source of randomness (CSR) and a constant fraction of the nodes are honest. Green text indicates optimal, red indicates non-optimal, and orange indicates that the user can tune the level of optimality (at the expense of the optimality of another figure of merit).}
    \label{tab:summary-of-protocols}
\end{table*}

We note that our work is distinct from previous literature on quantum differential privacy~\cite{zhou2017,yoshida2020,hirche2023,li2023,guan2024,zhong2024}, which is concerned with deriving differentially private quantum algorithms to prevent an adversary from learning whether or not a specific quantum state was used as input based on the output of a quantum circuit. In our work, the input to the problem is a set of purely classical parameters in a quantum Hamiltonian that is then coupled to a qubit. Then, the qubits, with the encoded classical data, are unitarily evolved according to the Hamiltonian, after which we measure the final quantum state and are left with classical measurement results and post-measurement quantum states. Thus, our differential privacy analysis requires an interplay between classical and quantum techniques. Furthermore, we consider differential privacy in the distributed setting~\cite{li2021,guan2024,zhong2024}. We note some recent work connecting classical Fisher information~\cite{barnes2020} and quantum parameter estimation~\cite{farokhi2025} to differential privacy.

The rest of the article is organized as follows. In the remainder of this section, we outline our main results and define the criteria that characterize the quality of a differentially private quantum sensing protocol. Then, in~\cref{sec:prelims}, we introduce the standard quantum sensing protocol that we follow, motivate why we need privacy, review the relevant theory from classical and quantum differential privacy needed to develop our protocols, and connect differential privacy and sensing in an intuitive way. We then introduce the centralized and decentralized network settings and the corresponding differentially private quantum sensing protocols in~\cref{sec:centralized-network,sec:decentralized-network}, respectively, and show their correctness and soundness. Finally, we offer some concluding remarks in~\cref{sec:discussion}. We relegate additional analysis of our main protocol to~\cref{app-sec:ghz-state-leakiness,app-sec:noisy-hamiltonian-analysis}.

\subsection{Main result}\label{subsec:main-result}
Our main result is the introduction of several entangled sensor network protocols that achieve different levels of success across several measures; see~\cref{tab:summary-of-protocols}. We consider two different models for the structure of the sensor networks and introduce differentially private quantum sensing protocols for both. While the scaling of the mean-squared error $\varepsilon_{\mathrm{MSE}}$ of the function estimate as a function of the number of nodes in the network is the single figure of merit in the original quantum sensing protocol~\cite{eldredge2018}, we define in~\cref{def:optimality-conditions} a set of \emph{three} criteria that characterize the optimality of a differentially private quantum sensing protocol. Note that we introduce a few variables and terms (e.g., $\varepsilon$, $\delta$) that we define more precisely later on.

Let $q(\boldsymbol{\theta}) = \frac{1}{n}\defsum{i=1}{n}{\theta_i}$ be the average of the parameters $\boldsymbol{\theta} \in \RR^n$ and let $\varepsilon \geq 0$ denote the privacy level and $0 \leq \delta \leq 1$ denote the probability of a mechanism \emph{failing} to achieve an $\varepsilon$-level of privacy; a smaller $\varepsilon$ corresponds to higher privacy and $\delta$ can be arbitrarily small. Furthermore, denote by $Q$ an estimate of $q(\boldsymbol{\theta})$ and define the mean-squared error as $\varepsilon_{\mathrm{MSE}} = \mathbb{E}[(Q - q(\boldsymbol{\theta}))^2]$. Then, we have the following:

\begin{definition}[Optimal differentially private quantum sensing conditions]\label{def:optimality-conditions}
    The \emph{optimality} of an $(\varepsilon, \delta)$-differentially private quantum sensing protocol $\mathcal{M}$ depends on its \emph{correctness} and its \emph{soundness}. Correctness quantifies how, in the honest setting, the mean-squared error $\varepsilon_{\mathrm{MSE}}$ scales. If $\varepsilon_{\mathrm{MSE}} = \bigOh{n^{-\alpha}}$, then
    \begin{itemize}
        \item $\alpha = 1$ corresponds to the standard quantum limit,
        \item $\alpha = 2$ corresponds to the Heisenberg limit, and
        \item $1 < \alpha < 2$ is an intermediate regime.
    \end{itemize}
    The soundness is determined by two criteria:
    \begin{itemize}
        \item \textbf{Locality}: Does the protocol require a trusted, central party or can it be implemented locally by each node in the network?
        \item \textbf{Privacy}: How well does the protocol protect individual nodes' data, measured by the privacy parameters $\varepsilon$ and $\delta$?
    \end{itemize}
\end{definition}

Here, "honest setting" refers to the setting in which there are no adversarial nodes, and so the sensing protocol $\mathcal{M}$ is run, with added noise as prescribed. If, however, there \emph{are} adversarial nodes that wish to perform some sort of privacy-breaching attack, then the conditions in the soundness criteria become relevant: namely, how much privacy is afforded. An optimal differentially private quantum sensing protocol retains Heisenberg scaling of $\varepsilon_{\mathrm{MSE}}$ in the honest setting, is $\varepsilon$- or $(\varepsilon, \delta)$-differentially private for $\varepsilon = \bigOh{1}$ and small $\delta$, and uses a local mechanism.

Using the criteria in~\cref{def:optimality-conditions} as a guideline, we categorize our differentially private quantum sensor network protocols according to the network model assumed. Each protocol has its own set of tradeoffs that are chosen according to the user's preferences with respect to accuracy and privacy. Here, \emph{centralized} and \emph{decentralized} refer to the structure of the network, which effectively comes down to the existence or lack of a trusted central party that administers the network (i.e., creates and distributes entanglement, collects the measurement results, calculates the function estimate, and applies noise). As such, this affects how the noise is added: a centralized network with a trusted central party applies noise \emph{globally}, giving rise to a global differentially private mechanism, while a decentralized network does not have a central party and so each node in the network must add noise \emph{locally}, giving rise to a local differentially private mechanism.

\begin{enumerate}
    \item \emph{Centralized network}: In this setting, we assume that there exists a trusted central party called a \emph{curator} that administers the network. Here, we allow for both "internal" adversaries (i.e., adversarial nodes) and "external" adversaries, such as a malicious third party delegating a sensing task to the network. To protect a given node in this setting, we use a classical \emph{global Laplace mechanism}, which we show has a mean-squared error that scales according to the Heisenberg limit and is differentially private with privacy budget $\varepsilon = \bigOh{1}$ and $\delta = 0$. However, because it requires a trusted curator, it cannot be implemented locally by each node and thus fails the locality condition.
    
    \item \emph{Decentralized network}: If a trusted central party does not control the network, then the sensing task must be performed among the nodes themselves. To protect a given node in this setting, we introduce two protocols:
    \begin{enumerate}
        \item The \emph{noisy Hamiltonian protocol} admits a privacy-utility tradeoff between the scaling of $\varepsilon_{\mathrm{MSE}}$ as a function of $n$ and the privacy budget $(\varepsilon, \delta)$. Namely, we find $\varepsilon_{\mathrm{MSE}} = \bigOh{1/n^2} + \bigOh{1/n^\alpha}$ and $\varepsilon = \Theta(n^{(\alpha-1)/2})$, where $1 \leq \alpha \leq 2$ can be chosen such that
        \begin{enumerate}
            \item For $\alpha=1$, $\varepsilon_{\mathrm{MSE}} = \bigOh{1/n}$ scales according to the standard quantum limit but $\varepsilon = \Theta(1)$ is a constant,
            \item For $\alpha=2$, $\varepsilon_{\mathrm{MSE}} = \bigOh{1/n^2}$ scales according to the Heisenberg limit but $\varepsilon = \Theta(\sqrt{n})$ grows with the number of nodes, and
            \item For $1 < \alpha < 2$, we get intermediate behavior for both $\varepsilon_{\mathrm{MSE}}$ and $\varepsilon$.
        \end{enumerate}
        We also find conditions under which $\delta = 0$. Note that the locality condition is satisfied by construction, as the protocol does not require a trusted curator.

        \item The noisy Hamiltonian protocol with an \emph{honest fraction} is an $(\varepsilon, \delta)$-differentially private, locally implementable protocol that retains Heisenberg scaling in $\varepsilon_{\mathrm{MSE}}$ under certain conditions, but requires additional assumptions and uses a different noise model (the Gaussian mechanism, rather than the Laplace mechanism).
    \end{enumerate}
\end{enumerate}

The basic noisy Hamiltonian protocol is analyzed in the worst-case adversarial model, where we have a single honest node that is trying to protect its parameter against all other nodes. This yields the strongest privacy model but forces a tradeoff between privacy and Heisenberg scaling. In contrast, the honest-fraction noisy Hamiltonian protocol relaxes this model by assuming a common source of randomness and a constant fraction of honest nodes. Under these additional assumptions and a restriction to average-type queries on sufficiently large subsets, the protocol achieves Heisenberg scaling with privacy budget $\varepsilon = \bigTheta{1}$ and an arbitrarily small $\delta$, which is considered fully optimal according to our definition in~\cref{def:optimality-conditions}. We believe our work to be a novel unification of quantum sensing and differential privacy, and we hope that it will encourage further development of secure sensing protocols.

\section{Preliminaries}\label{sec:prelims}
In this section, we outline the physical setup of the entangled sensor network in~\cref{subsec:sensor-network} and briefly motivate its vulnerability. To address this vulnerability, we use tools from differential privacy, so we give an overview of the theories of classical and quantum differential privacy in~\cref{subsec:classical-dp,subsec:quantum-dp}, respectively. We give an interpretation of differential privacy in the context of quantum sensing in~\cref{subsec:dp-in-sensing}.

\begin{figure}[t]
    \centering
    \includegraphics[scale=0.4]{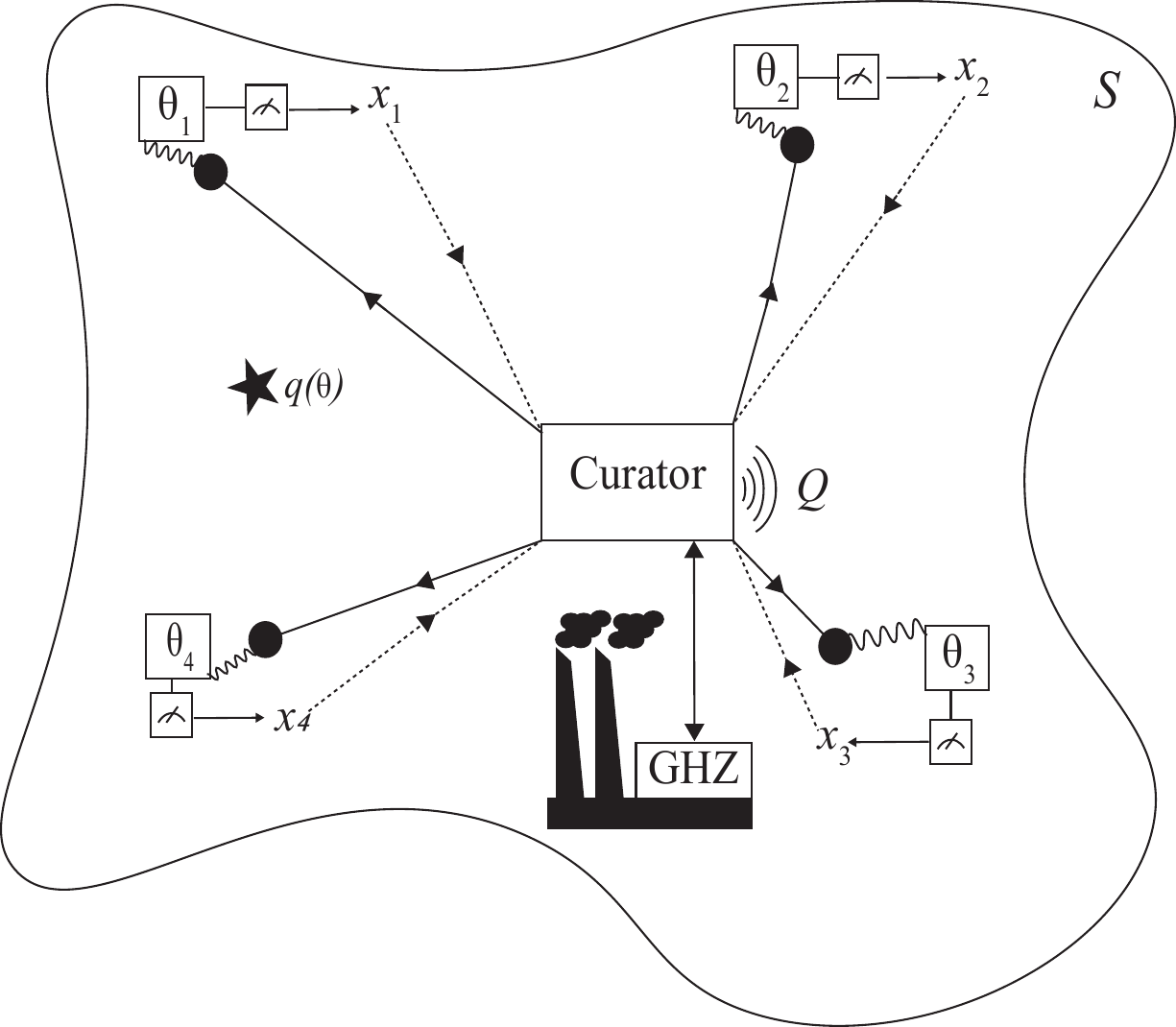}
    \caption{Schematic of the physical setup of the entangled sensor network. The nodes are embedded in some space $S$ and the goal is to estimate a function $q(\boldsymbol{\theta})$, labeled with a star. Here, we show the "centralized setting," in which a central curator administers the sensing protocol. This involves creating an $n$-qubit GHZ state and distributing the state to the network. Then, each qubit is coupled to the local parameter $\theta_i$ at each node, evolved according to the Hamiltonian defining the problem, and measured. The measurement results are then sent to the curator, from which the curator calculates an estimate $Q$ for the function $q$ and broadcasts the answer to the client (either a third party or one of the nodes themselves), which may be adversarial.}
    \label{fig:physical-picture}
\end{figure}

\subsection{Entangled sensor network}\label{subsec:sensor-network}
We first describe the entangled quantum sensor network that forms the basis of our differentially private sensing protocols, as sketched in~\cref{fig:physical-picture}; for an introduction to quantum sensing in general, we refer the reader to Ref.~\cite{degen2017}. Let $n$ denote the number of quantum sensors in the network, which we take to be qubits that evolve unitarily under a Hamiltonian
\begin{align}
    H(t) = H_c(t) + \frac{1}{2}\defsum{i=1}{n}{\theta_i\sigma_i^z},
\end{align}
where $H_c(t)$ is a time-dependent control Hamiltonian that is independent of the parameters $\boldsymbol{\theta}$, $\theta_i \in \RR$ is a parameter characterizing a physical quantity (e.g., a local magnetic field), and $\sigma_i^z$ is a Pauli-$Z$ operator acting on the $i\numth$ qubit sensor. Our goal is to estimate the value of a linear function $q(\boldsymbol{\theta})$. We do this by making measurements of a final quantum state $\ket{\psi_f}$, which we get from unitary evolution of an initial state $\ket{\psi_0}$. The unitary operator is informed by the parameters $\boldsymbol{\theta}$, namely
\begin{align}
    U(t,\boldsymbol{\theta}) = \mathcal{T}\exp{-i\defint{0}{t}{d\tau H(\tau)}},
\end{align}
where $\mathcal{T}$ is the time-ordering operator. Thus, the final state is given as
\begin{align}
    \ket{\psi_f} = U(t,\boldsymbol{\theta})\ket{\psi_0}.
\end{align}
The measurements are specified by a positive operator-valued measure (POVM):
\begin{align}
    \{\Pi_\xi\}\text{ such that } \int{\mathop{d\xi}\Pi_\xi} = 1.
\end{align}
We run the experiment $M \geq 1$ times and get an estimate $Q$ from the measurement results. Thus, a \emph{protocol} is the specification of the initial quantum state $\ket{\psi_0}$, the control Hamiltonian $H_c(t)$, the POVM $\{\Pi_\xi\}$ to estimate $q(\boldsymbol{\theta})$, and the estimator $Q$. An \emph{optimal protocol} minimizes the mean-squared error (MSE) $\varepsilon_{\mathrm{MSE}}$ on the function value given a fixed time $t$:
\begin{align}
    \varepsilon_{\mathrm{MSE}} &= \evsymb{(Q - q(\boldsymbol{\theta}))^2}\\
    &= \Var{Q} + \lp \evsymb{Q} - q(\boldsymbol{\theta}) \rp^2,\label{eqn:mse-explicit}
\end{align}
where the first term in~\cref{eqn:mse-explicit} is the variance of the estimator $Q$ and the second term is the bias. If $Q$ is an unbiased estimator, then $\mathbb{E}[Q] = q(\boldsymbol{\theta})$ and the second term becomes 0. The specific protocol that we follow in this paper is that proposed in Ref.~\cite{eldredge2018}. The protocol takes the initial quantum state to be a Greenberger-Horne-Zeilinger (GHZ) state, which is either produced by a trusted central party or is created among the nodes in the network according to a pre-specified protocol. Here, each node in the network receives one qubit of the GHZ state, and so each node is effectively a single qubit sensor; we mostly refer to the "nodes" of the network, but we may also use the terms "sensor" and "qubit" to refer to the nodes throughout the paper and thus treat all three terms interchangeably. At each node in the network, the qubit is coupled to a parameter $\theta_i$ and then evolved under the Hamiltonian defining the protocol. It is during this time that the function $q(\boldsymbol{\theta})$ is imprinted on the phase of the quantum state. In this work, we only consider the average function, which we define as
\begin{align}
    q(\boldsymbol{\theta}) = \frac{1}{n}\defsum{i=1}{n}{\theta_i}.
\end{align}
Here, the imprinting of $q(\boldsymbol{\theta})$ onto the phase happens directly as a result of the state evolution without the need for ancilla qubits or a control Hamiltonian. After this evolution, a parity measurement, which we define below, is made across all qubits. Repeating the protocol $M$ times gives us statistics that we then use to estimate the function of interest.

We now show in more detail how such a sensor network~\cite{eldredge2018,belliardo2020,ehrenberg2023} can be used to estimate $q(\boldsymbol{\theta})$; see~\cref{alg:standard-entangled-sensing-protocol}. As mentioned above, we take our initial state to be the $n$-qubit GHZ state:
\begin{align}
    \ket{\psi_0} = \frac{1}{\sqrt{2}}\lp \ket{0}^{\otimes n} + \ket{1}^{\otimes n} \rp.
\end{align}
Since we have the same coefficient (i.e., $1/n$) multiplying each parameter, we do not need a control Hamiltonian or any ancilla qubits to dial in different coefficients, as would be needed for a general function with coefficients that are not all equal. Thus, the Hamiltonian describing our problem is
\begin{align}
    H = \defsum{i=1}{n}{\frac{1}{2}\theta_i\sigma_i^z},
\end{align}
where we have written $H(t)$ as $H$ because the Hamiltonian is now time-independent. We assume that we evolve the GHZ state for a time $t$, which yields the following unitary operator describing the evolution:
\begin{align}
    U(t) = e^{-\frac{it}{2}\defsum{i=1}{n}{\theta_i\sigma_i^z}}.
\end{align}
Applying $U(t)$ to $\ket{\psi_0}$, we find the final state to be
\begin{align}
    \ket{\psi_f} &= U(t)\ket{\psi_0}\\
    &= \frac{1}{\sqrt{2}}\lp e^{-\frac{int}{2}q}\ket{0}^{\otimes n} + e^{\frac{int}{2}q}\ket{1}^{\otimes n} \rp,
\end{align}
where $q \equiv q(\boldsymbol{\theta})$. To extract an estimate for $q(\boldsymbol{\theta})$, we make a parity measurement on the final state, where the parity operator is $P = \bigotimes_{i=1}^{n}{\sigma_i^x}$. We calculate the expectation value, $\expval{P}$, as
\begin{align}
    \expval{P} &= \mel{\psi_f}{P}{\psi_f} = \cos(ntq).\label{eqn:exp-val-parity-measurement}
\end{align}
Then, to recover a function estimate $Q$ for the true function value $q(\boldsymbol{\theta})$, we find an estimate of the expectation value for $P$ using the measurement results and invert~\cref{eqn:exp-val-parity-measurement}, to get
\begin{align}
    Q = \frac{1}{nt}\arccos{\lp \expval{P}_{\text{est}} \rp},\label{eqn:solve-for-estimate}
\end{align}
where
\begin{align}
    \expval{P}_{\text{est}} = \frac{1}{M}\defsum{m=1}{M}{p_m}.
\end{align}
Here,
\begin{equation}
    p_m = \prod_{i=1}^{n}{x_i^{(m)}} \in \{-1,1\}
\end{equation}
is the measurement outcome of the parity operator $P$, where each $x_i^{(m)}$ is the $\sigma^x$ measurement result of the $i\numth$ qubit in round $m$. As explained in~\cite{eldredge2018,wineland1994}, the variance $\Var{Q}$ of the function estimate $Q$ is given by
\begin{align}
    \Var{Q} &= \frac{\Var{P}}{\lp \pdv{\expval{P}}{q} \rp^2}.
\end{align}
We calculate $\Var{P}$ as
\begin{align}
    \Var{P} &= \expval{P^2} - \expval{P}^2\\
    &= \sin^2(ntq).
\end{align}
Likewise, we have
\begin{align}
    \lp \pdv{\expval{P}}{q} \rp^2 &= \lp \pdv{}{q}\lp \cos(ntq) \rp \rp^2\\
    &= n^2t^2\sin^2(ntq).
\end{align}
Putting everything together, we have
\begin{align}
    \Var{Q} = \frac{1}{n^2t^2}.\label{eqn:standard-heisenberg-scaling}
\end{align}
Note that this is for a single-shot protocol; if we do this $M$ times, then we pick up a factor of $M$ in the denominator. An unentangled strategy (e.g., taking the initial state to be a product state), where each qubit sensor effectively acts independently of the others, yields $\Var{Q} \sim \frac{1}{nt^2}$~\cite{eldredge2018} for a single shot, so we see an improvement by a factor of $1/n$ with entanglement. This is the Heisenberg scaling advantage that forms one of our three optimality conditions in~\cref{def:optimality-conditions}.

\begin{figure}[t]
    \begin{algorithm}[H]
        \small
        \caption{Standard entangled sensing protocol}\label{alg:standard-entangled-sensing-protocol}
        \begin{algorithmic}[1]
            \Require $H = \defsum{i=1}{n}{\frac{1}{2}\theta_i\sigma^z_i}$, $\boldsymbol{\theta} = \lp \theta_1, \ldots, \theta_n \rp$, $t$, $M$
            \Ensure $Q$
            \For{$m \in [M]$\footnote{We introduce the notation $[M] = \{1, 2, \ldots, M\}$ as a convenient shorthand.}}
                \State Couple each sensor to local parameter $\theta_i$
                \State Evolve $\ket{\psi_0}$ under $U(t)$ to get $\ket{\psi_f}$
                \State Make parity measurement $P$ of $\ket{\psi_f}$
                \State $p_m \gets$ Parity measurement result
            \EndFor
            \State Estimate parity measurement expectation from measurement outcomes: $\expval{P}_{\text{est}} \gets \frac{1}{M}\defsum{m=1}{M}{p_m}$
            \State $Q \gets \frac{1}{nt}\arccos{\lp \expval{P}_{\text{est}} \rp}$
            \State \Return $Q$
        \end{algorithmic}
    \end{algorithm}
\end{figure}

The above protocol comes with an ambiguity as to which $\pi$-interval the phase lies in. Thus, with this protocol, there is an assumption of an estimate on the function that puts the function in a specific $\pi$ interval, where typically we use the entangled sensor network protocol to get a final bit of precision to pinpoint exactly where in this interval the phase lies. However, if we do not have access to such an estimate or we need more than just a final bit of precision, we need to make use of the full \emph{robust phase estimation} or \emph{bit-by-bit learning} protocol, which is described in Ref.~\cite[Appendix C]{ehrenberg2023} and in greater detail in Ref.~\cite{belliardo2020}. This allows us to estimate the function to $K$ bits of precision without requiring strong prior knowledge of the value of the function. We now describe this protocol.

We assume that, after unitary evolution, the state picks up a phase proportional to the function $q(\boldsymbol{\theta})$, namely
\begin{align}
    \ket{\psi_f} = \frac{1}{\sqrt{2}}\lp \ket{0} + e^{inqt}\ket{1} \rp \otimes \ket{0 \ldots 0},
\end{align}
where the phase can be pushed to the first qubit using a set of unentangling gates and where we are ignoring an overall phase. We start by dividing the total experiment time $t$ into $K$ stages, where $K$ is the number of bits of precision we measure the function to, and we divide each stage into $2\nu_j$ equal pieces. Then, we evolve the system in the $j\numth$ stage for time $t_j = M_j \delta t$, where $\delta t$ is some constant unit of time, $M_j = 2^{j-1}$, and where $j \in [K]$. As the sensing time for each stage increases exponentially, we can see that most of the total experiment time is spent estimating the \emph{least} significant bit of precision, as one would expect when trying to resolve the finer details of the accumulated phase. We assume that we have some $(n,t)$-independent prior knowledge of the function $q(\boldsymbol{\theta})$, in particular the domain of the parameters $\theta_i \in [\theta_{\min}, \theta_{\max}]$ for all $i \in [n]$, so we choose $\delta t$ such that it satisfies
\begin{align}
    n\delta t(\theta_{\max} - \theta_{\min}) \in [0, 2\pi).
\end{align}
In the $j\numth$ stage, we evolve the system as described above for a time $M_j \delta t$ and repeat, obtaining $2\nu_j$ independent copies of the state
\begin{align}
    \ket{\psi_j} = \frac{1}{\sqrt{2}}\lp \ket{0} + e^{inqM_j\delta t}\ket{1} \rp \otimes \ket{0 \ldots 0},\label{eqn:post-evolution-state}
\end{align}
where $\nu_j$ decreases linearly with $j$ and is taken to be $\nu_j = \lfloor x_j \rceil$, where
\begin{align}
    x_j = \frac{3}{\log_2{C}}(K - j) + x_K
\end{align}
for constants $C$ and $x_K$ and where $\lfloor \cdot \rceil$ denotes rounding to the nearest integer. $\nu_j$ has this specific form to minimize the total time $t$ for a fixed desired precision (i.e., number of bits $K$), as outlined in Ref.~\cite{belliardo2020}. The total time of this $K$-stage protocol is given by
\begin{align}
    t = 2\defsum{j=1}{K}{\nu_j M_j \delta t}.
\end{align}
With the post-evolution states of~\cref{eqn:post-evolution-state}, we make $2\nu_j$ single-qubit measurements, whose outcomes allow us to estimate $q(\boldsymbol{\theta})$ bit-by-bit. Specifically, we make two measurements, each $\nu_j$ times for each stage: a $\sigma^x$ measurement and a $\sigma^y$ measurement. Each measurement has an outcome of either $+1$ or $-1$, which we map to 0 and 1, respectively, where the outcome probabilities are given by
\begin{align}
    p^{(x)}(0) &= \frac{1 + \cos(M_j nq \delta t)}{2},\\
    p^{(x)}(1) &= 1 - p^{(x)}(0),\\
    p^{(y)}(0) &= \frac{1 + \sin(M_j nq \delta t)}{2},\\
    p^{(y)}(1) &= 1 - p^{(y)}(0).
\end{align}
The two sets of measurements allow us to resolve the two-fold degeneracy in the phase $nqM_j\delta t$ in a given $[0,2\pi)$ interval that arises from measurement along only one of the axes. As we increase $\nu_j$, the observed probabilities $f_0^{(x)}$ and $f_0^{(y)}$ for obtaining 0 for $\sigma^x$ and $\sigma^y$, respectively, converge to their expectation values. In particular, at each stage we get an estimator $\Tilde{\phi}$ for $\phi \coloneqq M_jnq\delta t$ as
\begin{align}
    \Tilde{\phi} = \arctan2\lp 2f_0^{(y)} - 1, 2f_0^{(x)} - 1 \rp \in [0, 2\pi),
\end{align}
where $\arctan2$ is the two-argument arc-tangent with range $[0, 2\pi)$. With the other variables known and the estimate $\tilde{\phi}$, we have enough information to estimate $q(\boldsymbol{\theta})$, thus concluding the sensing experiment. We use this full bit-by-bit learning protocol in the numerical analysis of the privacy of the noisy Hamiltonian protocol in~\cref{app-subsec:hockey-stick-divergence}.

As we discuss in more detail below, these entangled sensor networks are vulnerable to \emph{differencing attacks}, mentioned briefly in~\cref{sec:intro}. Here, an adversarial party can run the sensing protocol in such a way that they can learn the parameter or parameters of one or more nodes in the network, compromising the privacy of the network. As such, we need a mechanism to protect against such attacks, which is where differential privacy comes into play.

\subsection{Classical differential privacy}\label{subsec:classical-dp}
With the sensing protocol introduced and a motivation for why we should consider differentially private mechanisms, we now offer a brief introduction to classical differential privacy. Differential privacy was formally introduced in its modern form most notably by Dwork and others in 2006~\cite{dwork2006a,dwork2006b}. It is motivated primarily in the context of honestly published databases compromising privacy, where the output of database queries alone can violate privacy. Techniques like secure multiparty computation generally do not address this type of privacy violation, but they can be combined with differential privacy to give stronger privacy guarantees than either alone~\cite{shi2011privacy-preserving}. Differential privacy has been used in several industries for various applications. Examples include protecting user privacy for statistical analyses on Google searches during the COVID-19 pandemic~\cite{aktay2020google} and protecting the data of power companies' customers~\cite{finster2015smart,pare2020applying}. For a general introduction to differential privacy, we refer the reader to the work by Dwork \etal~\cite{dwork2014}, and for an introduction with a focus on the complexity of differentially private algorithms, we suggest the work by Vadhan~\cite{vadhan2017}. We also highlight recent NIST guidelines~\cite{near2025} for evaluating differential privacy guarantees, suggesting the maturity and promise of differential privacy as a practical tool.

Intuitively, a good differentially private mechanism provides both privacy for individuals in a group and good utility. Stated another way, a differentially private algorithm is one in which a potentially adversarial observer cannot tell whether an individual's data was used in a computation based on the output of the algorithm. Differential privacy provides a way to hide the data of specific individuals while also keeping the integrity of the conclusions we can draw from the group. In differential privacy, we often work with \emph{datasets}:

\begin{definition}[Dataset]\label{def:dataset}
    A \emph{dataset} $D$ is a collection of (potentially sensitive) information about one or more individuals. If $D$ contains the data of $n \geq 1$ individuals, we take the $i\numth$ row to contain the attributes of the $i\numth$ individual.
\end{definition}

As an intuitive example of a dataset that demonstrates the need for a differentially private mechanism, consider a hospital that stores the medical records of $n$ patients, where each row contains the attributes of each patient such as name, age, sex, address, medical record number, and status for a certain disease (e.g., diabetes, hypertension, etc.). If a research team at the hospital or at an outside academic institution wishes to use this dataset for a study, the patients included in the dataset would likely want their data to remain private and anonymous while the research team would like to access as much of the information as possible to be able to draw meaningful conclusions about their research hypothesis. Thus, we must have a mechanism that allows the hospital or manager of the dataset to release as much information as possible while still retaining sufficient privacy for all patients included in the dataset. The tools from differential privacy allow us to do exactly this.

A differentially private mechanism should protect against certain classes of adversarial attacks. One such attack that we have already mentioned is the differencing attack, where an adversary can make two or more queries and, using the two outputs, learn the data of an individual in the dataset. Another type of attack that differential privacy is designed to protect against is a \emph{linkage attack}, where an adversary can combine anonymized data from multiple sources to piece together the identity of an individual in a dataset. For example, Latanya Sweeney conducted a famous linkage attack in 1997 to link the then-governor of Massachusetts with his medical records using only publicly available information~\cite{sweeney1997}. Thus, even though direct identifiers (e.g., an individual's name) may be removed from a dataset, a clever combination of data from several sources is sufficient to reveal an individual's identity in a dataset. By adding controlled, randomized noise to queries, differential privacy makes it difficult to perform such attacks.

With this in mind, we formally define the relevant concepts from classical differential privacy. We start by defining classical global and local differentially private mechanisms, which will be relevant depending on which network setting we are considering later on in the sensing context. We start with global differential privacy, where the adversary can compute functions on datasets $D$ of all individuals~\cite{vadhan2017}. We call the set of possible functions we can compute on the dataset the \emph{query space}. In the definition, we include a subscript "$c$" to denote "classical" variables, where later, no subscript indicates "quantum" variables, and we include a superscript "$g$" to indicate that this is a global mechanism.

\begin{definition}[Classical global differential privacy]\label{def:cgdp}
    Let $\mathcal{X}^n$ be the set of all possible datasets of $n$ individuals and let $\varepsilon \geq 0$. A trusted curator holds a dataset $D \in \mathcal{X}^n$. Access to the data is provided via a randomized (i.e., probabilistic) mechanism $\mathcal{M}^g_c : \mathcal{X}^n \times \mathcal{Q}^g_c \to \mathcal{Y}^g_c$, where $\mathcal{Q}^g_c$ is the query space and $\mathcal{Y}^g_c$ is the output space of $\mathcal{M}^g_c$. Denote a query as $q^g_c \in \mathcal{Q}^g_c$. Then, $\mathcal{M}^g_c$ is \emph{$\varepsilon$-globally differentially private} if, for every pair of datasets $D, D' \in \mathcal{X}^n$ that differ by one row and for all $A \subseteq \mathcal{Y}^g_c$, we have
    \begin{align}
        \Pr[\mathcal{M}^g_c(D, q^g_c) \in A] \leq e^\varepsilon \cdot \Pr[\mathcal{M}^g_c(D', q^g_c) \in A].
    \end{align}
\end{definition}
        
We interpret this as saying that a global randomized mechanism is differentially private if, given access to two datasets that are close (i.e., they differ by one row), the output of the mechanism, and a desire for high privacy (i.e., a small value for $\varepsilon$), an adversary will have a small probability of being able to tell which dataset was used as the input.

While global differential privacy is usually what is meant by "differential privacy" in the classical literature, local differential privacy~\cite{bebensee2019} has in recent years become quite popular, with companies like Apple~\cite{apple2017} and Google~\cite{erlingsson2014} developing and deploying their own local differential privacy mechanisms~\cite{vadhan2017}. Note that the superscript "$\ell$" indicates that this is a local mechanism:

\begin{definition}[Classical local differential privacy]\label{def:cldp}
    Let $\mathcal{X}$ be the set of (potentially sensitive) possible individuals' data points, $\varepsilon \geq 0$, and $\mathcal{M}^\ell_c : \mathcal{X} \times \mathcal{Q}^\ell_c \to \mathcal{Y}^\ell_c$ be a randomized mechanism, where $\mathcal{Q}^\ell_c$ is the query space and $\mathcal{Y}^\ell_c$ is the output space of $\mathcal{M}^\ell_c$. Denote a query as $q_c^\ell \in \mathcal{Q}_c^\ell$. Then, $\mathcal{M}^\ell_c$ is $\varepsilon$-\emph{locally differentially private} if, for any pair of neighboring data points $x, x' \in \mathcal{X}$ and for all $A \subseteq \mathcal{Y}^\ell_c$, we have
    \begin{align}
        \Pr[\mathcal{M}^\ell_c(x, q_c^\ell) \in A] \leq e^\varepsilon \cdot \Pr[\mathcal{M}^\ell_c(x', q_c^\ell) \in A].
    \end{align}
\end{definition}

That is, given access to the output of a locally differentially private mechanism, there is a small probability of being able to determine if the input to the mechanism was $x$ or $x'$, which are neighboring data points. Thus, if an adversary has access to the response of an individual, then the adversary will be unable to precisely learn the individual's personal data. What is meant by "neighboring" depends on the structure of the data. If the input is continuous, then "neighboring" refers to $x$ and $x'$ that are close in some distance measure. If the data is discrete, then "neighboring" refers to $x$ and $x'$ that differ by one entry. For example, if one attribute encoded in each individual's $x$ is hair color, then $x$ and $x'$ are neighboring if $x$ and $x'$ are equal in all attributes except for the value of the hair color (e.g., $x$ has brown hair and $x'$ has blond hair). We denote $x$ and $x'$ as neighboring with the notation $x \sim x'$.

Note that~\cref{def:cgdp,def:cldp} are very similar, with the main difference being that the local mechanism takes as input a single individual's data $x$ or $x'$ whereas the global mechanism takes as input \emph{all} individuals' data in the form of a dataset $D$ or $D'$. As we will see, notions of global differential privacy apply to the centralized network setting (\cref{sec:centralized-network}) while those of local differential privacy apply to the decentralized setting (\cref{sec:decentralized-network}).

While a differentially private mechanism provides privacy, it should still give accurate answers to our queries. In classical differential privacy, we often consider the \emph{additive error}, which is how much error the differentially private mechanism introduces. Here, we take the output space $\mathcal{Y}_c^\ell$ to be the set of real numbers: $\mathcal{Y}_c^\ell = \RR$.

\begin{definition}[Additive error with probability $1-\zeta$]\label{def:additive-error}
    Given a query $q_c^\ell \in \mathcal{Q}_c^\ell$ and a data point $x \in \mathcal{X}$, a randomized mechanism $\mathcal{M}_c^\ell : \mathcal{X} \times \mathcal{Q}_c^\ell \to \mathcal{Y}_c^\ell$ has \emph{additive error} at most $\alpha(x, q_c^\ell)$ with probability $1-\zeta$ if
    \begin{equation}
        \Pr[\abs{\mathcal{M}_c^\ell(x, q_c^\ell) - y} \leq \alpha(x, q_c^\ell)] \geq 1 - \zeta,
    \end{equation}
    where $y = q_c^\ell(x)$ and the probability is over the randomness of the mechanism $\mathcal{M}^\ell_c$.
\end{definition}

The definition for a global mechanism is similar. In classical differential privacy, there is often a tradeoff between the scaling of the additive error as a function of $n$ with the mechanism type, where global mechanisms can achieve an additive error scaling as $\bigOh{1/n}$ while local mechanisms can only achieve $\bigOh{1/\sqrt{n}}$. In the sensing setting, the analogous quantity to the additive error is the mean-squared error $\varepsilon_{\mathrm{MSE}}$, and we will see a similar tradeoff between the scaling of $\varepsilon_{\mathrm{MSE}}$ in the centralized versus decentralized network settings.

We now introduce several differentially private mechanisms. We start with the \emph{randomized response mechanism}. Randomized response relies on the notion of a \emph{counting query}:

\begin{definition}[Counting query]\label{def:counting-query}
    A \emph{counting query} $C$ is specified by a predicate on rows: $C : \mathcal{X} \to \{0, 1\}$, and can be extended to a dataset $D \in \mathcal{X}^n$ to give the fraction of rows that satisfy the predicate such that $C' : \mathcal{X}^n \to [0,1]$:
    \begin{align}
        C'(D) = \frac{1}{n}\defsum{x \in D}{}{C(x)}.
    \end{align}
\end{definition}

Intuitively, this is simply a query on a given dataset that counts the fraction of individuals that satisfy the criteria of our query (e.g., what fraction of people in our dataset have a disease?). Then, the randomized response mechanism~\cite{warner1965} is defined as follows:

\begin{definition}[Randomized response]\label{def:randomized-response}
    Let $\varepsilon \geq 0$ and $\mathcal{M}^{\text{RR}}_c : \mathcal{X} \times \mathcal{Q}_c^\ell \to \mathcal{Y}_c^\ell$ be a randomized mechanism. Here, we take $\mathcal{Q}_c^\ell : \mathcal{X} \to \{0,1\}$ to be a counting predicate, and $\mathcal{Y}_c^\ell \to \{0,1\}$. Then, for $C \in \mathcal{Q}^\ell_c$ and $x \in \mathcal{X}$, we have
    \begin{align}
        \mathcal{M}^{\text{RR}}_c(x) =
        \begin{cases}
            C(x) & \text{with probability } \frac{e^{\varepsilon}}{\left(e^{\varepsilon}+1\right)}\\
            \neg C(x) & \text{with probability } \frac{1}{\left(e^{\varepsilon}+1\right)},
        \end{cases}
    \end{align}
    where "$\neg$" denotes the bit flip of the output of $C$.
\end{definition}

Comparing the definition above with the definition of differential privacy in~\cref{def:cldp}, we can see that randomized response is $\varepsilon$-differentially private. Note that the randomized response mechanism as defined in~\cref{def:randomized-response} is a \emph{local} mechanism.

Another commonly used differential privacy mechanism is the \emph{Laplace mechanism}. In the following, we assume that we are in the global setting; we mention the differences for the local Laplace mechanism afterwards. First, we define the Laplace distribution:

\begin{definition}[Laplace distribution]\label{def:laplace-distribution}
    The \emph{Laplace distribution} (or double exponential distribution) is the distribution with probability density function
    \begin{align}
        f(x;\mu, b) = \frac{1}{2b}e^{-\frac{\abs{x-\mu}}{b}},
    \end{align}
    where $x \in \RR$, $\mu$ is the mean and $b > 0$ is a scale parameter. We denote the Laplace distribution by $\Lap{\mu,b}$ or, if we take $\mu = 0$, then by $\Lap{b}$. The variance is then $\sigma^2 = 2b^2$.
\end{definition}

We now define the global Laplace mechanism~\cite{vadhan2017}:

\begin{definition}[Global Laplace mechanism]\label{def:laplace-mechanism}
    For $\varepsilon > 0$, a query $q_c^g : \mathcal{X}^n \to \RR$, datasets $D, D' \in \mathcal{X}^n$, and global sensitivity given by
    \begin{align}
        \mathrm{GS} \coloneqq \max_{D \sim D'}{\abs{q_c^g(D) - q_c^g(D')}},
    \end{align}
    where $D \sim D'$ denotes that $D$ and $D'$ differ by one row, the \emph{Laplace mechanism} $\mathcal{M}^{\text{Lap}}_{c,\mathrm{GS}}$ over a domain $\mathcal{X}^n$ takes a dataset $D \in \mathcal{X}^n$ and outputs
    \begin{align}
        \mathcal{M}^{\text{Lap}}_{c,\mathrm{GS}}(D) = q_c^g(D) + \Lap{\mathrm{GS}/\varepsilon}.
    \end{align}
\end{definition}

From the definition, we see that the differential privacy mechanism outputs a query with some noise drawn from the Laplace distribution. We refer the reader to Ref.~\cite{vadhan2017} for a detailed analysis of the benefits of using the Laplace distribution for the noise source rather than the Gaussian distribution, though as we will see, in some cases the Gaussian distribution is preferred. Importantly, we have the following result:

\begin{theorem}[Privacy of Laplace mechanism]\label{thm:laplace-mechanism-properties}
    The Laplace mechanism is $\varepsilon$-differentially private.
\end{theorem}

See Ref.~\cite{vadhan2017} for a proof of this statement. In the local setting, the domain is $\mathcal{X}$, so the local sensitivity is given by
\begin{align}
    \text{LS} = \max_{x \sim x'}{\abs{q_c^\ell(x) - q_c^\ell(x')}}, 
\end{align}
where $x,x' \in \mathcal{X}$ and $x \sim x'$ means that $x$ and $x'$ are neighboring, as described above. Correspondingly, this gives the local Laplace mechanism, where noise is added locally to each individual's data, rather than to queries of the whole database. As we will see, the noisy Hamiltonian protocol that we introduce in~\cref{subsec:noisy-hamiltonian-protocol} is based on the local Laplace mechanism.

We also make use of the  \emph{Gaussian mechanism}~\cite{dwork2014} in a variant of the noisy Hamiltonian protocol in~\cref{subsec:nhp-honest-fraction}:

\begin{definition}[Gaussian mechanism]\label{def:gaussian-mechanism}
    Let $\varepsilon, \delta > 0$, and let $q_c^g : \mathcal{X}^n \to \RR$ be a query. The global sensitivity of $q_c^g$ is
    \begin{align}
        \mathrm{GS} \coloneqq \max_{D \sim D'}\abs{q_c^g(D)-q_c^g(D')},
    \end{align}
    where the maximum is over all neighboring datasets $D, D' \in \mathcal{X}^n$
    that differ in one row. The \emph{Gaussian mechanism} with global sensitivity $\mathrm{GS}$ is the randomized mechanism $\mathcal{M}_{c,\mathrm{GS}}^{\mathrm{Gauss}} : \mathcal{X}^n \to \RR$ defined by
    \begin{align}
        \mathcal{M}_{c,\mathrm{GS}}^{\mathrm{Gauss}}(D) = q_c^g(D) + Z,\quad Z \sim \mathcal{N}(0,\sigma^2),
    \end{align}
    where
    \begin{align}
        \sigma = \frac{\kappa\cdot\mathrm{GS}}{\varepsilon}
    \end{align}
    and $\kappa^2 > 2\ln(1.25/\delta)$.
\end{definition}

In contrast to the Laplace mechanism, the Gaussian mechanism achieves $(\varepsilon,\delta)$-differential privacy (we denote $(\varepsilon,0)$-differential privacy by $\varepsilon$-differential privacy, often called \emph{perfect} differential privacy). As such, $(\varepsilon, \delta)$-differential privacy is strictly weaker than $\varepsilon$-differential privacy when $\delta > 0$ and is often called \emph{approximate} differential privacy. In practice, we aim for small values of $\delta$, on the order $10^{-5}$. Informally, we think of $\delta$ as the probability of \emph{failing} to achieve $\varepsilon$-differential privacy. For the Gaussian mechanism, we have the following:

\begin{theorem}[Privacy of Gaussian mechanism]\label{thm:gaussian-mechanism-properties}
    The Gaussian mechanism is $(\varepsilon, \delta)$-differentially private.
\end{theorem}

This statement is proven in Ref.~\cite{dwork2014}. Finally, we highlight the composition theorem:

\begin{theorem}[Composition theorem]\label{thm:composition-thm}
    Let $\mathcal{M}_1^g, \ldots, \mathcal{M}_k^g$ be randomized mechanisms, where each $\mathcal{M}_j^g$ is $(\varepsilon_j, \delta_j)$-differentially private. Then, the mechanism that releases all outputs $\mathcal{M}^g(D) = (\mathcal{M}_1^g(D), \ldots, \mathcal{M}_k^g(D))$ is $(\varepsilon,\delta)$-differentially private, where
    \begin{equation}
        \varepsilon = \defsum{j=1}{k}{\varepsilon_j}\quad\text{and}\quad \delta = \defsum{j=1}{k}{\delta_j}.
    \end{equation}
\end{theorem}

We refer the reader to Refs.~\cite{dwork2006a,dwork2014} for proof and more discussion on the composition theorem. As is often the case, if each $\mathcal{M}_j^g$ is itself $(\varepsilon_0, \delta_0)$-differentially private, then the full mechanism is $(k\varepsilon_0, k\delta_0)$-differentially private.

\subsection{Quantum differential privacy}\label{subsec:quantum-dp}
We now provide some definitions corresponding to quantum differential privacy in analogy with the classical definitions discussed above. We also review some important results from the quantum differential privacy literature. We start with the global setting.

\begin{definition}[Global quantum mechanism]\label{def:global-quantum-mechanism}
    A \emph{global quantum mechanism} $\mathcal{M}^g : \mathcal{X}^n \to \mathcal{Y}^g$ for a set of classical datasets $\mathcal{X}^n$ in which each dataset contains $n$ rows (corresponding to data for $n$ individuals) and a set of output quantum states $\mathcal{Y}^g$ is a two-step process that
    \begin{enumerate}
        \item Encodes a classical dataset $D \in \mathcal{X}^n$ into a quantum state $\rho(D)$ and
        \item Applies a quantum channel (i.e., a completely-positive trace-preserving (CPTP) map) to $\rho(D)$ to output a quantum state.
    \end{enumerate}
\end{definition}

We could also think of steps 1 and 2 as a single step, though we delineate them in the definition for clarity. In the context of quantum sensing, a query is a run of the sensing protocol on a specific function of interest (e.g., the average of the full network). With this in mind, we define our notion of quantum global differential privacy~\cite{hirche2023}:
    
\begin{definition}[Quantum global differential privacy]\label{def:qgdp}
    Let $\mathcal{X}^n$ be a set of possible datasets and $\mathcal{M}^g$ be a global quantum mechanism. We consider $\mathcal{M}^g$ to be $(\varepsilon, \delta)$-differentially private if for any POVM $\{P_i\}$ such that $\defsum{i}{}{P_i} = \identity$ and all datasets $D,D'\in \mathcal{X}^n$ that differ by one row (i.e., $D \sim D'$), we have
    \begin{align}
        \label{eqn:qgdp-prob}
        \Tr[P\mathcal{M}^g(D)] \leq e^\varepsilon\Tr[P\mathcal{M}^g(D')] + \delta,
    \end{align}
    for POVM element $P$.
\end{definition}

We note two things about this definition. First, since we want to protect datasets $D$, we take $D, D'$ as input and the quantum state produced by the quantum mechanism $\mathcal{M}^g$ as output. Compare this to previous definitions in, for example~\cite{hirche2023,angrisani2023}, where the input is a quantum state $\rho$ and "quantum differential privacy" is based on neighboring states $\rho \sim \sigma$, according to an appropriate distance measure. Our definition can also be viewed in this way by first applying some encoding algorithm $\mathcal{A}_{\text{enc}}$ to encode $D$ into the state $\rho$, that is, $D \xmapsto{\mathcal{A}_{\text{enc}}} \rho(D)$. Then, by $\rho(D) \sim \rho(D')$, we mean that $D \sim D'$. Second, we interpret~\cref{eqn:qgdp-prob} as relating the probabilities that an adversary can distinguish between $D$ and $D'$ as the inputs to the quantum channel $\mathcal{M}^g$ by making a measurement $P$ on the output states. This is analogous to how we interpret the indistinguishability of the output values from two datasets that differ by one row in the classical case. 

We now move on to the local setting, where the space of operations is restricted.

\begin{definition}[Local quantum mechanism]\label{def:local-quantum-mechanism}
    A \emph{local quantum mechanism} $\mathcal{M}^\ell : \mathcal{X} \to \mathcal{Y}^\ell$ for a set of possible data points $\mathcal{X}$ and a set of output quantum states $\mathcal{Y}^\ell$ is a two-step process that
    \begin{enumerate}
        \item Encodes a classical data point $x \in \mathcal{X}$ into a quantum state $\rho(x)$ and
        \item Applies a local quantum channel to $\rho(x)$ to output a quantum state.
    \end{enumerate}
\end{definition}

This leads us to an analogous definition of local differential privacy:

\begin{definition}[Quantum local differential privacy]\label{def:qldp}
    Let $\mathcal{X}$ be the set of possible individual data points and $\mathcal{M}^\ell$ be a local quantum mechanism. $\mathcal{M}^\ell$ is $(\varepsilon, \delta)$-differentially private if for any POVM $\{P_i\}$ such that $\defsum{i}{}{P_i} = \identity$ and all possible data points $x,x'\in \mathcal{X}$ such that $x \sim x'$, we have
    \begin{align}
        \label{eqn:qldp-prob}
        \Tr[P\mathcal{M}^\ell(x)] \leq e^\varepsilon\Tr[P\mathcal{M}^\ell(x')] + \delta,
    \end{align}
    for POVM element $P$.
\end{definition}

We conclude this subsection with some useful results that we will use in our privacy proofs. We start by defining the \emph{quantum hockey-stick divergence}:

\begin{definition}[Quantum hockey-stick divergence]\label{def:quantum-hockey-stick-divergence}
    Given $\gamma \in \RR_+$ and quantum states $\rho$ and $\rho'$, the \emph{quantum hockey-stick divergence} is defined as
    \begin{align}
        E_{\gamma}(\rho\|\rho') \coloneqq \Tr[(\rho-\gamma\rho')_+],
    \end{align}
    where $\Tr[\sigma_+]$ is the sum of the positive eigenvalues of $\sigma$.
\end{definition}

The hockey-stick divergence is a generalization of statistical distance and is frequently used in both classical and quantum differential privacy; it is a special case of an $f$-divergence~\cite{renyi1961}. Using this definition, we have the following:

\begin{theorem}[Quantum differential privacy equivalence]\label{thm:hockey-stick-div}
    Let $\varepsilon > 0$, $\delta \geq 0$, $\gamma = e^\varepsilon$, and $\mathcal{M}$ be a quantum mechanism (either global or local). Then, the following two statements are equivalent:
        \begin{enumerate}
            \item $\mathcal{M}$ is $(\varepsilon, \delta)$-differentially private.
            \item $\delta \geq \sup_{\alpha \sim \alpha'}{E_\gamma(\mathcal{M}(\alpha) \| \mathcal{M}(\alpha'))},$
        \end{enumerate}
    where $E_\gamma$ is the quantum hockey-stick divergence defined in~\cref{def:quantum-hockey-stick-divergence} and where we take $\alpha = D$ in the global model and $\alpha = x$ in the local model.
\end{theorem}

We refer the reader to Ref.~\cite{hirche2023} for a proof of this statement. Finally, we will find the following post-processing theorem helpful in our proofs:

\begin{theorem}[Post-processing theorem]\label{thm:post-prosessing}
    Let $\mathcal{M}$ be a local quantum mechanism that is $(\varepsilon,\delta)$-locally differentially private. Let $\mathcal{N}$ be an arbitrary quantum channel. Then $\mathcal{N} \circ \mathcal{M}$ is also $(\varepsilon,\delta)$-locally differentially private.
\end{theorem}

We again refer the reader to Ref.~\cite{hirche2023} for a proof of this theorem.

\subsection{Interpretation of differential privacy in sensing}\label{subsec:dp-in-sensing}
In the protocols that we introduce below, we must interpret the implications of the resulting differential privacy in the context of sensing, so we introduce a general formalism here that can be applied to each protocol. Let $T$ denote the full \emph{transcript} released by a differentially private sensing protocol. Depending on the protocol, $T$ may contain a single noisy function estimate, $k$ distinct function estimates, classical measurement outcomes, or any other post-processed object released to the adversary. Then, for any two neighboring parameter vectors $\boldsymbol{\theta}, \boldsymbol{\theta}'$ that differ only on the target node $j$, and for any event $A$ in the transcript space (i.e., the output from the differentially private sensing mechanism), we have
\begin{equation}
    \Pr[T(\boldsymbol{\theta}) \in A] \leq e^{\varepsilon}\Pr[T(\boldsymbol{\theta}') \in A] + \delta,\label{eqn:dp-inequality}
\end{equation}
where $\varepsilon$ and $\delta$ are the privacy parameters for the full transcript. By symmetry, the same bound also holds when $\boldsymbol{\theta}$ and $\boldsymbol{\theta}'$ are swapped. This says that the adversary's entire view of the transcript cannot distinguish whether node $j$'s parameter was $\theta_j$ or $\theta_j'$ by more than a likelihood-ratio factor of $e^{\varepsilon}$, except with failure probability $\delta$. This includes every possible attack obtained from post-processing the results of running the protocol.

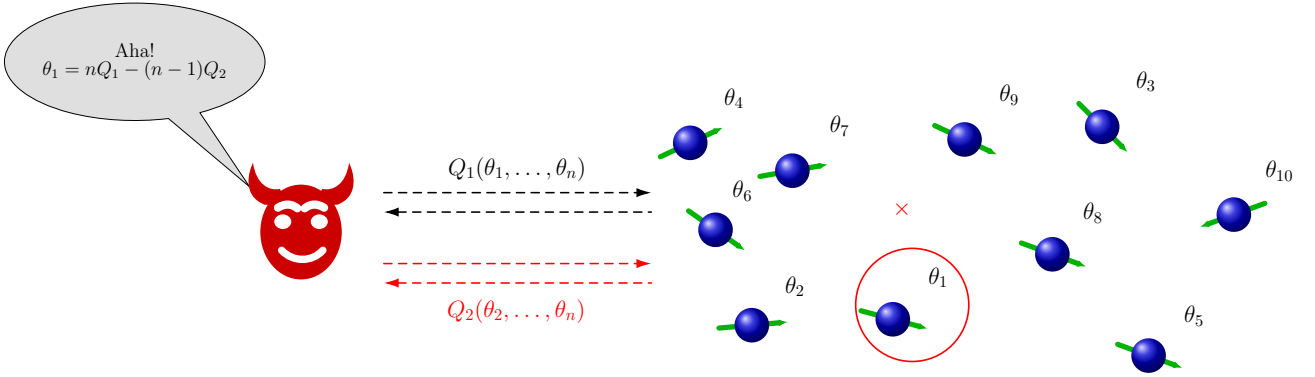
\begin{figure*}[tb]
    \centering
    \resizebox{\textwidth}{!}{%
    \begin{tikzpicture}[x=1.812cm,y=1.812cm,line cap=round,line join=round]
        \newcommand{\sensor}[3]{%
            \begin{scope}[shift={(#2,#3)},rotate=#1]
                \draw[green!70!black,line width=3.2pt,
                      -{Latex[length=2.0mm,width=1.5mm]}]
                    (-0.39,0) -- (0.43,0);
                \shade[ball color=blue!85!black] (0,0) circle (0.205);
            \end{scope}%
        }
    
        \path[use as bounding box] (-0.05,0) rectangle (16,5.565);
        \clip (-0.05,0) rectangle (16,5.565);
        
        \node[
            ellipse callout,
            callout absolute pointer={(3.06,2.83)},
            draw=black,
            fill=gray!25,
            line width=0.55pt,
            minimum width=3.44cm,
            minimum height=2.52cm,
            align=center,
            font=\large,
            inner sep=0pt
        ] at (1.70,4.31) {Aha!\\[-1mm]
        $\theta_1=nQ_1-(n-1)Q_2$};
        
        \begin{scope}[shift={(3.67,2.38)},scale=0.72]
        \fill[devilred]
          (-0.60,0.33)
            .. controls (-0.98,0.48) and (-0.86,0.88) .. (-0.79,1.02)
            .. controls (-0.74,0.66) and (-0.52,0.55) .. (-0.35,0.59)
            .. controls (-0.28,0.47) and (-0.20,0.37) .. (-0.12,0.31)
            -- cycle;
        \fill[devilred]
          (0.60,0.33)
            .. controls (0.98,0.48) and (0.86,0.88) .. (0.79,1.02)
            .. controls (0.74,0.66) and (0.52,0.55) .. (0.35,0.59)
            .. controls (0.28,0.47) and (0.20,0.37) .. (0.12,0.31)
            -- cycle;
        \fill[devilred] (0,-0.12) ellipse[x radius=0.67,y radius=0.78];
        \draw[white,line width=5.0pt]
            (-0.38,0.28) .. controls (-0.27,0.35) and (-0.16,0.35) .. (-0.06,0.27);
        \draw[white,line width=5.0pt]
            (0.38,0.28) .. controls (0.27,0.35) and (0.16,0.35) .. (0.06,0.27);
        \fill[white] (-0.30,0.05) ellipse[x radius=0.14,y radius=0.12];
        \fill[white] ( 0.30,0.05) ellipse[x radius=0.14,y radius=0.12];
        \draw[white,line width=4.5pt]
            (-0.32,-0.42) .. controls (-0.13,-0.62) and (0.18,-0.62) .. (0.36,-0.42);
        \end{scope}
        
        \draw[black,dashed,line width=0.95pt,dash pattern=on 5.0pt off 4.0pt,
              -{Latex[length=3.2mm,width=2.0mm]}]
            (4.65,2.76) -- (7.82,2.76);
        \draw[black,dashed,line width=0.95pt,dash pattern=on 5.0pt off 4.0pt,
              {Latex[length=3.2mm,width=2.0mm]}-]
            (4.65,2.53) -- (7.82,2.53);
        \node[font=\Large] at (6.23,3.04) {$Q_1(\theta_1,\ldots,\theta_n)$};
        
        \draw[red,dashed,line width=0.95pt,dash pattern=on 5.0pt off 4.0pt,
              -{Latex[length=3.2mm,width=2.0mm]}]
            (4.65,1.92) -- (7.82,1.92);
        \draw[red,dashed,line width=0.95pt,dash pattern=on 5.0pt off 4.0pt,
              {Latex[length=3.2mm,width=2.0mm]}-]
            (4.65,1.69) -- (7.82,1.69);
        \node[red,font=\Large] at (6.22,1.36) {$Q_2(\theta_2,\ldots,\theta_n)$};
        
        \sensor{25}{8.28}{3.35}
        \node[font=\Large] at (8.80,3.88) {$\theta_4$};
        
        \sensor{10}{9.49}{3.02}
        \node[font=\Large] at (10.05,3.55) {$\theta_7$};
        
        \sensor{-25}{11.53}{3.39}
        \node[font=\Large] at (12.06,3.92) {$\theta_9$};
        
        \sensor{-45}{13.16}{3.54}
        \node[font=\Large] at (13.68,4.08) {$\theta_3$};
        
        \sensor{200}{14.72}{2.50}
        \node[font=\Large] at (15.26,3.02) {$\theta_{10}$};
        
        \sensor{-20}{12.57}{2.03}
        \node[font=\Large] at (13.05,2.52) {$\theta_8$};
        
        \sensor{-20}{13.71}{0.83}
        \node[font=\Large] at (14.25,1.30) {$\theta_5$};
        
        \sensor{5}{9.01}{1.19}
        \node[font=\Large] at (9.52,1.65) {$\theta_2$};
        
        \sensor{-15}{10.68}{1.26}
        \draw[red,line width=1.0pt] (10.91,1.43) circle (0.67);
        \node[font=\Large] at (11.22,1.76) {$\theta_1$};
        
        \node[red,font=\Large\bfseries] at (10.79,2.57) {$\times$};
        
        \sensor{-35}{8.58}{2.32}
        \node[font=\Large] at (8.9,2.78) {$\theta_6$};
    \end{tikzpicture}%
    }
    \caption{Schematic of the differencing attack, the primary attack considered in this paper. The adversary (either a malicious internal node of the network collaborating with other malicious nodes or a malicious third party delegating a sensing task to the network, where the third party either acts independently or collaborates with malicious nodes) can ask for any linear function estimation task to be run on the network and receive back an estimate $Q$. If the adversary asks for the task $Q_1(\theta_1, \ldots, \theta_n) = \frac{1}{n}\defsum{i=1}{n}{\theta_i}$ including all of the nodes in the network and then $Q_2(\theta_2, \ldots, \theta_n) = \frac{1}{n-1}\defsum{i=2}{n}{\theta_i}$ excluding only $\theta_1$, then the adversary can effectively learn $\theta_1$ as $\theta_1 = nQ_1(\theta_1, \ldots, \theta_n) - (n-1)Q_2(\theta_2, \ldots, \theta_n)$.}
    \label{fig:adversarial-model}
\end{figure*}

We can interpret this statement in terms of the probability for an adversary to successfully distinguish between some $\theta_j$ and $\theta_j'$. In particular, the adversary is given the output transcript $T$ from the whole protocol and wants to distinguish between two hypotheses
\begin{equation}
    H_0 : \theta_j = a,\quad H_1 : \theta_j' = b,
\end{equation}
where $a,b \in [\theta_{\min}, \theta_{\max}]$. We assume that the two hypotheses are equally likely and we denote by $P$ the transcript distribution under $H_0$ and by $Q$ the distribution under $H_1$. If $A_0$ is the set of transcripts for which the adversary guesses $H_0$, then the success probability for the adversary to choose the correct hypothesis (and therefore distinguish $\theta_j$ from $\theta_j'$) is
\begin{align}
    p_{\mathrm{succ}} &= \frac{1}{2}P(A_0) + \frac{1}{2}Q(A_0^c)\\
    &= \frac{1}{2}(1 + P(A_0) - Q(A_0)).
\end{align}
Since we want to bound how well the adversary can successfully discriminate between the two neighboring parameters $\theta_j$ and $\theta_j'$, we can view this as bounding how large $P(A_0) - Q(A_0)$ can be. From~\cref{eqn:dp-inequality}, we find
\begin{equation}
    P(A_0) - Q(A_0) \leq \frac{e^\varepsilon-1+2\delta}{e^\varepsilon+1},
\end{equation}
and so
\begin{equation}
    p_{\mathrm{succ}} \leq \frac{e^\varepsilon+\delta}{1+e^\varepsilon}.
\end{equation}
Thus, differential privacy bounds the success probability of every adversarial decision rule, including differencing attacks and arbitrary post-processing of the transcript.

For pure differential privacy, that is, when $\delta=0$, this becomes
\begin{equation}
    p_{\mathrm{succ}} \leq \frac{e^\varepsilon}{1+e^\varepsilon}.
\end{equation}
For small $\varepsilon$ (i.e., high privacy), we have
\begin{equation}
    \frac{e^\varepsilon}{1+e^\varepsilon} \approx \frac{1}{2} + \frac{\varepsilon}{4} + \bigOh{\varepsilon^3},
\end{equation}
so if the two hypotheses induce only a small effective privacy loss, then the adversary can do only slightly better than random guessing. In the analysis of our protocols below, the noise is calibrated to protect a worst-case change over the full parameter range,
\begin{equation}
    \Delta = \theta_{\max} - \theta_{\min}.
\end{equation}
For the linear sensitivity-calibrated mechanisms considered below, if two hypotheses differ only by $d \coloneqq \abs{b-a} \leq \Delta$, then the relevant pairwise sensitivity is reduced by $d/\Delta$. As a result, when each query is calibrated to a privacy budget of $(\varepsilon_0, \delta_0)$ for the worst-case change $\Delta$, the same release is actually $(\varepsilon_0d/\Delta, \delta_0)$-differentially private for distinguishing two values separated by $d$. By the composition theorem in~\cref{thm:composition-thm}, the release of $k$ distinct queries, each with privacy budget $(\varepsilon_0, \delta_0)$, has privacy budget
\begin{equation}
    \varepsilon = k\varepsilon_0\frac{d}{\Delta},\quad \delta = k\delta_0
\end{equation}
and so
\begin{equation}
    p_{\mathrm{succ}} \leq \frac{e^{k\varepsilon_0d/\Delta} + k\delta_0}{1 + e^{k\varepsilon_0d/\Delta}}.
\end{equation}
Thus, when $d \ll \Delta/(k\varepsilon_0)$ (i.e., the hypotheses are very close to each other) and $k\delta_0$ is small (as is often the case with the mechanisms we use), then the adversary's optimal success probability is close to $1/2$. When $d=\Delta$, then this reduces to the usual worst-case statement.

\section{Centralized network}\label{sec:centralized-network}
We now introduce our differentially private quantum sensing protocols. We do so in two different settings, determined by the structure of the network. In the \emph{centralized network setting}, we assume that the sensor network is administered by a trusted central party. This allows us to implement a secure quantum sensing protocol based on \emph{global} differential privacy. On the other hand, if the sensor network does not have access to a trusted central party, we consider the \emph{decentralized network setting}, and we resort to \emph{local} differential privacy mechanisms for secure quantum sensing; see~\cref{sec:decentralized-network}. In this section, we consider the centralized network setting, which we define as follows:

\begin{definition}[Centralized network setting]\label{def:centralized-network-setting}
    The \emph{centralized network} is one where a trusted, central party or \emph{curator} runs the differentially private quantum sensing protocol.
\end{definition}

The centralized network is vulnerable to both internal attacks (attacks by nodes in the network) and external attacks (attacks by third parties). Our goal is to calculate an estimate $Q$ of the function $q(\boldsymbol{\theta})$ in such a way that we reveal as little information as possible about each node's parameter (i.e., their input). However, we allow for the nodes to be adversarial in that they could choose not to follow the protocol or collude in some way to try to learn the parameter of another node in the network. We also allow for a malicious third party that can delegate sensing tasks to the curator in order to try to learn the parameter value of a target node in the network.

The primary attack (see~\cref{fig:adversarial-model}) that we consider in this work is a \emph{differencing attack}, which we have mentioned a few times already. In the context of quantum sensor networks, this attack can be achieved in a few physically different, though mathematically equivalent, ways. One way (considered in~\cref{sec:intro}) is for the external adversary to delegate the estimation of a function involving all of the nodes, then to estimate a function excluding the target node, and finally to calculate the difference between these two function estimates (scaled appropriately, since we are only considering the average function) to get the value of the excluded parameter. For example, assume that we have three parameters such that $\boldsymbol{\theta} = (\theta_1, \theta_2, \theta_3)$, and the adversary wishes to know the value of $\theta_1$. The adversary could first ask for the estimate of the function $q_1(\boldsymbol{\theta}) = \frac{\theta_1 + \theta_2 + \theta_3}{3}$ and then ask for the estimate of $q_2(\boldsymbol{\theta}) = \frac{\theta_2 + \theta_3}{2}$. Then, given the estimates $Q_1$ and $Q_2$ of $q_1(\boldsymbol{\theta})$ and $q_2(\boldsymbol{\theta)}$, respectively, all the adversary needs to do to estimate $\theta_1$ is calculate
\begin{align}
    3Q_1 - 2Q_2 &\approx 3\lp \frac{\theta_1 + \theta_2 + \theta_3}{3} \rp - 2\lp \frac{\theta_2 + \theta_3}{2} \rp\\
    &= \theta_1.
\end{align}
An even more straightforward way to achieve this that only requires a single run of the protocol is for the adversary to take the coefficient of the parameter that they wish to know to be 1 and all other coefficients equal to 0 such that the "function" that the network is estimating is actually just the parameter value itself. This is also equivalent to the adversarial nodes not coupling their part of the GHZ state to their parameter. If the trusted node follows the sensing protocol honestly, then it will end up just broadcasting its bare parameter, compromising its privacy.

One objection to these attacks is that a mechanism could be put in place to "catch" queries that may be compromising of private information. This is called \emph{query auditing}, and it is problematic for two reasons~\cite{dwork2014}. First, the denial of a specific query or set of queries alone may reveal some information. Second, query auditing is in general challenging: the space of possible queries is so vast that trying to handle all of them is often infeasible.

\begin{figure}[t]
    \begin{algorithm}[H]
        \caption{Global Laplace mechanism}\label{alg:global-sensing-protocol}
        \begin{algorithmic}[1]
            \Require Same input as~\cref{alg:standard-entangled-sensing-protocol}, $\varepsilon$
            \Ensure $\Tilde{Q} = Q + \eta$
            \State $Q \gets$ function estimate using standard sensing protocol
            \State $b \gets \frac{\theta_{\max} - \theta_{\min}}{n\varepsilon}$
            \State $\eta \gets \Lap{b}$
            \State $\Tilde{Q} \gets Q + \eta$
            \State \Return $\Tilde{Q}$
        \end{algorithmic}
    \end{algorithm}
\end{figure}

We now give the differentially private quantum sensing protocol for the centralized network setting to measure the average function, which is based on a global Laplace mechanism. First, the curator creates a GHZ state and sends one qubit to each of the nodes in the network. Then, each sensor is coupled to a local parameter and evolved unitarily (with no controls), after which a measurement is made at each sensor and the results sent as classical bits to the curator. The curator calculates an estimate $Q$ of the function $q$ according to the standard protocol, applies noise to this function estimate in such a way as to achieve $\varepsilon$-differential privacy, and then broadcasts the noisy result; see~\cref{fig:global-dp-picture} and~\cref{alg:global-sensing-protocol}. Since the curator is assumed to be trusted in this model, we do not worry about the curator running any attacks or compromising the data in any way. The main attacks we are concerned with in this model are those by adversarial third parties and adversarial nodes within the network.

\begin{figure}[t]
    \centering
    \begin{quantikz}
        \gategroup[wires=5,steps=3,style={dashed,rounded corners,fill=red!20,inner xsep=2pt,inner ysep=8pt},background,label style={label position=below,anchor=north,yshift=-0.2cm}]{Repeated $M$ times}\midstick[wires=5,brackets=right]{$\ket{\text{GHZ}}_n$} & \gate{U_1(t_j)} & \meter{x_1} & \setwiretype{c} \midstick[wires=5,brackets=left]{$\vb{x}$}\\
        & \gate{U_2(t_j)} & \meter{x_2} & \setwiretype{c}\\
        \setwiretype{n} & \vdots & \vdots & \arrow[r] & \rstick{\fcolorbox{black}{gray!20}{\parbox{0.75cm}{\centering $\Tilde{Q}$}}}\\
        & \gate{U_{n-1}(t_j)} & \meter{x_{n-1}} & \setwiretype{c}\\
        & \gate{U_n(t_j)} & \meter{x_n} & \setwiretype{c}
    \end{quantikz}
    \caption{Centralized entangled sensor network protocol using a global Laplace mechanism. The entangled sensor protocol is run as normal with no privacy mechanism, and the $i\numth$ node sends its $M$ classical measurement results $\vb{x}_i = (x^{(1)}_i, \ldots, x^{(M)}_i)$ to the trusted curator over a secure, classical channel, who then collects all measurement results into a vector $\vb{x} = (\vb{x}_1, \ldots, \vb{x}_n)$ and calculates the function estimate $Q$. Before the curator broadcasts the function estimate, though, it applies noise $\eta$ drawn from the Laplace distribution to $Q$ and then broadcasts this noisy value, $\Tilde{Q} = Q + \eta$.}
    \label{fig:global-dp-picture}
\end{figure}

To be more concrete, recall that we are estimating the average function $q(\boldsymbol{\theta}) = \frac{1}{n}\defsum{i=1}{n}{\theta_i}$, where $\boldsymbol{\theta} = (\theta_1, \ldots, \theta_n)$ are unknown parameters, and we assume that the $n$ nodes each send their $\sigma^x$ measurement results $x_i$ to the curator over a secure classical communication channel. Assume also that $\theta_i \in [\theta_{\min}, \theta_{\max}]$ for some known $\theta_{\min}$ and $\theta_{\max}$. The curator then uses the measurement results received from the nodes to get an estimate $Q$ of $q(\boldsymbol{\theta})$ to one bit of precision. The remaining bits of precision can be estimated analogously using the bit-by-bit learning protocol outlined in~\cref{subsec:sensor-network}. The curator now has an estimate $Q$ of the true function value $q(\boldsymbol{\theta})$, to which the curator then adds noise $\eta \sim \Lap{b}$, where the scale parameter is
\begin{align}
    b = \frac{\theta_{\max} - \theta_{\min}}{n\varepsilon}.
\end{align}
As mentioned in~\cref{sec:prelims}, we choose the Laplace distribution because it has been shown~\cite{vadhan2017} that this distribution is able to achieve $\varepsilon$-differential privacy, while a more familiar distribution like the Gaussian distribution is unable to achieve this level of privacy due to its behavior at the tails of the distribution. If the privacy budget is relaxed to $(\varepsilon, \delta)$-differential privacy for $\delta > 0$, then Gaussian noise can be used. Either way, the curator broadcasts the noisy function estimate $\Tilde{Q} = Q + \eta$. Under the no-coupling attack, in which only a single honest node $j$ couples its qubit to its parameter $\theta_j$, the full network average being estimated is effectively $\theta_j/n$. The curator thus releases
\begin{equation}
    \Tilde{Q} = Q + \eta,
\end{equation}
where 
\begin{equation}
    \eta \sim \Lap{\frac{\Delta}{n\varepsilon}}.
\end{equation}
If the adversary post-processes the release by multiplying by $n$, then they obtain
\begin{equation}
    n\Tilde{Q} = \theta_j + \eta',
\end{equation}
where $\eta' \sim \Lap{\Delta/\varepsilon}$, which is the usual Laplace scale to protect a parameter varying over an interval of width $\Delta$. We state the optimality of this protocol in the following theorem:

\begin{theorem}[Global Laplace mechanism]\label{thm:global-laplace-mechanism}
    The global Laplace mechanism, in which a trusted curator releases $\Tilde{Q} = Q + \eta$ with $\eta \sim \Lap{b}$ with scale parameter $b = \frac{\mathrm{GS}}{\varepsilon}$, yields a differentially private entangled sensing protocol for estimating the function $q(\boldsymbol{\theta}) = \frac{1}{n}\defsum{i=1}{n}{\theta_i}$. The mean-squared error achieves Heisenberg scaling and the mechanism is $\varepsilon$-differentially private with $\varepsilon = \bigTheta{1}$, but it is a global mechanism that requires the extra assumption of a trusted curator and so fails to be local by construction.
\end{theorem}

\begin{proof}    
    We first determine the mean-squared error, which will allow us to bound the noise we must add to achieve both Heisenberg scaling and $\varepsilon$-differential privacy. We start by finding the expected value of $\Tilde{Q}$ over quantum randomness, indicated with the subscript $\psi$, conditioned on the noise $\eta$:
    \begin{align}
        \mathbb{E}_{\psi}[\Tilde{Q} \mid \eta] &= \mathbb{E}_\psi[Q + \eta]\\
        &= \mathbb{E}_\psi[Q] + \mathbb{E}_\psi[\eta]\\
        &= q(\boldsymbol{\theta}) + \eta.\label{eqn:global-biased-estimator}
    \end{align}
    We find the variance, conditioned on the noise added, as
    \begin{equation}
        \mathrm{Var}_\psi[\Tilde{Q} \mid \eta] = \mathrm{Var}_\psi[Q] = \bigOh{\frac{1}{n^2t^2}}.
    \end{equation}
    Here, the noise does not contribute to the variance because we calculate the variance over the quantum noise from the sensing protocol, and so $\eta$ (which is not resampled) is effectively a constant, where the variance of a constant is zero. Putting this together, the mean-squared error conditioned on the sampled noise is
    \begin{equation}
        \mathbb{E}_\psi[(\Tilde{Q} - q(\boldsymbol{\theta}))^2 \mid \eta] = \bigOh{\frac{1}{n^2t^2}} + \eta^2.
    \end{equation}
    Now, we average over the Laplace noise, to find
    \begin{align}
        \varepsilon_{\mathrm{MSE}} &= \mathbb{E}_\eta[\mathbb{E}_\psi[(\Tilde{Q} - q(\boldsymbol{\theta}))^2 \mid \eta]]\\
        &= \bigOh{\frac{1}{n^2t^2}} + \mathbb{E}_\eta[\eta^2].
    \end{align}
    Then, we use the fact that $\eta$ is sampled from a Laplace distribution $\Lap{b}$ with mean 0 and scale parameter $b = \mathrm{GS}/\varepsilon$, according to~\cref{def:laplace-mechanism}. For neighboring parameter vectors $\boldsymbol{\theta}$ and $\boldsymbol{\theta}'$ that differ in one coordinate, we have
    \begin{equation}
        \mathrm{GS} = \abs{q(\boldsymbol{\theta}) - q(\boldsymbol{\theta}')} = \frac{\theta_{\max} - \theta_{\min}}{n},
    \end{equation}
    Thus, we have
    \begin{equation}
        b = \frac{\theta_{\max} - \theta_{\min}}{n\varepsilon},\label{eqn:b-global-mechanism}
    \end{equation}
    and so
    \begin{equation}
        \mathbb{E}_\eta[\eta^2] = \Var{\eta} = 2b^2 = 2\lp \frac{\theta_{\max} - \theta_{\min}}{n\varepsilon} \rp^2.
    \end{equation}
    Putting everything together, we have
    \begin{equation}
        \varepsilon_{\mathrm{MSE}} = \bigOh{\frac{1}{n^2t^2}} + 2\lp \frac{\theta_{\max} - \theta_{\min}}{n\varepsilon} \rp^2.
    \end{equation}
    The domain for the parameter values is independent of the number of sensors and the sensing time, so $\lp \theta_{\max} - \theta_{\min} \rp$ is a constant, and it suffices to take $\varepsilon = \bigTheta{1}$. Thus, $\varepsilon_{\mathrm{MSE}} = \bigOh{1/n^2}$ and $\varepsilon = \bigTheta{1}$, making this an $\varepsilon$-differentially private mechanism by~\cref{thm:laplace-mechanism-properties}. By the conditions in~\cref{def:optimality-conditions}, this protocol is nearly optimal, only failing the locality condition.
\end{proof}

We note that the above analysis is for a single release of the full network average query. More generally, we can apply the same mechanism to average-type queries of sufficiently large subsets of the network, which we denote as
\begin{equation}
    q_S(\boldsymbol{\theta}) = \frac{1}{\abs{S}}\defsum{i\in S}{}{\theta_i},
\end{equation}
where the chosen subset $S \subseteq [n]$ is public and fixed before running the protocol. To retain the favorable $\bigOh{1/n^2}$ Heisenberg scaling in the mean-squared error, we restrict the subsets to be of size
\begin{equation}
    \abs{S} \geq \beta n,
\end{equation}
for some $0 < \beta \leq 1$. For such a query, the global sensitivity is
\begin{equation}
    \mathrm{GS}(q_S) = \frac{\Delta}{\abs{S}},
\end{equation}
where $\Delta \coloneqq \theta_{\max} - \theta_{\min}$. After running the protocol, the trusted curator releases
\begin{equation}
    \Tilde{Q}_S = Q_S + \eta_S,
\end{equation}
where $Q_S$ is the estimate for the subset average $q_S(\boldsymbol{\theta})$ and
\begin{equation}
    \eta_S \sim \Lap{\frac{\Delta}{\abs{S}\varepsilon_0}}
\end{equation}
for desired privacy parameter $\varepsilon_0$. The resulting mean-squared error scales as
\begin{equation}
    \varepsilon_{\mathrm{MSE}} = \bigOh{\frac{1}{\abs{S}^2t^2}} + 2\lp \frac{\Delta}{\abs{S}\varepsilon_0} \rp^2,
\end{equation}
which, for $\abs{S} \geq \beta n$ and $\varepsilon_0 = \bigTheta{1}$ gives us $\varepsilon_{\mathrm{MSE}} = \bigOh{1/n^2}$.

This also protects a single node from the no-coupling attack, where the targeted node's parameter is isolated. For example, suppose only node $j \in S$ couples its qubit to its parameter while all other nodes in $S$ do not couple their qubits, then the effective average being estimated is $\theta_j/\abs{S}$. The released value is thus
\begin{equation}
    \Tilde{Q} = \frac{\theta_j}{\abs{S}} + \eta_S,
\end{equation}
where
\begin{equation}
    \eta_S \sim \Lap{\frac{\Delta}{\abs{S}\varepsilon_0}}.
\end{equation}
Even if the adversary post-processes the release by multiplying by $\abs{S}$, the adversary merely learns
\begin{equation}
    \abs{S}\Tilde{Q}_S = \theta_j + \eta'_S,
\end{equation}
where
\begin{equation}
    \eta'_S \sim \Lap{\frac{\Delta}{\varepsilon_0}},
\end{equation}
which is exactly the scale required to protect a parameter that can vary over an interval of width $\Delta$.

An adversary could also ask for several different functions $q_{S_1}(\boldsymbol{\theta}), q_{S_2}(\boldsymbol{\theta}), \ldots, q_{S_k}(\boldsymbol{\theta})$ for some constant $k$ and where $\abs{S_r} \geq \beta n$ for each $r \in [k]$, where the functions are chosen such that they can reveal information about one or more nodes, as illustrated in~\cref{fig:adversarial-model}. Thus, for each distinct\footnote{If the same function is asked for again, the same estimate found before can be released.} query $q_{S_r}(\boldsymbol{\theta})$ with privacy budget $\varepsilon_r$ and freshly sampled noise
\begin{equation}
    \eta_r \sim \Lap{\frac{\Delta}{\abs{S_r}\varepsilon_r}},
\end{equation}
the curator releases an estimate $\Tilde{Q}_{S_r} = Q_{S_r} + \eta_r$ that is $\varepsilon_r$-differentially private. Therefore, the release of the full transcript, $(\Tilde{Q}_{S_1}, \Tilde{Q}_{S_2}, \ldots, \Tilde{Q}_{S_k})$, is $\lp\defsum{r=1}{k}{\varepsilon_r}\rp$-differentially private by the composition theorem (see~\cref{thm:composition-thm}). If each query uses the same privacy parameter $\varepsilon_0$, then the $k$-query transcript is $k\varepsilon_0$-differentially private. If $k = \bigTheta{1}$ and $\varepsilon_0 = \bigTheta{1}$, then the release of multiple function estimates remains $\bigTheta{1}$-differentially private.

The interpretation of differential privacy in the sensing context from~\cref{subsec:dp-in-sensing} bounds the adversary's ability to discriminate between two possible values for a target node's parameter $\theta_j$, separated by $d \leq \Delta$. Then, for every allowable queried subset $S$ containing $j$, the sensitivity is $d/\abs{S} \leq \Delta/\abs{S}$. Thus, with noise calibrated to privacy budget $\varepsilon_0$ with worst-case change $\Delta$ gives an effective privacy level of $\varepsilon_0d/\Delta$. For $k$ distinct queries, this implies
\begin{equation}
    p_{\mathrm{succ}} \leq \frac{e^{k\varepsilon_0d/\Delta}}{1 + e^{k\varepsilon_0d/\Delta}},
\end{equation}
where $\delta=0$ for the Laplace mechanism. When the two candidate parameter values are close together such that $d \ll \Delta/(k\varepsilon_0)$, this probability becomes close to $1/2$, meaning the adversary has a success probability only slightly better than a coin flip for discriminating between the two parameter values. When $d=\Delta$, this reduces to the worst-case guarantee.

While this global mechanism does not meet all of the optimality conditions (it fails the locality condition), it is useful if there is a known trusted curator administering the network. For one, this protocol is easier to implement experimentally since the only modification to the original sensing protocol is that noise sampled from a classical distribution is applied globally to the function estimate at the end of the sensing protocol. A natural setting for this paradigm is in \emph{delegated sensing}. Here, the client is a third party that instructs the network what function to compute and a trusted server sends back a noisy function estimate that is close to the true function value but hides the network's parameter values from the client.

\section{Decentralized network}\label{sec:decentralized-network}
We now drop the assumption that a trusted curator administers the network and consider the \emph{decentralized network setting} in which the nodes themselves orchestrate the protocol. This setting is more general in the sense that we do not assume the existence of a trusted central party, which may be a strong assumption in many cases. For the basic noisy Hamiltonian protocol in~\cref{subsec:noisy-hamiltonian-protocol}, we focus on the single-release setting; composition can be applied analogously. For the honest-fraction protocol in~\cref{subsec:nhp-honest-fraction}, we discuss the multi-query setting explicitly. We start with a formal definition of the decentralized network setting and then discuss it in more detail below:

\begin{definition}[Decentralized network setting]\label{def:decentralized-network-setting}
    The \emph{decentralized network setting} consists of a network of qubit sensors with no trusted central party.
\end{definition}

Like the centralized network in~\cref{def:centralized-network-setting}, the decentralized network is vulnerable to both internal and external attacks. The physical picture is mostly the same as the centralized network in~\cref{sec:centralized-network}, where we have a network of $n$ qubit sensors tasked with estimating a function $q(\boldsymbol{\theta})$ without revealing the individual parameters. The primary difference is that we now lack a trusted curator that can create GHZ states, distribute them throughout the network, and receive each node's measurement outcomes to calculate the function estimate $Q$. As such, we need a way for the mutually adversarial nodes to create entanglement among themselves. This can be done, for example, by first establishing a ring of Bell pairs and then performing local operations~\cite{komar2016}, which involves a series of \textsc{cnot} gates between the pairs in the ring, measurement of the target qubits, and single-qubit rotations. Since the nodes are mutually adversarial, some nodes may claim to be doing the local operations specified in this protocol, when they are actually doing some other operations (or nothing at all). Nonetheless, we can verify the entanglement generation by using a protocol such as that presented in Ref.~\cite{pappa2012}, which assumes a common source of randomness to make it efficient. We treat this GHZ verification step as a black box and leave an explicit protocol tailored to our setting to future work. We introduce a local mechanism that we call the \emph{noisy Hamiltonian protocol} in~\cref{subsec:noisy-hamiltonian-protocol} and an extension with the additional assumption of an honest fraction in~\cref{subsec:nhp-honest-fraction}.

\subsection{Failure of classical randomized response}\label{subsec:classical-randomized-response}
We start by considering one of the most commonly used mechanisms to implement differential privacy: randomized response, which we introduced in~\cref{def:randomized-response} in~\cref{subsec:classical-dp}. Here, we map the Pauli-$X$ measurement outcomes +1 and -1 to bits 0 and 1, respectively. A straightforward approach is to apply this mechanism directly before each node broadcasts its bit, that is, each node flips its measured bit with probability $\frac{1}{e^{\varepsilon}+1}$. While it has been shown that randomized response has optimal additive error within the local model~\cite{chan2012,vadhan2017}, we demonstrate that this approach results in an exponential decrease in the ability to estimate the underlying function. Though in general we estimate the function to several bits of precision, for ease of calculation we analyze a single bit of precision, which suffices to demonstrate the lack of utility of this protocol. The expected value of the noisy parity bit after applying the randomized response mechanism is given by
\begin{align}
    \mathbb{E}[\Tilde{B} \mid x] = x \cdot \Pr[\Tilde{B} = x] + \overline{x} \cdot \Pr[\Tilde{B} = \overline{x}],\label{eqn:rr-prob}
\end{align}
where $x$ represents the true parity of the bits such that $x = \defsum{i=1}{n}{x_i} \pmod{2}$ and where $\overline{x} = x \oplus 1$. Each bit $x_i \in \{0,1\}$ is processed independently by randomized response, producing a noisy bit $\tilde{B}_i \in \{0,1\}$, and $\Tilde{B} = \defsum{i=1}{n}{\Tilde{B}_i \pmod{2}}$ is a random variable that represents the parity of the sum of the output noisy bits after applying randomized response. We define $p \coloneqq \frac{1}{e^{\varepsilon}+1}$ as the probability of flipping the bit and expand the probability terms in~\cref{eqn:rr-prob}:
\begin{align}
  \Pr[\Tilde{B} = \overline{x}] &= \sum_{\substack{k=1,\\ k \text{ odd}}}^{n} \binom{n}{k}  p^k\lp 1-p \rp^{n-k}\\
  &= \frac{1}{2} \lp (1-p+p)^n - (1-p-p)^n \rp\\
  &= \frac{1}{2} \lp 1- \left(1-\frac{2}{e^{\varepsilon}+1}\right)^n \rp\\
  &=\frac{1}{2}-\frac{1}{2}\,\left(\tanh{\frac{\varepsilon}{2}}\right)^n.
\end{align}
This gives
\begin{equation}
   \mathbb{E}[\Tilde{B} \mid x] = \frac{1}{2} + \frac{2x-1}{2} \left(\tanh{\frac{\varepsilon}{2}}\right)^n.
\end{equation}
The bias magnitude is 
\begin{align}
    \abs{\mathbb{E}[\Tilde{B} \mid x] - \frac{1}{2}} &= \abs{\frac{2x-1}{2}}\lp \tanh{\frac{\varepsilon}{2}} \rp^n\\
    &= \frac{1}{2} \lp \tanh{\frac{\varepsilon}{2}} \rp^n.
\end{align}
For any fixed privacy parameter $\varepsilon > 0$,
\begin{align}
    0 < \tanh\lp \frac{\varepsilon}{2} \rp < 1,
\end{align}
so
\begin{align}
    \lp\tanh\left(\frac{\varepsilon}{2}\right)\rp^n &= \exp \lp n\ln \tanh\left(\frac{\varepsilon}{2}\right)\rp\\
    &= \exp(-cn),
\end{align}
where $c \coloneqq -\ln\tanh(\varepsilon/2) > 0$, which is exactly exponential decay in $n$. Thus, we can see that trivially adding the classical randomized response mechanism to each node's measurement result strongly dampens the signal, so the reported parity becomes exponentially close to a completely random bit as $n$ grows. This makes sense because the parity function is a sensitivity-1 function, meaning even a single bit flip can drastically change the value of the function. While we could choose $\varepsilon$ to be $n$-dependent in a way that the bias becomes constant, such an $n$-dependent $\varepsilon$ would destroy privacy as $n$ increases.

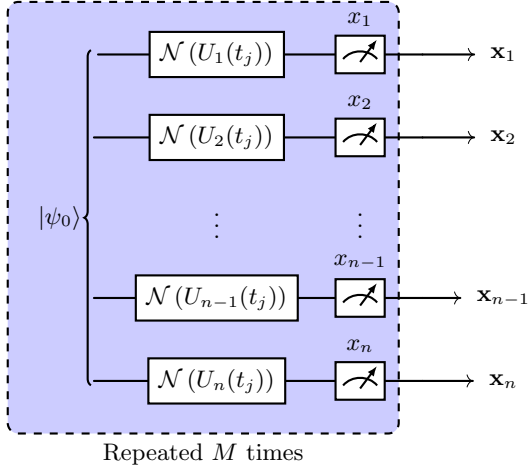
\begin{figure}[t]
    \centering
    \begin{quantikz}
        \gategroup[wires=5,steps=3,style={dashed,rounded corners,fill=blue!20,inner xsep=2pt,inner ysep=8pt},background,label style={label position=below,anchor=north,yshift=-0.2cm}]{Repeated $M$ times}\midstick[wires=5,brackets=right]{$\ket{\psi_0}$} & \gate{\mathcal{N}\lp U_1(t_j)\rp} & \meter{x_1} & \arrow[r] & \midstick[1]{$\vb{x}_1$} \\
        & \gate{\mathcal{N}\lp U_2(t_j)\rp} & \meter{x_2} & \arrow[r] & \midstick[1]{$\vb{x}_2$} \\
        \setwiretype{n} & \vdots & \vdots & & \\
        & \gate{\mathcal{N}\lp U_{n-1}(t_j)\rp} & \meter{x_{n-1}} & \arrow[r] & \midstick[1]{$\vb{x}_{n-1}$} \\
        & \gate{\mathcal{N}\lp U_n(t_j)\rp} & \meter{x_n} & \arrow[r] & \midstick[1]{$\vb{x}_n$} 
    \end{quantikz}
    \caption{Decentralized entangled sensor network protocol using a local differential privacy mechanism. A GHZ state $\ket{\psi_0}$ is generated by the nodes, where each node holds one qubit of the GHZ state. Each node evolves their qubit under a "noisy" unitary, denoted $\mathcal{N}(U_i(t))$, measures their qubit in the Hadamard basis, and broadcasts the measurement results $\vb{x}_i$ to all nodes in the network, each of which can then calculate the noisy function estimate $\Tilde{Q}$.}
    \label{fig:local-dp-picture}
\end{figure}

\subsection{Noisy Hamiltonian protocol}\label{subsec:noisy-hamiltonian-protocol}
As a direct application of the randomized response mechanism does not work for our goal of differentially private quantum sensing, we must resort to another mechanism. As discussed in~\cref{subsec:classical-dp}, the \emph{Laplace mechanism} is another popular and successful mechanism for achieving differential privacy and it is the mechanism we used in~\cref{sec:centralized-network}. However, because we are now in the decentralized network setting, the noise must be applied locally, rather than by a trusted curator. Thus, we instead add noise sampled from the Laplace distribution directly to the Hamiltonian parameter coupled to each node's qubit, giving us the \emph{noisy Hamiltonian protocol}, as detailed in~\cref{alg:noisy-hamiltonian-qldp} and the schematic in~\cref{fig:local-dp-picture}. Starting with the original sensing Hamiltonian from~\cref{subsec:sensor-network},
\begin{align}
    H = \frac{1}{2}\defsum{i=1}{n}{\theta_i\sigma_i^z},
\end{align}
we introduce random noise $\eta_i$ to each node to obtain a noisy Hamiltonian $\Tilde{H}$. That is, we take $\theta_i \to \theta_i + \eta_i$ for all $i \in [n]$, yielding
\begin{align}
    \Tilde{H} = \frac{1}{2}\defsum{i=1}{n}{\lp \theta_i + \eta_i \rp\sigma_i^z}.\label{eqn:noisy-hamiltonian}
\end{align}
Experimentally, this can be achieved by having each node adjust their applied signal, determined by the noise term, which is sampled locally (and thus known only to the local node). We again choose the noise to be sampled from a Laplace distribution, that is, $\eta_i \sim \Lap{b}$, where $b$ is the scale parameter and where we have taken the mean to be 0.

We note a general bound on $b$ to achieve $\varepsilon$-differential privacy, which will be useful to determine what the resulting privacy parameter $\varepsilon$ will be:

\begin{theorem}[Generic bound on $b$]\label{thm:general-sigma-bound}
    Let $\varepsilon > 0$ and $\theta_{\min}, \theta_{\max} \in \RR$ be the minimum and maximum values a parameter can take, respectively. The noisy Hamiltonian protocol is $\varepsilon$-locally differentially private against both classical and quantum adversaries, where noise is sampled from $\Lap{b}$ with
    \begin{align}
        b \geq \frac{\theta_{\max} - \theta_{\min}}{\varepsilon}.
    \end{align}
\end{theorem}

\begin{proof}
    The choice of $b \geq \frac{\theta_{\max} - \theta_{\min}}{\varepsilon}$ follows directly from the definitions of local differential privacy and the Laplace mechanism. Then, with $b \geq \frac{\theta_{\max} - \theta_{\min}}{\varepsilon}$, the sub-algorithm in the sensing protocol, which adds $\eta_i \sim \Lap{b}$ to $\theta_i$, is $\varepsilon$-locally differentially private by the Laplace mechanism, as stated in~\cref{thm:laplace-mechanism-properties}. Since the input parameter $\theta_i$ is not further accessed by the protocol, the complete private sensing protocol remains $\varepsilon$-locally differentially private due to the post-processing property, as stated in~\cref{thm:post-prosessing}.
\end{proof}

This theorem should be viewed as the most robust privacy guarantee for the basic noisy Hamiltonian protocol. Since the honest node's raw parameter is replaced locally by a noisy parameter $\theta_i + \eta_i$ before the rest of the protocol sees it, any subsequent information released about it to the adversary is a post-processing of this local value. This includes product state attacks, adversarially chosen measurements, individual measurement outcomes, post-measurement quantum systems, and arbitrary classical or quantum side information. The cost of this generality is a more conservative local Laplace calibration.

This bound is generic and is not necessarily tight for all of the protocols that we introduce. In particular, for a limited regime where we can find a closed-form expression for $\varepsilon$ in the noisy Hamiltonian protocol using the quantum hockey-stick divergence, we find that the resulting bound is notably different from what we get from this generic bound (see~\cref{app-subsec:hockey-stick-divergence}). We also note that because differential privacy is an information-theoretic concept, the privacy guarantee holds against both classical and quantum adversaries (i.e., the computational power of the adversary does not help).

\begin{figure}[t]
    \begin{algorithm}[H]
        \small
        \caption{Noisy Hamiltonian protocol}\label{alg:noisy-hamiltonian-qldp}
        \begin{algorithmic}[1]
            \Require Same input as~\cref{alg:standard-entangled-sensing-protocol}, scale parameter $b = b(\alpha)$, $1 \leq \alpha \leq 2$
            \Ensure $\Tilde{Q}$
            \For{$i \in [n]$}
                \State $\eta_i \sim \Lap{b}$
            \EndFor
            \State $\Tilde{H} \gets \frac{1}{2}\defsum{i=1}{n}{(\theta_i + \eta_i)\sigma^z_i}$
            \State Run sensing protocol with $\Tilde{H}$ as Hamiltonian
            \State $\Tilde{Q} \gets$ noisy estimate of $\Tilde{q}(\boldsymbol{\theta})$
            \State \Return $\Tilde{Q}$
        \end{algorithmic}
    \end{algorithm}
\end{figure}

With this result in mind, we consider two variations of the noisy Hamiltonian protocol, one in which each node samples their noise anew with each shot of the sensing protocol and one in which each node samples their noise once and keeps it constant for every shot of sensing. Interestingly, as we show in~\cref{app-sec:noisy-hamiltonian-analysis}, we find that when resampling the noise between each shot of sensing, we end up destroying both the accuracy with which we can estimate the function and the resulting privacy, so our protocol requires that the noise is sampled only once at the start of the experiment. We quantify the performance of the noisy Hamiltonian protocol in the following theorem:

\begin{theorem}[Performance of the noisy Hamiltonian protocol]\label{thm:noisy-hamiltonian-optimality}
    Assume that each honest node samples their local Laplace noise $\eta_i \overset{\mathrm{i.i.d.}}{\sim} \Lap{b}$ once at the start of the protocol and retains the same noise for all sensing shots. Then, the $n$-sensor noisy Hamiltonian protocol in~\cref{alg:noisy-hamiltonian-qldp} is a local protocol that admits mean-squared error
    \begin{equation}
        \varepsilon_{\mathrm{MSE}} = \bigOh{\frac{1}{n^2t^2}} + \frac{2b^2}{n}\label{eqn:mse-scaling-nhp-thm}
    \end{equation}
    in the honest setting. For any node $i$ to protect its data against up to $n-1$ adversaries, the noisy Hamiltonian protocol is $\varepsilon$-differentially private, where to achieve a mean-squared error scaling as $\bigOh{n^{-\alpha}}$, each honest node takes
    \begin{equation}
        b = \bigTheta{n^{-(\alpha-1)/2}},
    \end{equation}
    which, under the generic bound in~\cref{thm:general-sigma-bound}, yields
    \begin{equation}
        \varepsilon = \Theta\lp n^{(\alpha-1)/2} \rp.
    \end{equation}
\end{theorem}

\begin{proof}
    Note that, by construction, this protocol is local since each node applies their noise locally and no trusted curator is assumed to exist. To analyze the mean-squared error, we note that by introducing Laplace noise, the Hamiltonian becomes
    \begin{equation}
        \Tilde{H} = \frac{1}{2}\defsum{i=1}{n}{(\theta_i + \eta_i)\sigma_i^z}.
    \end{equation}
    Thus, the GHZ state is actually exposed to the quantity
    \begin{equation}
        \Tilde{q}(\boldsymbol{\theta}) = \frac{1}{n}\defsum{i=1}{n}{(\theta_i+\eta_i)} = q(\boldsymbol{\theta}) + \overline{\eta},
    \end{equation}
    where $\overline{\eta} = \frac{1}{n}\defsum{i=1}{n}{\eta_i}$. We calculate the mean-squared error as
    \begin{equation}
        \varepsilon_{\mathrm{MSE}} = \mathbb{E}[(\Tilde{Q} - q(\boldsymbol{\theta}))^2] = \mathbb{E}_{\eta}[\mathbb{E}_\psi[(\Tilde{Q} - q(\boldsymbol{\theta}))^2 \mid \overline{\eta}]],
    \end{equation}
    where we separately take the expectation over the quantum randomness and the randomness from the noise. We can break the inner expectation into the variance part plus the bias part as
    \begin{equation}
        \mathbb{E}_\psi[(\Tilde{Q} - q(\boldsymbol{\theta}))^2 \mid \overline{\eta}] = \mathrm{Var}_\psi[\Tilde{Q} \mid \overline{\eta}] + (\mathbb{E}_\psi[\Tilde{Q} \mid \overline{\eta}] - q(\boldsymbol{\theta}))^2.
    \end{equation}
    We calculate $\mathrm{Var}_\psi[\Tilde{Q} \mid \overline{\eta}]$ over the quantum noise, thus treating $\overline{\eta}$ as a constant, as
    \begin{align}
        \mathrm{Var}_\psi[\Tilde{Q} \mid \overline{\eta}] &= \mathrm{Var}_\psi[Q + \overline{\eta} \mid \overline{\eta}]\\
        &= \mathrm{Var}_\psi[Q \mid \overline{\eta}] + \underbrace{\mathrm{Var}_\psi[\overline{\eta} \mid \overline{\eta}]}_0 + 2\underbrace{\CovSpec{\psi}{Q}{\overline{\eta} \mid \overline{\eta}}}_0\\
        &= \bigOh{\frac{1}{n^2t^2}},
    \end{align}
    where we used~\cref{eqn:standard-heisenberg-scaling} to find $\mathrm{Var}_\psi[Q | \overline{\eta}]$, the fact that the variance of a constant is 0, and the fact that the noise is constant and hence the covariance term is 0. Conditioned on $\overline{\eta}$, the expected value over the quantum randomness is
    \begin{equation}
        \mathbb{E}_\psi[\Tilde{Q}|\overline{\eta}] = \Tilde{q}(\boldsymbol{\theta}) = q(\boldsymbol{\theta}) + \overline{\eta}.
    \end{equation}
    Thus, we now have
    \begin{equation}
        \varepsilon_{\mathrm{MSE}} = \mathbb{E}[(\Tilde{Q} - q(\boldsymbol{\theta}))^2] = O\lp \frac{1}{n^2t^2} \rp + \mathbb{E}_\eta[\overline{\eta}^2],
    \end{equation}
    so it remains to find $\mathbb{E}_\eta[\overline{\eta}^2]$. Averaging over the Laplace noise, where each $\eta_i \sim \Lap{b}$ with mean 0 and variance $2b^2$, and using the fact that each noise term is sampled independently, we have
    \begin{equation}
        \mathbb{E}_\eta[\overline{\eta}] = \frac{1}{n}\defsum{i=1}{n}{\mathbb{E}[\eta_i]} = 0,
    \end{equation}
    and
    \begin{align}
        \mathrm{Var}_\eta[\overline{\eta}] &= \Var{\frac{1}{n}\defsum{i=1}{n}{\eta_i}}\\
        &= \frac{1}{n^2}\defsum{i=1}{n}{\Var{\eta_i}}\\
        &= \frac{2b^2}{n}.
    \end{align}
    Thus, we have
    \begin{equation}
        \mathbb{E}[\overline{\eta}^2] = \Var{\overline{\eta}} = \frac{2b^2}{n}.
    \end{equation}    
    Putting everything together, we have
    \begin{equation}
        \varepsilon_{\mathrm{MSE}} = \bigOh{\frac{1}{n^2t^2}} + \frac{2b^2}{n}.\label{eqn:mse-scaling}
    \end{equation}
    By~\cref{thm:general-sigma-bound},
    \begin{equation}
        b \geq \frac{\theta_{\max} - \theta_{\min}}{\varepsilon}.
    \end{equation}
    Thus,
    \begin{equation}
        \varepsilon \geq \frac{\theta_{\max} - \theta_{\min}}{b}.
    \end{equation}
    We choose $b = \bigTheta{n^{-(\alpha-1)/2}}$ to enforce $\varepsilon_{\mathrm{MSE}} = \bigOh{n^{-\alpha}}$, which implies
    \begin{equation}
        \varepsilon = \bigTheta{n^{(\alpha-1)/2}},
    \end{equation}
    as claimed in the theorem.
\end{proof}

We note that we restrict our analysis to the GHZ state because, as we show in~\cref{app-sec:ghz-state-leakiness}, the GHZ state is among the set of states that is the worst for privacy, while being optimal for sensing. Note also that~\cref{thm:noisy-hamiltonian-optimality} implies a tradeoff between the privacy and the utility of the noisy Hamiltonian protocol, which we parameterize by $\alpha$. We are only interested in the regime $1 \leq \alpha \leq 2$, where $\alpha = 1$ gives us the standard quantum limit for the scaling of $\varepsilon_{\mathrm{MSE}}$, $\alpha = 2$ gives us Heisenberg scaling, and $1 < \alpha < 2$ gives us an intermediate scaling. Analyzing the resulting privacy parameter for each of these in turn, we have

\begin{enumerate}
    \item \textbf{Standard quantum limit}: Taking $\alpha = 1$ in the expression for $\varepsilon$, we find
    \begin{equation}
        \varepsilon = \bigTheta{1},
    \end{equation}
    that is, a constant privacy parameter $\varepsilon$ is possible, which is desired from the privacy perspective.
    \item \textbf{Heisenberg limit}: Taking $\alpha = 2$ in the expression for $\varepsilon$, we have
    \begin{equation}
        \varepsilon = \bigTheta{\sqrt{n}},
    \end{equation}
    that is, the privacy parameter grows with the square root of the number of sensors $n$.
    \item \textbf{Intermediate regime}: Between $\alpha = 1$ and $\alpha = 2$, $\varepsilon$ scales sub-linearly with $n$, but slower than the square root of $n$.
\end{enumerate}

Thus, we see that constant-$\varepsilon$ local differential privacy is only possible if we accept a mean-squared error that scales according to the standard quantum limit. In this case, there is no benefit to using the entangled protocol and the same levels of privacy and function estimation quality can be achieved using an unentangled protocol with local Laplace noise added by each node. We prove this in~\cref{app-subsec:no-privacy-advantage-entanglement}. On the other hand, if we wish to achieve Heisenberg-limited scaling in $\varepsilon_{\mathrm{MSE}}$, then $\varepsilon = \bigTheta{\sqrt{n}}$, which means that we lose privacy as we try to get a more accurate function estimate. As we often desire small values for $\varepsilon$, taking $\alpha = 2$ is therefore not desirable. Thus, the advantage of this protocol lies in the intermediate regime, where the user can choose the best value for $\alpha$ that satisfies both their desired privacy and utility.

The scaling in~\cref{thm:noisy-hamiltonian-optimality} follows from the generic bound that we found in~\cref{thm:general-sigma-bound}, and so it applies to arbitrary attacks and arbitrary transcripts obtained from the locally randomized parameter by post-processing. However, this is a conservative bound, and in~\cref{app-subsec:hockey-stick-divergence}, we analyze the privacy using the quantum hockey-stick divergence. Using these techniques, we can find smaller values for $\varepsilon$ while still achieving Heisenberg scaling for the mean-squared error. In particular, when considering the full bit-by-bit learning protocol and for the special case of a one-sample, one-quadrature version of the $K=1$ setting, we find $\varepsilon = \bigTheta{\ln(1+n^{\alpha-1})}$, which is constant for $\alpha=1$ and scales as $\bigTheta{(\alpha-1)\ln{n}}$ for $1 < \alpha \leq 2$. This is markedly better than the bound we get using the generic bound on $b$ in~\cref{thm:general-sigma-bound}. However, for larger values of $K$, we are unable to find simple, closed-form expressions and instead provide numerical results.

Finally, we clarify the adversarial model assumed in the above analysis. Our statement about the scaling of $\varepsilon_{\mathrm{MSE}}$ applies in the \emph{honest setting}, that is, where all of the nodes behave honestly (i.e., not adversarially) by sampling their own noise independently and running the sensing protocol correctly: coupling their qubit to their (noisy) parameter, allowing it to evolve for time $t$, making a measurement, and broadcasting their result. On the other hand, the statements we make about the privacy pertain to the \emph{worst-case scenario}, where all but one of the nodes are assumed to be adversarial. The connection between these two perspectives is that it tells us how much noise must be added to ensure that, if there is only a single honest node, then that node's data is protected with privacy budget $\varepsilon$. If, however, we are in the honest setting, the fact that every node samples noise according to what is required in the worst-case scenario gives us the resulting scaling of $\varepsilon_{\mathrm{MSE}}$ in~\cref{eqn:mse-scaling-nhp-thm}. The key is that any given node does not know which setting they are in, and this gives us the tradeoff in the theorem. We relax this adversarial model in the next subsection, where we show that we can get an improved privacy-utility tradeoff for the noisy Hamiltonian protocol if we allow for an honest fraction of nodes.

\subsection{Noisy Hamiltonian protocol with an honest fraction}\label{subsec:nhp-honest-fraction}
In the previous subsection, we derived a general tradeoff between the achievable level of privacy and the utility of the sensing protocol for the noisy Hamiltonian protocol. While we are able to achieve a more favorable scaling of $\varepsilon = \bigTheta{(\alpha-1)\ln{n}}$ for the specific case of $K=1$ (see~\cref{app-subsec:hockey-stick-divergence}), our general analysis results in a scaling for $\varepsilon$ that is polynomial in $n$, which is undesirable. Recall that the analysis in the previous subsection was for a single honest node protecting its parameter against all other nodes in the network. In this subsection, we consider a less pessimistic setting, where we assume that the noisy Hamiltonian protocol is implemented where at least a constant fraction of the network is honest. As usual, the correctness and privacy statements that we make below refer to two different scenarios. In the honest setting, all nodes follow the protocol, and we analyze the resulting mean-squared error of the function estimate. In the adversarial setting, malicious nodes can behave arbitrarily and do not necessarily follow the protocol as specified. The privacy guarantee in this subsection is thus an aggregate-output guarantee: the released subset-average estimate is protected by the total Gaussian noise contributed by the honest nodes in the queried subset. When the verified GHZ implementation ensures that the honest parameters enter the adversary’s view only through the aggregate phase, the same guarantee applies to any transcript obtained by post-processing that aggregate information. However, we do not claim that the honest-fraction noise scale protects arbitrary unverified transcripts that reveal individual honest nodes’ noisy parameters. We summarize this protocol in~\cref{alg:honest-fraction-nhp}.

\begin{figure}[tb]
    \begin{algorithm}[H]
        \caption{Honest-fraction noisy Hamiltonian protocol}
        \label{alg:honest-fraction-nhp}
        \begin{algorithmic}[1]
            \Require Same input as~\cref{alg:noisy-hamiltonian-qldp}, constants $c>1$,
            $\beta > 1 - 1/c$, subset $S \subseteq [n]$ with $\abs{S} \geq \beta n$,
            privacy parameters $\varepsilon,\delta$
            \Ensure Noisy estimate $\Tilde{Q}_S$ of
            $q_S(\theta) = |S|^{-1}\sum_{i\in S}\theta_i$
            
            \State Choose $\kappa$ such that $\kappa^2 > 2\ln(1.25/\delta)$
            \State $\sigma_\ell \gets \frac{\kappa \Delta}{\varepsilon} \frac{1}{\sqrt{(\beta - 1 + 1/c)n}}$
            \State Generate and verify GHZ state on nodes in $S$
            \If{GHZ verification fails}
                \State \Return $\bot$
            \EndIf
            \For{$i \in S$}
                \State Honest node $i$ samples noise $\eta_i \sim \mathcal{N}(0,\sigma_\ell^2)$
            \EndFor            
            \State Define noisy Hamiltonian $\Tilde{H}_S = \frac{1}{2}\defsum{i\in S}{}{(\theta_i+\eta_i)\sigma_i^z}$
            \For{$m \in [M]$}
                \State Evolve GHZ state under $\Tilde{H}_S$ for time $t$
                \State Measure parity operator on nodes in $S$
                \State $p_m \gets$ parity measurement result
            \EndFor
            \State Estimate noisy parity expectation $\langle P_S\rangle_{\mathrm{est}} \gets \frac{1}{M}\sum_{m=1}^M p_m$
            \State Compute noisy subset-average estimate $\Tilde{Q}_S \gets \frac{1}{\abs{S}t} \arccos\lp\langle P_S\rangle_{\mathrm{est}}\rp$
            \State \Return $\widetilde Q_S$
        \end{algorithmic}
    \end{algorithm}
\end{figure}

We assume that we are promised that at least a $1/c$ fraction of the $n$ nodes in the network are honest, where $c > 1$ is a known constant. If we denote by $S_{\mathrm{hon}} \subseteq [n]$ the honest set, then we have
\begin{equation}
    \abs{S_{\mathrm{hon}}} \geq \frac{n}{c}.
\end{equation}
We also assume a common source of randomness to verify the GHZ state. We restrict the allowed queries to averages of subsets of the nodes' parameters of the form
\begin{equation}
    q_S(\boldsymbol{\theta}) = \frac{1}{\abs{S}}\defsum{i \in S}{}{\theta_i},
\end{equation}
where the queried subset $S \subseteq [n]$ is public and fixed before the execution of the protocol. In particular, the size of the subset satisfies
\begin{equation}
    \abs{S} \geq \beta n
\end{equation}
for some constant
\begin{equation}
    \beta > 1 - \frac{1}{c}.
\end{equation}
This guarantees that every allowable subset contains a linear number of honest nodes. That is, since there are at most $(1 - 1/c)n$ dishonest nodes, any allowable subset $S$ contains at least
\begin{equation}
    \abs{S \cap S_{\mathrm{hon}}} \geq \abs{S} - \lp 1 - \frac{1}{c} \rp n \geq \lp \beta - 1 + \frac{1}{c} \rp n
\end{equation}
honest nodes. In the analysis below, we assume that adversarial nodes do not add noise and do not couple their parameters to their qubits, as each of these things does nothing to help the adversary learn anything about the honest nodes' parameters. As a result, our privacy analysis finds the minimum amount of noise that honest nodes must add in order to achieve $(\varepsilon, \delta)$-differential privacy; any additional noise added by an adversary will simply add more noise than is necessary to achieve the desired level of privacy and can decrease utility. For each distinct allowable query $S$, we assume that each honest node $i \in S$ samples fresh independent Gaussian noise $\eta_i \sim \mathcal{N}(0,\sigma_\ell^2)$, where
\begin{equation}
    \sigma_\ell = \frac{\kappa\Delta}{\varepsilon} \cdot \frac{1}{\sqrt{(\beta-1+\frac{1}{c})n}},\label{eqn:gaussian-sigma}
\end{equation}
where
\begin{equation}
    \Delta \coloneqq \theta_{\max} - \theta_{\min}\quad \text{and}\quad \kappa^2 > 2\ln(1.25/\delta),\label{eqn:delta-kappa}
\end{equation}
in line with the definition of the Gaussian mechanism in~\cref{def:gaussian-mechanism}. We hold the sampled noise fixed for all sensing shots used to answer the query so as to not destroy the sensing advantage (see~\cref{app-subsec:resampled-noise}), but we stress that the noise \emph{is} resampled independently for each distinct query; we discuss this more below. If the same query is asked again, the protocol can simply return the same cached function estimate that was released before, rather than rerunning the sensing protocol with the same privacy budget but new noise. Within each query, each honest participating node evolves its qubit under the noisy Hamiltonian
\begin{equation}
    \Tilde{H}_S = \frac{1}{2}\defsum{i \in S}{}{(\theta_i + \eta_i)\sigma_i^z},
\end{equation}
where adversarial nodes behave arbitrarily. With this, we state the main result of this subsection:

\begin{theorem}[Honest-fraction noisy Hamiltonian protocol for allowed subset averages]\label{thm:nhp-honest-fraction}
    Fix constants $c > 1$, $\beta > 1-\frac{1}{c}$, $\varepsilon > 0$, $\delta \in (0,1)$, and $\kappa^2 > 2\ln(1.25/\delta)$ and let $\Delta$ and $\sigma_\ell$ be as given above in~\cref{eqn:gaussian-sigma,eqn:delta-kappa}, respectively. For every allowed subset $S \subseteq [n]$ satisfying $\abs{S} \geq \beta n$, consider the noisy Hamiltonian protocol, where each honest node $i \in S$ applies noise $\eta_i \sim \mathcal{N}(0,\sigma_\ell^2)$ independently and uses the same, fixed noise for all sensing shots associated with that query. Then, we have the following:
    \begin{enumerate}
        \item \textbf{Correctness in the honest setting}. If all nodes in $S$ are honest and follow the protocol, then the protocol estimates
        \begin{equation}
            q_S(\boldsymbol{\theta}) = \frac{1}{\abs{S}}\defsum{i \in S}{}{\theta_i}
        \end{equation}
        with mean-squared error
        \begin{equation}
            \varepsilon_{\mathrm{MSE}} = \bigOh{\frac{1}{\abs{S}^2t^2}} + \frac{\kappa^2\Delta^2}{\varepsilon^2(\beta-1+\frac{1}{c})n\abs{S}}.
        \end{equation}
        Since $\abs{S} \geq \beta n$, the mean-squared error scales as $\bigOh{1/n^2}$ for constant $c$, $\beta$, $\varepsilon$, and $\delta$, and thus retains Heisenberg scaling.

        \item \textbf{Privacy in the adversarial setting}. Suppose $S_{\mathrm{hon}} \subseteq [n]$ is the set of honest nodes, where $\abs{S_{\mathrm{hon}}} \geq n/c$. For any allowed queried $S$, if the honest nodes follow the protocol, then the released noisy estimate of the average of the subset's parameters is $(\varepsilon,\delta)$-differentially private with respect to changing one honest node's parameter.
    \end{enumerate}
\end{theorem}

\begin{proof}
    We first prove the correctness in the honest setting. When all nodes in the queried subset $S$ are honest, the protocol estimates the function
    \begin{equation}
        \Tilde{q}_S(\boldsymbol{\theta}) = \frac{1}{\abs{S}}\defsum{i \in S}{}{(\theta_i+\eta_i)} = q_S(\boldsymbol{\theta}) + Z_S,
    \end{equation}
    where
    \begin{equation}
        Z_S = \frac{1}{\abs{S}}\defsum{i \in S}{}{\eta_i}
    \end{equation}
    is the aggregate noise contributed by the honest nodes. Conditioned on the sampled noise, the usual entangled sensing analysis gives us a variance of $\bigOh{1/(\abs{S}^2t^2)}$ for $\varepsilon_{\mathrm{MSE}}$. Averaging over the Gaussian noise, we have
    \begin{equation}
        \Var{Z_S} = \frac{1}{\abs{S}^2}\defsum{i \in S}{}{\Var{\eta_i}} = \frac{\sigma_\ell^2}{\abs{S}}.
    \end{equation}
    Using the value for $\sigma_\ell$ given in~\cref{eqn:gaussian-sigma}, we have
    \begin{equation}
        \varepsilon_{\mathrm{MSE}} = \bigOh{\frac{1}{\abs{S}^2t^2}} + \frac{\kappa^2\Delta^2}{\varepsilon^2(\beta-1+\frac{1}{c})n\abs{S}}.
    \end{equation}
    Since we have $\abs{S} \geq \beta n$, both terms scale as $\bigOh{1/n^2}$ for constant $c$, $\beta$, $\varepsilon$, and $\delta$, as claimed in the theorem.

    We now analyze the adversarial setting. The adversary may behave arbitrarily (e.g., not coupling its parameters to its qubits), and so we make no claim on the resulting function estimate accuracy and instead analyze the privacy of the protocol. With the query set $S$ and the set of honest nodes in the full network denoted by $S_{\mathrm{hon}} \subseteq [n]$, where $\abs{S_{\mathrm{hon}}} \geq n/c$, we have
    \begin{equation}
        \abs{S \cap S_{\mathrm{hon}}} \geq \lp \beta-1+\frac{1}{c} \rp n
    \end{equation}
    as the lower bound on the size of the honest set of nodes included in the query. Since only honest nodes are assumed to add noise, the aggregate noise added is
    \begin{equation}
        Z_{S_{\mathrm{hon}}} = \frac{1}{\abs{S}}\defsum{i \in S \cap S_{\mathrm{hon}}}{}{\eta_i},
    \end{equation}
    so we have
    \begin{equation}
        Z_{S_{\mathrm{hon}}} \sim \mathcal{N}\lp 0, \frac{\abs{S \cap S_{\mathrm{hon}}}\sigma_\ell^2}{\abs{S}^2} \rp.
    \end{equation}
    With the lower bound on $\abs{S \cap S_{\mathrm{hon}}}$ from above, we have
    \begin{align}
        \Var{Z_{S_\mathrm{hon}}} &\geq \frac{(\beta-1+\frac{1}{c})n}{\abs{S}^2} \cdot \frac{\kappa^2\Delta^2}{\varepsilon^2(\beta-1+\frac{1}{c})n}\\
        &= \frac{\kappa^2\Delta^2}{\varepsilon^2\abs{S}^2}.
    \end{align}
    The sensitivity of the subset average function $q_S(\boldsymbol{\theta})$ to changing one honest node's parameter value is
    \begin{equation}
        \mathrm{GS} = \frac{\Delta}{\abs{S}},
    \end{equation}
    so by~\cref{def:gaussian-mechanism}, the Gaussian mechanism requires the noise standard deviation to be
    \begin{equation}
        \frac{\kappa\mathrm{GS}}{\varepsilon} = \frac{\kappa\Delta}{\varepsilon\abs{S}},
    \end{equation}
     and so the variance is
     \begin{equation}
         \frac{\kappa^2\Delta^2}{\varepsilon^2\abs{S}^2}.
     \end{equation}
     As we showed above, the honest aggregate noise has at least this variance, and so the released noisy subset average is $(\varepsilon,\delta)$-differentially private by~\cref{thm:gaussian-mechanism-properties}, completing the proof.
\end{proof}

We highlight that the restriction on the type of queries (i.e., the allowed subsets $S$) is necessary. If the adversary could query a subset $S$ containing only one honest node $j$ and otherwise dishonest nodes, then the dishonest nodes could choose not to couple their qubits or add noise to their systems. The released value would then be equivalent to
\begin{equation}
    \Tilde{q}_S \approx \frac{\theta_j + \eta_j}{\abs{S}},
\end{equation}
which the adversary could then simply multiply by $\abs{S}$ to get $\theta_j + \eta_j$. With the noise used in~\cref{thm:nhp-honest-fraction}, $\eta_j$ has a standard deviation scaling as $\bigOh{1/\sqrt{n}}$, so this allows the adversary to learn $\theta_j$ with vanishing error as the size of the network, $n$, grows. The condition $\abs{S} \geq \beta n$ with $\beta > 1 - 1/c$ thus rules out this attack by ensuring that every allowed subset contains at least $(\beta-1+1/c)n$ honest nodes.

We also stress the importance of fresh independent noise being used for each distinct query. If the same $\eta_i$ were used across distinct subset averages, then an adversary could perform a differencing attack using, for example, two queries, one including the full network and one excluding the target parameter:
\begin{equation}
    n\Tilde{q}_{[n]} - (n-1)\Tilde{q}_{[n]\setminus\{j\}} = \theta_j + \eta_j.
\end{equation}
Since $\sigma_\ell = \bigOh{1/\sqrt{n}}$, this vanishes with increasing network size $n$. Instead, using resampled noise, we have
\begin{equation}
    n\Tilde{q}_{[n]} - (n-1)\Tilde{q}_{[n]\setminus\{j\}} = \theta_j + W_j,
\end{equation}
where
\begin{equation}
    W_j = \defsum{i\in S_{\mathrm{hon}}}{}{\eta_i^{(1)}} - \defsum{i\in S_{\mathrm{hon}}\setminus\{j\}}{}{\eta_i^{(2)}},
\end{equation}
where the superscript denotes which query the noise is coming from. Since these two queries use fresh independent Gaussian noise, we have
\begin{equation}
    W_j \sim \mathcal{N}(0, (2\abs{S_{\mathrm{hon}}}-1)\sigma_\ell^2).
\end{equation}
Since $\abs{S_{\mathrm{hon}}} \geq n/c$, we have
\begin{equation}
    \Var{W_j} \geq \lp \frac{2n}{c} - 1 \rp\sigma_\ell^2.
\end{equation}
Using
\begin{equation}
    \sigma_\ell^2 = \frac{\kappa^2\Delta^2}{\varepsilon_0^2(\beta-1+\frac{1}{c})n},
\end{equation}
where $\varepsilon_0$ is the per-query privacy parameter, we have
\begin{align}
    \Var{W_j} &\geq \lp \frac{2}{c}-\frac{1}{n} \rp\frac{\kappa^2\Delta^2}{\varepsilon_0^2(\beta-1+\frac{1}{c})}\\
    \implies W_j &\sim \mathcal{N}\lp 0, \bigOmega{\frac{\kappa^2\Delta^2}{\varepsilon_0^2}} \rp
\end{align}
for constant $c$ and $\beta > 1-1/c$. Thus, the adversary learns $\theta_j$ only up to constant-scale Gaussian noise, independent of $n$, and so the noise does not vanish as we increase $n$, as it did above where we did not resample the noise.

More generally, the adversary may ask for $k$ distinct subset-average queries $q_{S_1}, \ldots, q_{S_k}$, where each $S_r$ satisfies $\abs{S_r} \geq \beta n$ and each run of the protocol uses fresh independent Gaussian noise allowing for $(\varepsilon_r,\delta_r)$-differential privacy. By the composition theorem in~\cref{thm:composition-thm}, the release of all $k$ queries is $(\defsum{r=1}{k}{\varepsilon_r}, \defsum{r=1}{k}{\delta_r})$-differentially private. If $\varepsilon_r = \varepsilon_0$ and $\delta_r = \delta_0$ for all $r \in [k]$, then this release is $(k\varepsilon_0, k\delta_0)$-differentially private. Thus, to achieve a privacy budget of $(\varepsilon, \delta)$ for $k$ distinct queries, it suffices to take $\varepsilon_0 = \varepsilon/k$ and $\delta_0 = \delta/k$, with the corresponding increase in the noise parameter $\sigma_\ell$:
\begin{equation}
    \sigma_\ell = \frac{\kappa'\Delta}{\varepsilon/k} \cdot \frac{1}{\sqrt{(\beta-1+\frac{1}{c})n}},
\end{equation}
where
\begin{equation}
    (\kappa')^2 > 2\ln\lp \frac{1.25k}{\delta} \rp.
\end{equation}
Thus, for constant $k$, the mean-squared error in the honest setting still scales as $\bigOh{1/n^2}$. As discussed in~\cref{subsec:dp-in-sensing}, this transcript-level guarantee also bounds every post-processing attack on the released estimates. In particular, for two candidate values of the target parameter separated by a distance $d \leq \Delta$, the adversary's distinguishing advantage is controlled by the effective privacy parameters $\varepsilon_{\mathrm{eff}} = k\varepsilon_0d/\Delta$ and $\delta_{\mathrm{eff}} = k\delta_0$. Thus, we have
\begin{equation}
    p_{\mathrm{succ}} \leq \frac{e^{k\varepsilon_0d/\Delta} + k\delta_0}{1+e^{k\varepsilon_0d/\Delta}}.
\end{equation}
For small $k\delta_0$ and $d \ll \Delta/(k\varepsilon_0)$, this is close to $1/2$ (i.e., only slightly better than random guessing) and for $d = \Delta$, this reduces to the usual worst-case $(k\varepsilon_0, k\delta_0)$-differential privacy guarantee. We note that, using the strong composition theorem~\cite{dwork2010}, we can improve the dependence on $k$, giving a $\sqrt{k}$-type scaling rather than $k$, at the expense of a larger (but still small) value for $\delta$, but for simplicity we only use the basic composition theorem in~\cref{thm:composition-thm}.

\section{Discussion}\label{sec:discussion}
In this work, we introduced local and global protocols for differentially private quantum sensing in different network settings, where each protocol comes with a set of tradeoffs that the user can balance based on their needs. To measure the performance of each protocol, we introduced a set of three criteria: one correctness condition (the scaling of the mean-squared error on the function estimate) and two soundness conditions (locality and privacy). To the best of our knowledge, our work is the first use of explicit differentially private mechanisms in the context of entangled quantum sensing, and we hope that this work inspires future research along these lines. The privacy of quantum sensing protocols and clients' data should be considered an integral part of any application of quantum sensing, especially as we approach an era in which quantum sensors and other quantum technologies can be implemented regularly and reliably in commercial applications.

There are a number of potential applications for the differentially private quantum sensor network protocols that we have introduced in this paper. Biomedical data is an obvious domain that would benefit from differentially private sensing, but other applications such as geophysical sensors deployed in financially or militarily strategic settings would also benefit. We hope this work lays the groundwork for other exciting applications of differentially private quantum distributed sensing, such as in more general learning problems.

We leave open a number of possible future directions and improvements. One direction is to develop further notions of security, using both differential privacy as we have done in this paper but also invoking cryptographic assumptions (e.g., to simulate a trusted curator). We leave to future work exact mechanisms for realizing correlated noise that can be used by honest nodes to add noise in such a way that Heisenberg scaling is maintained with improved coefficients and milder assumptions while still retaining local privacy. Another approach could be for the network to verify that there is sufficient noise in the GHZ state, before nodes input their data. Furthermore, as mentioned in~\cref{sec:decentralized-network}, we leave to future work the analysis of an explicit GHZ state verification protocol, including errors, that is tailored to the settings we consider in this work. As we show in~\cref{app-sec:ghz-state-leakiness} below, there is a notion of optimality from both the sensing and the privacy perspectives. While we only consider the GHZ state in this work, optimizing over states to strike a balance between optimality for sensing versus privacy could be an interesting dimension to add to the analysis of our protocols.

In addition, we would like to improve the bound we get for $\varepsilon$ in~\cref{thm:noisy-hamiltonian-optimality} by analytically bounding the hockey-stick divergence, which we were only able to do for the specific case of $K=1$. Obtaining a tight analytic result for the full bit-by-bit learning protocol for arbitrary $K$ and $\nu$ would be useful. In particular, we expect that we can relax the assumptions to only a common source of randomness if we allow for $\varepsilon = \bigOh{\ln{n}}$. In this work, we only considered the average function to highlight the privacy aspects of our protocols, but the mechanisms can be adapted to more general linear functions of the form $q(\boldsymbol{\theta}) = \boldsymbol{\alpha} \cdot \boldsymbol{\theta}$, though the sensitivity and allowed-query conditions must be revisited. Further extensions to functions beyond linear~\cite{qian2019} would also be interesting. Finally, in this work we assumed our sensors were qubits; we leave open the possibility of implementing differentially private quantum sensing protocols using bosons~\cite{polino2020,zhuang2018,xia2019,bringewatt2024}, which may open new avenues for physically implementing noise to achieve differential privacy. We hope that this work generates interest in the unification of differential privacy and quantum sensing, and highlights the importance of privacy in practical implementations of quantum sensor networks.

\begin{acknowledgments}
    The authors thank Yusuf Alnawakhtha, Jacob Bringewatt, Adam Ehrenberg, Eleanor Rieffel, Yuxin Wang, and Yu Wei for helpful discussions. D.J.S.~acknowledges support from a graduate research fellowship from the Joint Quantum Institute (JQI) at the University of Maryland, College Park. K.S.~acknowledges support from the U.S. Army Research Office (K.S.) under grant No.~W911NF-20-1-0015. E.T.K.~acknowledges support from the NRC Research Associateship Program at the National Institute of Standards and Technology (NIST), administered by the Fellowships Office of the National Academies of Sciences, Engineering, and Medicine. D.J.S., K.S., E.T.K., and A.V.G.~were supported in part by ONR MURI, AFOSR MURI, the DoE ASCR Quantum Testbed Pathfinder program (award No.~DE-SC0024220), NSF QLCI (award No.~OMA-2120757), NSF STAQ program, DARPA SAVaNT ADVENT, ARL (W911NF-24-2-0107), and NQVL:QSTD:Pilot:FTL. D.J.S., K.S., E.T.K., and A.V.G.~also acknowledge support from the U.S.~Department of Energy, Office of Science, National Quantum Information Science Research Centers, Quantum Systems Accelerator (Award No. DE-SCL0000121), and from the U.S.~Department of Energy, Office of Science, Accelerated Research in Quantum Computing, Fundamental Algorithmic Research toward Quantum Utility (FAR-Qu).
\end{acknowledgments}

\bibliography{references}
\clearpage
\onecolumngrid
\appendix

\section{The GHZ state is in the family of states that leak the most information}\label{app-sec:ghz-state-leakiness}
In this appendix, we prove the claim made in~\cref{subsec:noisy-hamiltonian-protocol} that the GHZ state is among the set of worst states to use from the privacy perspective. We consider the GHZ state because it has been proven before~\cite{eldredge2018} that it is among the set of best states to use from the sensing perspective. We first characterize the least private pure states for a single honest qubit. Then, we extend the result to an honest qubit entangled with an arbitrary adversarial register. This shows that the GHZ state is among the family of least-private states.

We quantify privacy leakage using the quantum hockey-stick divergence $E_\gamma$, where $\gamma = e^\varepsilon$, as defined in~\cref{def:quantum-hockey-stick-divergence}. For a fixed pair of candidate parameter values $\theta$ and $\theta'$, we define the unitary evolution operators
\begin{equation}
    U_\theta = e^{-i\theta t\sigma^z/2},\quad U_{\theta'} = e^{-i\theta't\sigma^z/2}.
\end{equation}
The quantum hockey-stick divergence quantity that measures the distinguishability of states is $E_\gamma(\rho_\theta||\rho_{\theta'})$, where $\rho_\theta$ and $\rho_{\theta'}$ are the post-evolution states we obtain after applying $U_\theta$ or $U_{\theta'}$ to the honest qubit's initial state. We define the quantity $\delta_\theta \coloneqq t(\theta-\theta')$ for convenience. If $\delta_\theta \equiv 0\pmod{2\pi}$, then the two evolutions are identical up to a global phase, and so no state leaks any information about which parameter was used. Thus, in our analysis below, we restrict to cases where $\delta_\theta \not\equiv 0\pmod{2\pi}$.

\begin{lemma}[Least private single-qubit pure state]\label{lemma:least-private-single-qubit}
    For a single qubit, among all pure input states, the states that maximize $E_\gamma(\rho_\theta||\rho_{\theta'})$ are states of the form
    \begin{equation}
        \ket{+_\phi} = \frac{1}{\sqrt{2}}\lp \ket{0} + e^{i\phi}\ket{1} \rp,\label{eqn:single-qubit-pure-state}
    \end{equation}
    where $\phi \in [0,2\pi)$.
\end{lemma}

\begin{proof}
    Write an arbitrary initial single-qubit pure state as
    \begin{equation}
        \ket{\psi_0} = \sqrt{p}\ket{0} + e^{i\phi}\sqrt{1-p}\ket{1},\label{eqn:least-private-single-qubit-ansatz}
    \end{equation}
    where $p \in [0,1]$. After applying the two possible evolution operators $U_\theta$ or $U_{\theta'}$, the possible post-evolution states are
    \begin{equation}
        \ket{\psi_\theta} = U_\theta\ket{\psi_0},\quad \ket{\psi_{\theta'}} = U_{\theta'}\ket{\psi_0}.
    \end{equation}
    The overlap between these states is then
    \begin{equation}
        f \coloneqq \bra{\psi_\theta}\ket{\psi_{\theta'}} = pe^{i\delta_\theta/2} + (1-p)e^{-i\delta_\theta/2}.
    \end{equation}
    Taking the modulus squared, we have
    \begin{equation}
        \abs{f}^2 = 1 - 4p(1-p)\sin^2{\lp \frac{\delta_\theta}{2} \rp}.\label{eqn:f-squared}
    \end{equation}
    For two pure states $\rho_\theta = \ketbra{\psi_\theta}{\psi_\theta}$ and $\rho_{\theta'} = \ketbra{\psi_{\theta'}}{\psi_{\theta'}}$, the only possible positive eigenvalue of the quantity $\rho_\theta - \gamma\rho_{\theta'}$ is
    \begin{equation}
        \lambda_+ = \frac{1-\gamma}{2} + \frac{1}{2}\sqrt{(1-\gamma)^2 + 4\gamma(1-\abs{f}^2)}.
    \end{equation}
    Thus, $E_\gamma(\rho_\theta||\rho_{\theta'}) = \lambda_+$, and so maximizing the hockey-stick divergence is equivalent to minimizing $\abs{f}^2$ in~\cref{eqn:f-squared}, which in turn is equivalent to maximizing
    \begin{equation}
        4p(1-p)\sin^2{\lp \frac{\delta_\theta}{2} \rp}.
    \end{equation}
    Since $\delta_\theta$ is fixed, this quantity is clearly maximized when $p=1/2$. The phase $\phi$ does not affect the overlap, so substituting this value for $p$ into~\cref{eqn:least-private-single-qubit-ansatz}, we see that the least private pure states are exactly of the form
    \begin{equation}
        \ket{+_\phi} = \frac{1}{\sqrt{2}}\lp \ket{0} + e^{i\phi}\ket{1} \rp,
    \end{equation}
    as claimed in the statement of the lemma.
\end{proof}

Note that~\cref{lemma:least-private-single-qubit} applies to pure states. It turns out that mixed single-qubit states cannot leak more information. In particular, we have:

\begin{corollary}[Mixed single-qubit states cannot leak more information]\label{cor:mixed-states}
    No mixed single-qubit input state yields a larger hockey-stick divergence than the pure states of the form given by~\cref{eqn:single-qubit-pure-state}.
\end{corollary}

\begin{proof}
    Let
    \begin{equation}
        \rho_0 = \defsum{i}{}{p_i\ketbra{\psi_i}{\psi_i}}
    \end{equation}
    be an arbitrary mixed input state. Then, the two possible post-evolution states under parameters $\theta$ and $\theta'$ are
    \begin{equation}
        \rho_\theta = \defsum{i}{}{p_iU_\theta\ketbra{\psi_i}{\psi_i}U_\theta^\dagger},\quad \rho_{\theta'} = \defsum{i}{}{p_iU_{\theta'}\ketbra{\psi_i}{\psi_i}U_{\theta'}^\dagger}.
    \end{equation}
    Using the convexity of the hockey-stick divergence~\cite{hirche2023}, we have
    \begin{equation}
        E_\gamma(\rho_\theta||\rho_{\theta'}) \leq \defsum{i}{}{p_iE_\gamma\lp U_\theta\ketbra{\psi_i}{\psi_i}U_\theta^\dagger||U_{\theta'}\ketbra{\psi_i}{\psi_i}U_{\theta'}^\dagger \rp}.
    \end{equation}
    Then, because each pure-state term is bounded above by the value we obtained in the proof of~\cref{lemma:least-private-single-qubit}, the mixture is bounded above by this value.
\end{proof}

We now consider the multi-qubit case, where the first qubit is given to the honest party and the remaining qubits are given to the adversary. We find that the GHZ state is among the \emph{set} of least private states, though it is not the only one. As we show below, the other states in this family of states give the same hockey-stick divergence, so it suffices to consider the GHZ state used in the sensing protocol as providing the least privacy.

\begin{lemma}[Least private pure states with an adversarial system]\label{lemma:least-private-adversary-state}
    Let $H$ denote the honest qubit and $A$ denote an arbitrary adversarial system. Suppose only the honest qubit undergoes the parameter-dependent evolution such that
    \begin{equation}
        \ket{\Psi_\theta} = (U_\theta \otimes I_A)\ket{\Psi_0},\quad \ket{\Psi_{\theta'}} = (U_{\theta'} \otimes I_A)\ket{\Psi_0}.
    \end{equation}
    Among all pure states $\ket{\Psi_0} \in \mathcal{H}_H \otimes \mathcal{H}_A$, the states that maximize the hockey-stick divergence $E_\gamma(\ketbra{\Psi_\theta}{\Psi_\theta}||\ketbra{\Psi_{\theta'}}{\Psi_{\theta'}})$ are of the form
    \begin{equation}
        \ket{\Psi_0} = \frac{1}{\sqrt{2}}\lp \ket{0}_H\ket{u}_A + \ket{1}_H\ket{v}_A \rp,
    \end{equation}
    where $\ket{u}_A$ and $\ket{v}_A$ are arbitrary normalized states of the adversarial system.
\end{lemma}

\begin{proof}
    We start by writing an arbitrary pure state as
    \begin{equation}
        \ket{\Psi_0} = \ket{0}_H\ket{a}_A + \ket{1}_H\ket{b}_A,
    \end{equation}
    where $\ket{a}_A$ and $\ket{b}_A$ are arbitrary (potentially unnormalized) states in the adversary's register and where
    \begin{equation}
        \braket{a} + \braket{b} = 1.
    \end{equation}
    If we denote by
    \begin{equation}
        p_0 = \braket{a},\quad p_1 = \braket{b} = 1-p_0,
    \end{equation}
    then the overlap between the post-evolution states is given by
    \begin{equation}
        f = \braket{\Psi_\theta}{\Psi_{\theta'}} = p_0e^{i\delta_\theta/2} + p_1e^{-i\delta_\theta/2}.
    \end{equation}
    Taking the modulus squared, we find
    \begin{equation}
        \abs{f}^2 = 1 - 4p_0p_1\sin^2{\lp \frac{\delta_\theta}{2} \rp}.
    \end{equation}
    Just as we saw in the proof for~\cref{lemma:least-private-single-qubit}, the hockey-stick divergence is maximized when $\abs{f}^2$ is minimized, which happens when $p_0 = p_1 = 1/2$. Thus, we take
    \begin{equation}
        \ket{a}_A = \frac{1}{\sqrt{2}}\ket{u}_A,\quad \ket{b}_A = \frac{1}{\sqrt{2}}\ket{v}_A,
    \end{equation}
    where $\ket{u}_A$ and $\ket{v}_A$ are now normalized states. Thus, the least private pure states are exactly
    \begin{equation}
        \ket{\Psi_0} = \frac{1}{\sqrt{2}}\lp \ket{0}_H\ket{u}_A + \ket{1}_H\ket{v}_A \rp,
    \end{equation}
    as claimed in the statement of the lemma.
\end{proof}

If we take $\ket{u}_A = \ket{0}^{\otimes(n-1)}$ and $\ket{v}_A = e^{i\phi}\ket{1}^{\otimes(n-1)}$, then the initial state becomes
\begin{equation}
    \ket{\Psi_0} = \frac{1}{\sqrt{2}}\lp \ket{0}^{\otimes n} + e^{i\phi}\ket{1}^{\otimes n}\rp.
\end{equation}
For arbitrary $\phi \in [0,2\pi)$, this is exactly the family of GHZ states, up to a relative phase; taking $\phi=0$ gives the usual GHZ state. Thus, this GHZ family is least private under the no-coupling attack, but note that this set of states is not unique. Taking $\ket{u}_A = \ket{v}_A = \ket{\chi}_A$, we have the product state
\begin{equation}
    \ket{\Psi_0} = \frac{1}{\sqrt{2}}\lp \ket{0}_H + e^{i\phi}\ket{1}_H \rp \otimes \ket{\chi}_A,
\end{equation}
which leaks the same amount of information about the honest party's parameter as the family of entangled states. More generally, for any adversarial mixed state $\sigma_A$, the mixed state
\begin{equation}
    \ketbra{+_\phi}_H \otimes \sigma_A
\end{equation}
achieves the same maximal leakage of information.

Thus, we have shown that the GHZ state is among the set of states that is worst for privacy. The same property that makes it optimal for sensing, namely that it has maximal coherence between the extremal eigenspaces of the generator of field evolution, also makes it the least private. We also showed that the set of states that are worst for privacy is larger than the set of states that are optimal for sensing, as evidenced by the inclusion of both entangled and product states.

Finally, we show that, for the average function and Hamiltonian considered in this work, the set of pure states that are optimal for sensing are of the form $\frac{1}{\sqrt{2}}(\ket{0}^{\otimes n} + e^{i\phi}\ket{1}^{\otimes n})$, and thus are contained in the family of least private states characterized above. This result is a consequence of the work done in Ref.~\cite{eldredge2018}, but we include it here for completeness.

\begin{lemma}[Set of optimal states for sensing average function]\label{lemma:optimal-sensing-states}
    Consider the $n$-qubit Hamiltonian $H = \frac{1}{2}\defsum{i=1}{n}{\theta_i\sigma_i^z}$ and the average function $q(\boldsymbol{\theta}) = \frac{1}{n}\defsum{i=1}{n}{\theta_i}$. Among all pure $n$-qubit states, those that maximize the quantum Fisher information for estimating $q(\boldsymbol{\theta})$, and thus are optimal for sensing, are of the form
    \begin{equation}
        \ket{\psi_0} = \frac{1}{\sqrt{2}}\lp \ket{0}^{\otimes n} + e^{i\phi}\ket{1}^{\otimes n} \rp,\quad \phi \in [0,2\pi).
    \end{equation}
\end{lemma}

\begin{proof}
    Given the average function $q(\boldsymbol{\theta}) = \frac{1}{n}\defsum{i=1}{n}{\theta_i}$, the associated generator is
    \begin{equation}
        G \coloneqq \frac{1}{2}\defsum{i=1}{n}{\sigma_i^z}.
    \end{equation}
    To maximize the sensitivity to $q$, we consider the family of states
    \begin{equation}
        \ket{\psi_q} = e^{-itqG}\ket{\psi_0},
    \end{equation}
    where the quantum Fisher information can be calculated as
    \begin{align}
        F_Q &= 4\lp \braket{\partial_q\psi_q} - \abs{\braket{\psi_q}{\partial_q\psi_q}}^2 \rp\\
        &= 4\lp t^2\mel{\psi_q}{G^2}{\psi_q} - t^2(\mel{\psi_q}{G}{\psi_q})^2 \rp\\
        &= 4t^2\lp \langle G^2 \rangle_{\psi_q} - \langle G \rangle^2_{\psi_q}\rp\\
        &= 4t^2\mathrm{Var}_{\psi_q}[G].
    \end{align}
    $G$ commutes with the evolution operator $e^{-itqG}$ and so its variance is independent of $q$, so we can write
    \begin{equation}
        F_Q = 4t^2\mathrm{Var}_{\psi_0}[G].
    \end{equation}
    As such, maximizing the quantum Fisher information comes down to maximizing the variance of $G$ with respect to the input state $\ket{\psi_0}$. $G$ has eigenvalues $-\frac{n}{2}, -\frac{n}{2}+1, \ldots, \frac{n}{2}-1, \frac{n}{2}$, and so the maximum and minimum eigenvalues are $\lambda_{\max} = \frac{n}{2}$ associated with eigenvector $\ket{0}^{\otimes n}$ and $\lambda_{\min} = -\frac{n}{2}$ associated with eigenvector $\ket{1}^{\otimes n}$, respectively. The variance of $G$, a Hermitian operator, is bounded by
    \begin{equation}
        \Var{G} \leq \frac{(\lambda_{\max} - \lambda_{\min})^2}{4} = \frac{n^2}{4}.
    \end{equation}
    This bound is saturated (i.e., we achieve the equality) only if the state has equally-weighted support on the maximum- and minimum-eigenspaces of $G$, which is achieved with pure states of the form
    \begin{equation}
        \ket{\psi_0} = \frac{1}{\sqrt{2}}\lp \ket{0}^{\otimes n} + e^{i\phi}\ket{1}^{\otimes n} \rp,
    \end{equation}
    where we allow for a relative phase $\phi \in [0,2\pi)$. Thus, we have
    \begin{equation}
        \mathrm{Var}_{\psi_0}[G] = \frac{n^2}{4}
    \end{equation}
    and therefore
    \begin{equation}
        F_Q = 4t^2\mathrm{Var}_{\psi_0}[G] = n^2t^2.
    \end{equation}
    Using the single-shot quantum Cram\'{e}r-Rao bound, we have
    \begin{equation}
        \Var{Q} \geq \frac{1}{F_Q} = \frac{1}{n^2t^2}.
    \end{equation}
    For $M$ shots, we pick up a factor of $M$ in the denominator. Thus, the relative-phase GHZ states are exactly the pure states that maximize the quantum Fisher information and are thus optimal for sensing the average function under this Hamiltonian, as claimed in the lemma.
\end{proof}

\section{Further analysis of noisy Hamiltonian protocol}\label{app-sec:noisy-hamiltonian-analysis}
In this appendix, we provide further analysis of the noisy Hamiltonian protocol, which we presented in~\cref{subsec:noisy-hamiltonian-protocol}. We start by analyzing the privacy of the protocol using the quantum hockey-stick divergence in~\cref{app-subsec:hockey-stick-divergence}. Then, in~\cref{app-subsec:resampled-noise}, we show why we cannot have each node resample their noise with each shot of the sensing protocol. Finally, in~\cref{app-subsec:no-privacy-advantage-entanglement}, we show that if we choose $\alpha=1$ for the noisy Hamiltonian protocol, which drops the scaling of the mean-squared error down to the standard quantum limit, there is no privacy advantage in using entanglement, and so an unentangled protocol may as well be used. 

\subsection{Privacy analysis via the quantum hockey-stick divergence}\label{app-subsec:hockey-stick-divergence}
Recall from the statement of~\cref{thm:noisy-hamiltonian-optimality} that the noisy Hamiltonian protocol is $\varepsilon$-differentially private, that is, $\delta = 0$. In this subsection, we analyze privacy using the quantum hockey-stick divergence defined in~\cref{def:quantum-hockey-stick-divergence}. This allows us to tightly characterize the amount of noise required to achieve privacy, including in the pure case $\delta = 0$. The analysis in this subsection should be distinguished from the analysis that we did in~\cref{subsec:noisy-hamiltonian-protocol}, which gives us a generic local differential privacy guarantee that protects the honest party's information against arbitrary state preparation attacks. Here, instead, we condition on the sensing protocol being run using the GHZ state as specified by the protocol; that is, we assume that the GHZ state verification step has succeeded.

Given that local differential privacy protects individual data at the point of local output, in scenarios involving multiple parties, it suffices to consider the worst-case scenario in which we have a single honest node and $n-1$ adversarial nodes. The strongest attack consists of all of the adversarial nodes colluding to try to learn the honest node's parameter, so, without loss of generality, we combine all of the adversarial nodes into one adversary. We label the honest node as node 1 (i.e., $\theta_1$ is being measured by the honest node). In this worst-case adversarial setting, where adversarial nodes neither couple their sensors to their parameters nor do they apply any noise, the adversarial sensors are granted maximum access to the honest node's parameter. We then compute the distinguishability of the adversary's conditional quantum information as the honest parameter is changed from $\theta_1$ to $\theta_1'$, which gives us a sharper privacy bound than we were able to achieve in our generic analysis in~\cref{subsec:noisy-hamiltonian-protocol}, but it does not apply to arbitrary malicious input states such as product states and is thus more limited in scope.

We consider $(\varepsilon,\delta)$-differential privacy and show how to achieve $\delta = 0$ to yield perfect $\varepsilon$-differential privacy. In the setting of the strongest attack, where the adversarial nodes do not couple their sensors to their parameters but the honest party does couple, the quantum state only encodes information about the honest party's parameter. As such, after the honest party couples their sensor to their noisy field parameter and evolves their part of the shared quantum state while the adversaries' qubits are idle, we have the following final state, where we include the full, bit-by-bit learning protocol detailed in~\cref{subsec:sensor-network}:
\begin{align}
    \ket{\Tilde{\Psi}_f(t) | X_1 = \theta_1, Y_1 = \eta_1} &= \bigotimes_{j=1}^K \ket{\Tilde{\psi}_f(t_j) | X_1 = \theta_1, Y_1 = \eta_1}.\label{eqn:general-final-state}
\end{align}
Here, we make explicit the coupling of the state to the honest node's parameter (denoted by $X_1$) and the honest node's sampled noise (denoted by $Y_1$). Within the tensor product on the right-hand side of~\cref{eqn:general-final-state}, the state for each stage (i.e., each level of precision) is given as
\begin{align}
    \ket{\Tilde{\psi}_f(t_j) | X_1 = \theta_1, Y_1 = \eta_1} &= \lp\lp\Tilde{U}(t_j) \otimes I^{\otimes (n-1)}\rp\ket{\psi_0}\rp^{\otimes 2\nu_j}\\
    &= \lp\frac{1}{\sqrt{2}} \lp e^{-\frac{it_j}{2}(\theta_1+\eta_1)}\ket{0} \otimes \ket{0}^{\otimes (n-1)} + e^{\frac{it_j}{2}(\theta_1+\eta_1)} \ket{1} \otimes \ket{1}^{\otimes (n-1)} \rp\rp^{\otimes 2\nu_j}\\
    &\begin{multlined}
        = \left(\frac{1}{\sqrt{2}} \ket{+} \otimes \frac{1}{\sqrt{2}} \lp e^{-\frac{it_j}{2}(\theta_1+\eta_1)} \ket{0}^{\otimes (n-1)} + e^{\frac{it_j}{2}(\theta_1+\eta_1)} \ket{1}^{\otimes (n-1)} \rp \right.\\
        +\left.\frac{1}{\sqrt{2}} \ket{-} \otimes \frac{1}{\sqrt{2}} \lp e^{-\frac{it_j}{2}(\theta_1+\eta_1)} \ket{0}^{\otimes (n-1)} - e^{\frac{it_j}{2}(\theta_1+\eta_1)} \ket{1}^{\otimes (n-1)} \rp \right)^{\otimes 2\nu_j}.
    \end{multlined}
    \label{eq:dishonest-post-measurement-state}
\end{align}
Here, recall from~\cref{subsec:sensor-network}, at each stage $j$ to learn the $j\numth$ bit of precision of the function, we assume $2\nu_j$ copies of the state, where $\nu_j = \lfloor x_j \rceil$, where $x_j$  is given in Refs.~\cite{belliardo2020,ehrenberg2023} by
\begin{align}
    x_j = \frac{3}{\log_2{C}}(K - j) + x_K,\ \forall j \in [K].
\end{align}
The constants $C$ and $x_K$ are taken as $C = 24.26\pi$ and $x_K = 1$, which is also shown in Ref.~\cite{belliardo2020}. It is important to note that this is a worst-case analysis. In other words, the scenario in which all other sensors act adversarially and attempt to extract information after the honest sensor broadcasts its output is the most challenging situation the honest sensor may encounter. This is because the only way other sensors can extract information from a quantum state is by a POVM. According to the data processing inequality~\cite{wilde2017}, the information content cannot be increased through any local physical operation. Consequently, it suffices to bound the distinguishability of the adversary's complete quantum information, as any later measurement is post-processing.

We analyze the information about $\theta_1$ that could be leaked from the post-measurement state that the adversary holds. After the honest party samples and applies its noise $\eta_1$, and then measures its qubit, the state of the system collapses to a conditional state based on the outcome, that is, $\ket{+}$ or $\ket{-}$: 
\begin{align} \label{eqn:collapsed-state}
    \ket{\phi_+(t_j)|\theta_1,\eta_1} &= \ket{+} \otimes \frac{1}{\sqrt{2}} \lp e^{-\frac{it_j}{2}(\theta_1+\eta_1)} \ket{0}^{\otimes (n-1)} + e^{\frac{it_j}{2}(\theta_1+\eta_1)} \ket{1}^{\otimes (n-1)} \rp \\
    \ket{\phi_-(t_j)|\theta_1,\eta_1} &= \ket{-} \otimes \frac{1}{\sqrt{2}} \lp e^{-\frac{it_j}{2}(\theta_1+\eta_1)} \ket{0}^{\otimes (n-1)} - e^{\frac{it_j}{2}(\theta_1+\eta_1)} \ket{1}^{\otimes (n-1)} \rp 
\end{align}
Repeating this process $2\nu_j$ times on independent copies of the GHZ state yields the following state held by the dishonest parties:
\begin{align} 
    \rho_{\theta_1,\eta_1}(t_j)
    &=
    \left(
    \ket{\phi_+(t_j)|\theta_1,\eta_1}
    \bra{\phi_+(t_j)|\theta_1,\eta_1}
    \right)^{\otimes (2\nu_j-a)}
    \otimes
    \left(
    \ket{\phi_-(t_j)|\theta_1,\eta_1}
    \bra{\phi_-(t_j)|\theta_1,\eta_1}
    \right)^{\otimes a}\\
    &=
    \left(
    \ket{\phi_+(t_j)|\theta_1,\eta_1}
    \bra{\phi_+(t_j)|\theta_1,\eta_1}
    \right)^{\otimes (2\nu_j-a)}
    \otimes
    \left[
    Z_1
    \ket{\phi_+(t_j)|\theta_1,\eta_1}
    \bra{\phi_+(t_j)|\theta_1,\eta_1}
    Z_1
    \right]^{\otimes a},
\end{align}
where $2\nu_j-a$ is the number of copies in the state
$\ket{\phi_+(t_j)|\theta_1,\eta_1}$, and $a$ is the number of copies in the state $\ket{\phi_-(t_j)|\theta_1,\eta_1}$. Moreover,
$Z_1 \coloneqq Z\otimes I^{\otimes(n-1)}$. Here, we highlight the $\eta_1$ dependence explicitly. After the honest party broadcasts its measurement result, all other sensors update their view of the system accordingly, where the new state can be found by integrating over $\eta_1$: 
\begin{align} \label{eqn:ugly-integral}
    \rho_{\theta_1}(t) &= \int_{-\infty}^{\infty}\bigotimes_{j=1}^{K}\rho_{\theta_1,\eta_1}(t_j)\Pr[Y_1=\eta_1]\, d\eta_1\\
    &= \int_{-\infty}^{\infty}\bigotimes_{j=1}^{K}\Bigg[\left(\ket{\phi_+(t_j)|\theta_1,\eta_1}\bra{\phi_+(t_j)|\theta_1,\eta_1}\right)^{\otimes (2\nu_j-a)} \nonumber\\
    &\qquad\qquad\qquad\qquad\otimes\left[Z_1\ket{\phi_+(t_j)|\theta_1,\eta_1}\bra{\phi_+(t_j)|\theta_1,\eta_1} Z_1\right]^{\otimes a}\Bigg]\frac{1}{2b}\exp\left(-\frac{|\eta_1|}{b}\right)\, d\eta_1,
\end{align}
where $\theta_1$ represents the honest party's parameter, incorporated into the state shared by all adversaries. By~\cref{thm:hockey-stick-div}, finding $\delta$ of this algorithm reduces to computing $\Tr[(\rho_{\theta_1} - \gamma\rho_{\theta'_1})^+]$, that is, the sum of the positive eigenvalues of $\rho_{\theta_1} - \gamma\rho_{\theta'_1}$. As the quantity in~\cref{eqn:ugly-integral} is rather complicated, we are only able to get a simple, analytic result for the case where $K = 1$ and $\nu = 1$. For $K > 1$ and $\nu > 1$, we make use of numerical techniques.

We start by analyzing the $K = 1$, $\nu = 1$ case. Recall from~\cref{subsec:sensor-network} that the full bit-by-bit learning protocol uses $2\nu_j$ samples at stage $j$, with $\nu_j$ samples measured in the $X$ basis and $\nu_j$ samples measured in the $Y$ basis as both quadratures are needed to resolve the phase in the full interval $[0,2\pi)$. However, for analytical convenience, we consider a single $X$-basis measurement, which is enough to produce one bit of information about the phase in units of $\pi$, rather than $2\pi$. This is sufficient for the privacy analysis in this subsection. Our goal is to determine $\delta$ and then which value of $\varepsilon$ we need to choose to achieve $\delta = 0$. To do so, we consider
\begin{align}
    \rho_{\theta_1}(t)-\gamma\rho_{\theta'_1}(t) &=
    \int_{-\infty}^{\infty} Z_1^a \left[ \ket{\phi_+(t)|\theta_1,\eta_1} \bra{\phi_+(t)|\theta_1,\eta_1} \right. \nonumber \\
    &\qquad\qquad \left. -\gamma \ket{\phi_+(t)|\theta'_1,\eta_1} \bra{\phi_+(t)|\theta'_1,\eta_1} \right] Z_1^a\frac{1}{2b} \exp\left(-\frac{|\eta_1|}{b}\right) \, d\eta_1.
\end{align}
where $a \in \{0,1\}$. This equation further simplifies to
\begin{equation}
    \begin{gathered}
        \rho_{\theta_1}(t)-\gamma\rho_{\theta'_1}(t) =\ket{\chi_a}\bra{\chi_a} \otimes
        \frac{1}{2} \Big[ (1-\gamma) \left( \ketbra{0}^{\otimes(n-1)} + \ketbra{1}^{\otimes(n-1)} \right)\\
        \qquad\qquad + \frac{1}{1+b^2t^2} \left( e^{-it\theta_1} - \gamma e^{-it\theta'_1} \right) \ketbra{0}{1}^{\otimes(n-1)}\\
        \qquad\qquad + \frac{1}{1+b^2t^2} \left(e^{it\theta_1} - \gamma e^{it\theta'_1} \right) \ketbra{1}{0}^{\otimes(n-1)} \Big],
    \end{gathered}
\end{equation}
where $\ket{\chi_a} \coloneqq Z^a\ket{+}$.
We find the eigenvalues of this state to be
\begin{align}
    \left\{\underbrace{0, \ldots, 0}_{2^{n}-2},\frac{1-\gamma}{2}\pm\frac{1}{2}\frac{1}{1+b^2t^2} \sqrt{1+\gamma^2-2\gamma\cos\left(t(\theta_1-\theta'_1)\right)} \right\}.
\end{align}
We note that the parameter $a$ does not affect the final result: different values for $a$ (i.e., the single-qubit outcome being $\ket{+}$ and $\ket{-}$) only involve additional unitary $Z_1$ gates. As these are unitary changes of basis, they do not change the eigenvalues. The only potentially positive eigenvalue is thus
\begin{align} \label{eq:K=1-delta}
    \lambda_+ &= \frac{1-\gamma}{2}+\frac{1}{2}\frac{1}{1+b^2t^2} \sqrt{1+\gamma^2-2\gamma\cos\left(t(\theta_1-\theta'_1)\right)}\\
    &\leq \frac{1-\gamma}{2}+\frac{1}{2}\frac{1}{1+b^2t^2} \sqrt{1+\gamma^2+2\gamma} \\
    &= \frac{1}{2}\left((1-\gamma)+\frac{1+\gamma}{1+b^2t^2}\right)\\
    &=\frac{2+(1-e^{\varepsilon})b^2t^2}{2(1+b^2t^2)}.
\end{align}
This quantity is less than or equal to $0$ when
\begin{equation}
    \varepsilon \geq \ln\lp 1 + \frac{2}{b^2t^2}\rp,
\end{equation}
or, equivalently, when
\begin{equation}
    b \geq \frac{1}{t}\sqrt{\frac{2}{e^{\varepsilon}-1}}.
\end{equation}
This is the condition for which $\delta = 0$.

Since in this setting we are resolving the final bit up to precision $1/t$, our sensing time is in the range $t \in [2^{K-1}, 2^K]$, where $K$ is the total number of bits of precision. As such, we map the time $t$ to the sensitivity with which we can estimate the parameters: $t \to \frac{1}{\theta_{\max} - \theta_{\min}}$. This gives 
    \begin{align}
        b \geq \left( \theta_{\max} - \theta_{\min} \right)\sqrt{\frac{2}{e^{\varepsilon}-1}},\label{eqn:k1-scale-parameter}
    \end{align}
which is in terms of the input parameters that the user knows at the start of the protocol. For small $\varepsilon \ll 1$ (i.e., high privacy), we can Taylor expand~\cref{eqn:k1-scale-parameter} to get
\begin{align}  
    b &\propto \sqrt{\frac{2}{e^{\varepsilon}-1}} \leq \sqrt{\frac{2}{\varepsilon}}.
\end{align} 
Recall, from~\cref{thm:general-sigma-bound}, we take the scale parameter as $b = \Delta/\varepsilon$, so using the hockey-stick analysis, we find a square-root improvement on the dependence of $\varepsilon$. In summary, the noisy Hamiltonian protocol achieves differential privacy with smaller noise injected than the direct application of the classical Laplace mechanism, thereby offering improved utility under the same privacy requirements.

Recalling the constraint on $b$ from the mean-squared error, $b = \bigOh{n^{-(\alpha-1)/2}}$, which we write as $b = cn^{-(\alpha-1)/2}$ for some constant $c$, to have $\delta = 0$, we must have
\begin{align}
    \varepsilon &\geq \ln{\lp 1 + \frac{2}{c^2t^2}n^{\alpha-1} \rp} = \bigTheta{(\alpha-1)\ln{n}}.
\end{align}
Thus, we see that for this special case of $K = 1$, analyzing the privacy using the quantum hockey-stick divergence, we can actually achieve $\varepsilon$ scaling logarithmically in $n$ for $1 < \alpha \leq 2$ and constant for $\alpha=1$. Compare this to the $\bigTheta{n^{(\alpha-1)/2}}$ scaling that we arrived at using the generic classical bound in~\cref{thm:general-sigma-bound}. Thus, for a fixed scale parameter $b$, which is given by the desired mean-squared error, the quantum hockey-stick analysis improves our bound on $\varepsilon$ by a substantial amount. We expect the $\bigOh{\ln{n}}$ scaling to hold for arbitrary constant $K$, since $K$ is independent of $n$.

If we take $K > 1$ and/or $\nu_j > 1$, we cannot obtain a simple, closed-form result as we did above. In this case, we resort to numerical analysis for small values of $K$ and $\nu_j$. We do this by analyzing the right-hand side of
\begin{equation}
    \delta \geq E_\gamma(\rho_{\theta_1}(t) || \rho_{\theta'_1}(t)) = \Tr[(\rho_{\theta_1}(t)- \gamma\rho_{\theta'_1}(t))^+]
\end{equation}
numerically. In the no-coupling attack, the adversarial register relevant to the honest phase is supported on the two-dimensional subspace spanned by $\ket{0}^{\otimes(n-1)}$ and $\ket{1}^{\otimes(n-1)}$, so for this calculation, the number of adversarial sensors does not change the nonzero spectrum, and it suffices to represent the adversarial register as an effective qubit, which we do by taking $n=2$. We choose specific values of $\theta_1$ and $\theta_1'$ to represent worst-case scenarios in our numerical experiments, that is, the values that are easiest to distinguish. We set $t = 1$ and take $t_j = 2^{j-1}\delta t$, for some unit step of time $\delta t$, following Refs.~\cite{belliardo2020,ehrenberg2023}. We calculate the corresponding privacy parameter $\delta$ (not to be confused with the time step $\delta t$) for scale parameter values $b = 0.2, 0.5, 1.0, 2.0$ and present the results in~\cref{fig:delta-vs-eps}.

\begin{figure}[tb]
    \centering
    \subfloat[]{\includegraphics[scale=0.35]{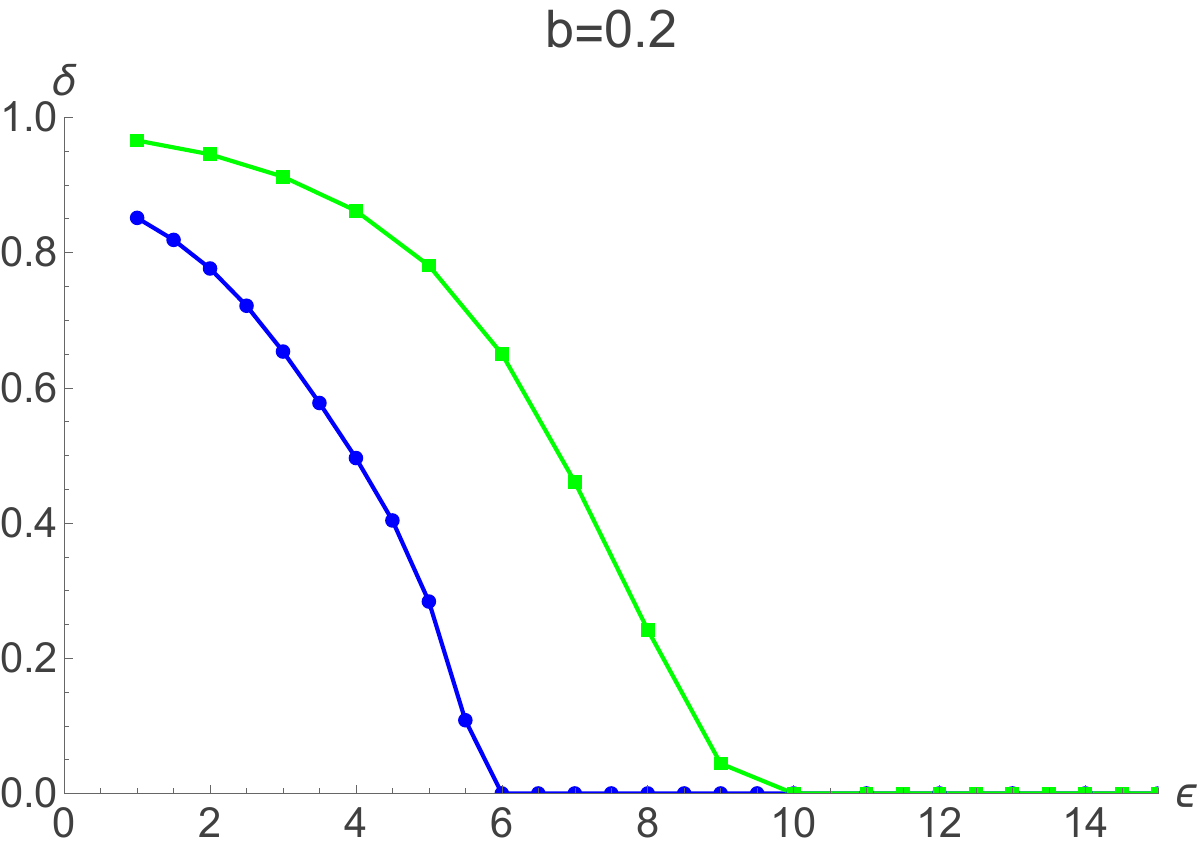}\label{subfig:0.2-sig}}
    \hfill
    \subfloat[]{\includegraphics[scale=0.35]{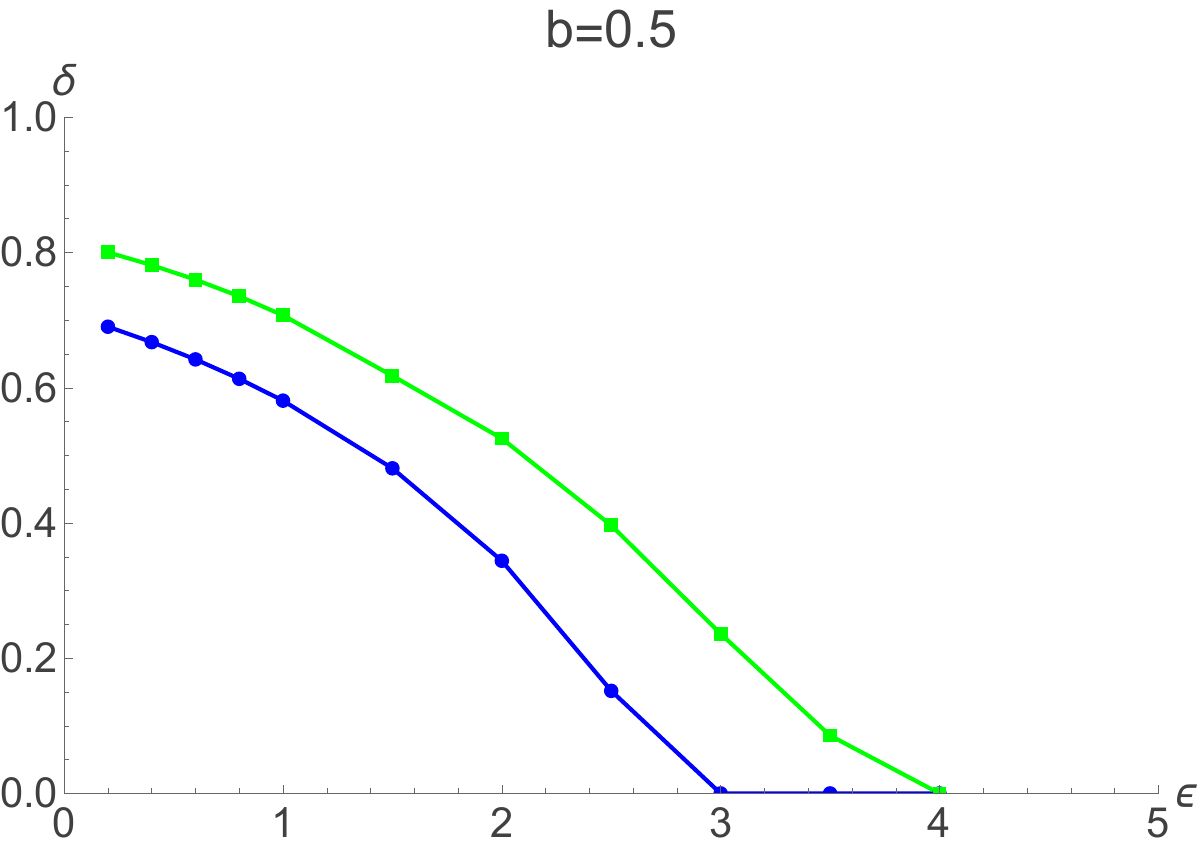}\label{subfig:0.5-sig}}
    \hfill
    \subfloat[]{\includegraphics[scale=0.35]{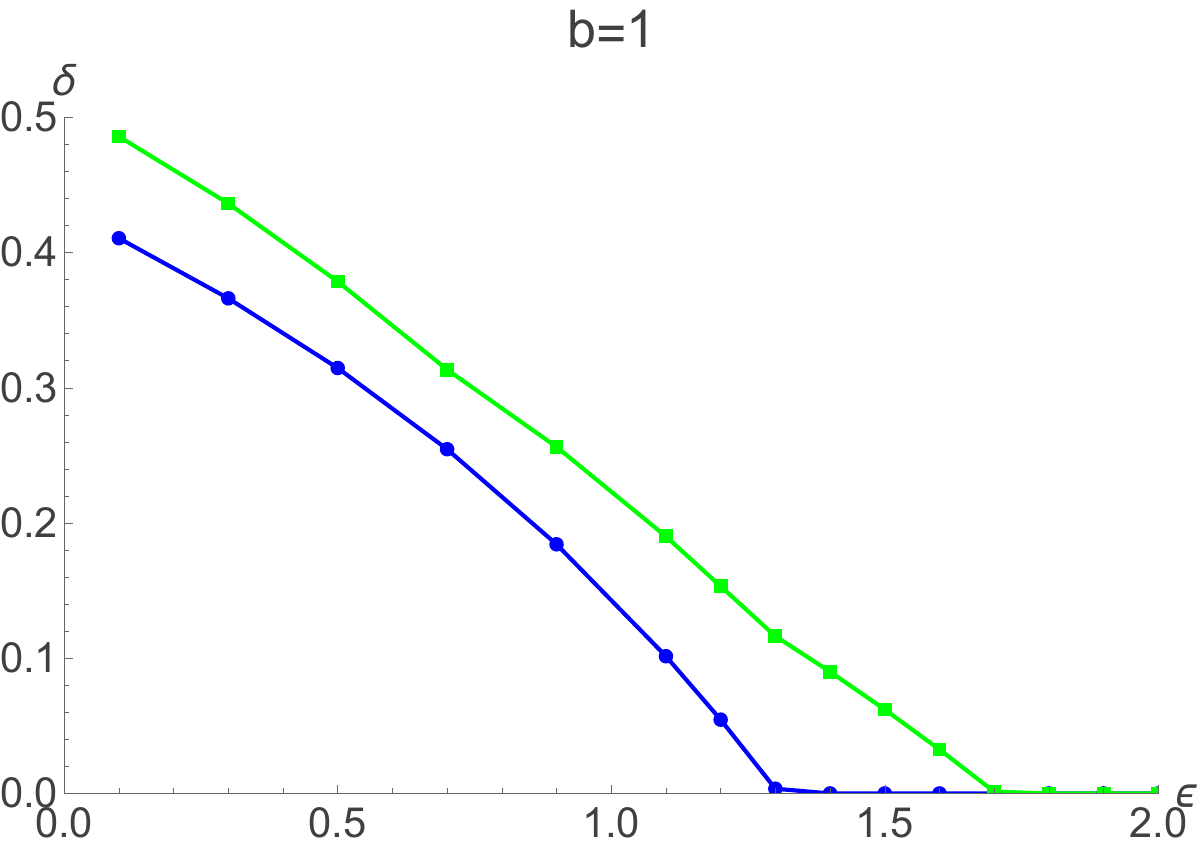}\label{subfig:1-sig}}
    \hfill
    \subfloat[]{\includegraphics[scale=0.35]{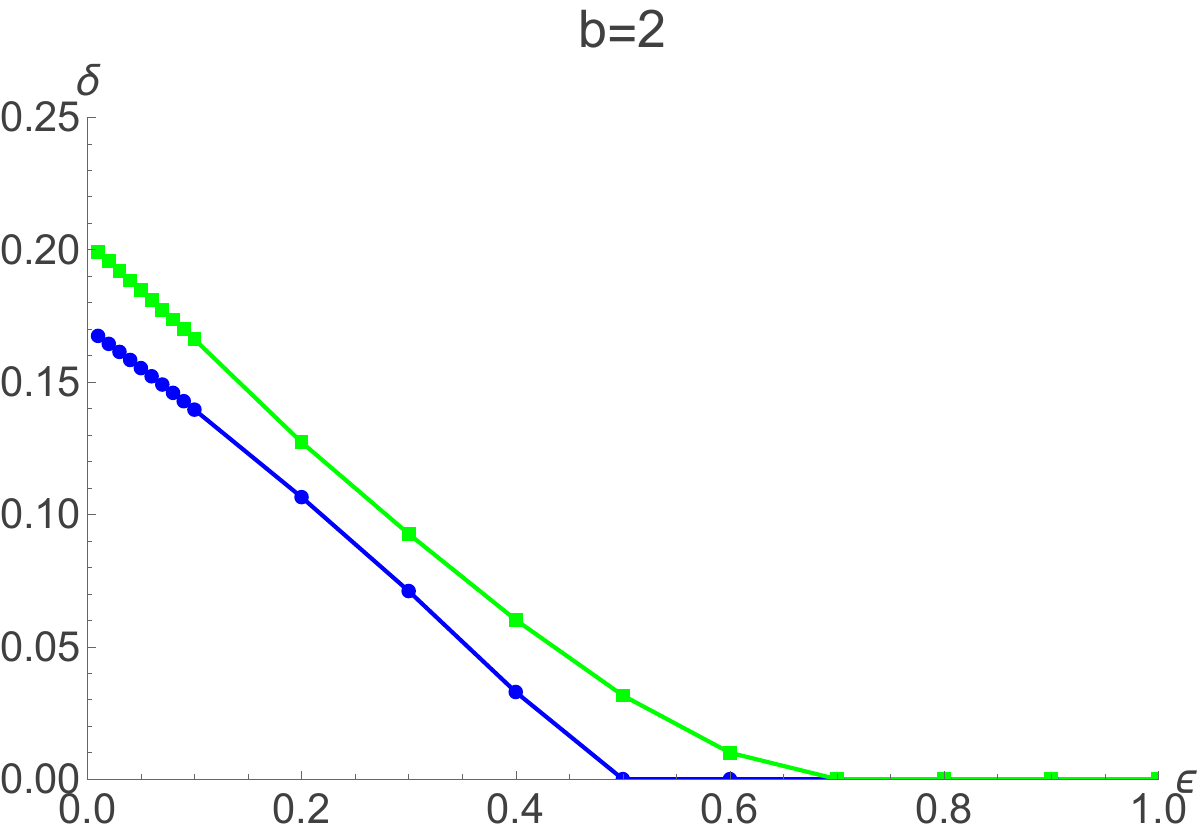}\label{subfig:2-sig}}
    \caption{Probability of failure to achieve $\varepsilon$-differential privacy, $\delta$, as a function of privacy level $\varepsilon$, for varying values of bits of precision $K=1$ (blue), $K=2$ (green) and noise level $b = 0.2, 0.5, 1.0, 2.0$.} 
    \label{fig:delta-vs-eps}
\end{figure}

These results indicate an intuitive trend: as $K$ increases (i.e., we seek more bits of precision), the adversary has more samples from which to obtain information, and so the probability of failure to achieve $\varepsilon$-differential privacy increases. At the same time, as we inject more noise into the system (i.e., we increase $b$), we can achieve a better level of privacy with a lower probability of failure. In real-world applications, we aim to minimize the failure probability, ideally achieving $\delta = 0$. As shown in~\cref{fig:delta-vs-eps}, the corresponding $\varepsilon$ value required to attain $\delta = 0$ increases with the number of bits of precision $K$ for a fixed noise parameter $b$. This is also intuitive: with a more precise function estimate and a fixed noise injection, the level of privacy decreases, resulting in a larger $\varepsilon$.

\subsection{Resampled noise ruins the sensing advantage}\label{app-subsec:resampled-noise}
In this subsection, we show that the protocol in which each node resamples their noise for every shot in the sensing protocol results in a variance for the function estimate that blows up with time $t$ and the number of parties $n$. To show this, we will need two probability distributions: the exponential distribution and the gamma distribution. We also make use of the Laplace distribution, which we already defined in~\cref{def:laplace-distribution}. We start by defining the exponential distribution:

\begin{definition}[Exponential distribution]\label{def:exponential-distribution}
    The \emph{exponential distribution} is the distribution with probability density function
    \begin{align}
        f(x;b) =
        \begin{cases}
            \frac{1}{b}e^{-x/b} & x \geq 0,\\
            0 & x < 0,
        \end{cases}
    \end{align}
    where $b > 0$ is the \emph{scale parameter}.
\end{definition}

The gamma distribution is defined as:

\begin{definition}[Gamma distribution]\label{def:gamma-distribution}
    The \emph{gamma distribution} is the distribution with probability density function
    \begin{align}
        f(x;\alpha,b) =
        \begin{cases}
            \frac{x^{\alpha-1}e^{-x/b}}{b^\alpha\Gamma(\alpha)} & x > 0,\\
            0 & x \leq 0,
        \end{cases}
    \end{align}
    where $\alpha > 0$ is the \emph{shape parameter}, $b > 0$ is the \emph{scale parameter}, and $\Gamma(\alpha)$ is the gamma function evaluated at $\alpha$:
    \begin{align}
        \Gamma(\alpha) = \defint{0}{\infty}{t^{\alpha-1}e^{-t}\mathop{dt}}.
    \end{align}
\end{definition}

With these distributions, we state and prove the following result, showing that the noise cannot be resampled in our protocol:

\begin{theorem}[The noisy Hamiltonian protocol with resampled noise ruins the entanglement advantage]\label{thm:resampled-noise-scaling}
    The variance of the function estimate using the noisy Hamiltonian protocol, with noise resampled between successive shots of sensing, scales as
    \begin{align}
        \Var{Q} = \frac{(1 + b^2t^2)^{2n} - \cos^2(nqt)}{n^2t^2\sin^2(nqt)}.
    \end{align}
\end{theorem}

\begin{proof}
    We start by evolving the initial $n$-qubit GHZ state under the noisy Hamiltonian $\Tilde{H} = \frac{1}{2}\defsum{i=1}{n}{(\theta_i + \eta_i)\sigma_i^z}$ for time $t$ and where $\eta_i$ is the noise that each node samples:
    \begin{align}
        \ket{\Tilde{\psi}_f} &= \Tilde{U}(t)\ket{\psi_0},
    \end{align}
    where
    \begin{align}
        \Tilde{U}(t) = e^{-\frac{it}{2}\defsum{i=1}{n}{(\theta_i + \eta_i)\sigma_i^z}},
    \end{align}
    and where we write the final state $\ket{\tilde{\psi}_f}$ with a tilde to indicate that it is a state obtained under this "noisy" evolution. Thus, the only quantity that we have access to is
    \begin{align}
        \Tilde{q} = q + \eta = \frac{1}{n}\defsum{i=1}{n}{(\theta_i + \eta_i)},
    \end{align}
    where $q = \frac{1}{n}\defsum{i=1}{n}{\theta_i}$ and $\eta = \frac{1}{n}\defsum{i=1}{n}{\eta_i}$. Evolving $\ket{\psi_0} = \frac{1}{\sqrt{2}}\lp \ket{0}^{\otimes n} + \ket{1}^{\otimes n} \rp$ under this operator, we get the following final state with noise injected:
    \begin{align}
        \ket{\Tilde{\psi}_f} &= \Tilde{U}(t)\ket{\psi_0}\\
        &= \frac{1}{\sqrt{2}}\lp e^{-\frac{int}{2}(q + \eta)}\ket{0}^{\otimes n} + e^{\frac{int}{2}(q + \eta)}\ket{1}^{\otimes n} \rp\\
        &= \frac{1}{\sqrt{2}}\lp e^{-\frac{intq}{2} - \frac{itS}{2}}\ket{0}^{\otimes n} + e^{\frac{intq}{2} + \frac{itS}{2}}\ket{1}^{\otimes n} \rp,
    \end{align}
    where we introduced a new variable $S \coloneqq n\eta$. In this case, the noise is sampled anew with every shot, so we cannot just assume a static noise term $\eta$ in our function estimate and instead have to average over all possible values for the noise. We write the final state as a density matrix:
    \begin{align}
        \rho_f &= \defint{-\infty}{\infty}{f(S)\ket{\Tilde{\psi}_f}\bra{\Tilde{\psi}_f}\mathop{dS}},\label{eqn:final-density-matrix}
    \end{align}
    where $f(S)$ is the probability density function (p.d.f.) of the noise term $S$. $S$ is defined as the sum of individual nodes' Laplace-distributed noise $\eta_i$, but we cannot take $S$ to follow a Laplace distribution because the sum of Laplace-distributed random variables is not itself a Laplace-distributed random variable. This is in contrast to, for example, the sum of independently-sampled Gaussian-distributed random variables, which is a Gaussian-distributed random variable. To find the distribution, we use a known relation between a Laplace-distributed random variable and exponentially distributed random variables: if $X_{i_1}, X_{i_2} \sim \expdistr{b}$, where $\expdistr{b}$ is the exponential distribution, then $X_{i_1} - X_{i_2} \sim \Lap{b}$. Thus, $S_i \sim \Lap{b}$ can be written as
    \begin{align}
        S_i \sim X_{i_1} - X_{i_2},
    \end{align}
    where $X_{i_1}, X_{i_2} \sim \expdistr{b}$. A sum of such variables is then
    \begin{align}
        S \sim \defsum{i=1}{n}{S_i} \sim \defsum{i=1}{n}{(X_{i_1} - X_{i_2})} \sim \defsum{i=1}{n}{X_{i_1}} - \defsum{i=1}{n}{X_{i_2}}.
    \end{align}
    We then use the fact that the sum of $n$ exponential-distributed random variables is a gamma-distributed random variable $G_i \sim \text{Gamma}(n,b)$. Thus, we have
    \begin{align}
        S \sim G_1 - G_2,
    \end{align}
    where $G_1 \sim \text{Gamma}(n,b)$ and $G_2 \sim \text{Gamma}(n,b)$ are gamma-distributed random variables. Thus, we now have a random variable equal to the difference of two other random variables, for which we can find the p.d.f. using convolution:
    \begin{align}
        f(S) = \defint{-\infty}{\infty}{f_{G_1}(g)f_{G_2}(g-S)\mathop{dg}}.
    \end{align}
    As mentioned in~\cref{def:gamma-distribution}, the gamma distribution is nonzero for $g > 0$ and $g - S > 0$, so we take $ g > S$. Using this and the p.d.f.~for the gamma distribution $\text{Gamma}(n,b)$, the integral becomes
    \begin{equation}
        f(S) = \frac{e^{S/b}}{\Gamma(n)^2b^{2n}}\defint{\max(0,S)}{\infty}{g^{n-1}(g-S)^{n-1}e^{-2g/b}\mathop{dg}}.
    \end{equation}
    We evaluate this integral separately for $S \leq 0$ and $S > 0$ and find
    \begin{equation}
        f(S) = \frac{\abs{S}^{n-1/2}K_{n-1/2}\lp \frac{\abs{S}}{b} \rp}{(2b)^n\sqrt{\frac{\pi b}{2}}\Gamma(n)},
    \end{equation}
    where $K_\nu(z)$ is the modified Bessel function of the second kind. With this p.d.f., we can now calculate the density matrix in~\cref{eqn:final-density-matrix}. We calculate the expectation value of the parity operator with respect to this final state as
    \begin{align}
        \ev{P} &= \Tr[P\rho_f]\\
        &= \Tr[\defint{-\infty}{\infty}{\frac{1}{2}f(S)\bigotimes_{i=1}^{n}{\sigma_i^x}\lp e^{-\frac{intq}{2} - \frac{itS}{2})}\ket{0}^{\otimes n} + e^{\frac{intq}{2} + \frac{itS}{2})}\ket{1}^{\otimes n} \rp\lp e^{\frac{intq}{2} + \frac{itS}{2})}\bra{0}^{\otimes n} + e^{-\frac{intq}{2} - \frac{itS}{2})}\bra{1}^{\otimes n} \rp\mathop{dS}}]\\
        &= \defint{-\infty}{\infty}{\cos(ntq + tS)f(S)\mathop{dS}}.
    \end{align}
    Evaluating for $f(S)$, we find
    \begin{align}
        \ev{P} = \frac{\cos(nqt)}{(1 + b^2t^2)^{n}}.
    \end{align}
    Using this, we calculate $\Var{P}$:
    \begin{align}
        \Var{P} = 1 - \frac{\cos^2(nqt)}{(b^2t^2 + 1)^{2n}}.
    \end{align}
    We also find
    \begin{align}
        \lp \pdv{\ev{P}}{q} \rp^2 &= \frac{n^2t^2\sin^2(nqt)}{(1 + b^2t^2)^{2n}}.
    \end{align}
    Putting everything together, we have
    \begin{align}
        \Var{Q} &= \frac{\Var{P}}{\lp \pdv{\ev{P}}{q} \rp^2}\\
        &= \frac{(1 + b^2t^2)^{2n} - \cos^2(nqt)}{n^2t^2\sin^2(nqt)},\label{eqn:divergent-variance}
    \end{align}
    as claimed in the theorem.
\end{proof}

We can see that this expression quickly blows up both for fixed $n$ and $t \to \infty$ and for fixed $t$ and $n \to \infty$. Furthermore, for fixed $b > 0$,~\cref{eqn:divergent-variance} grows exponentially in $n$ and therefore destroys even standard quantum limit scaling. As $(1+b^2t^2)^{2n} \approx \exp(2nb^2t^2)$, preserving Heisenberg scaling under resampling would require $nb^2t^2 = \bigOh{1}$, that is, $b = \bigOh{1/(t\sqrt{n})}$. Under the generic Laplace calibration, this translates to having $\varepsilon = \bigOmega{t\Delta\sqrt{n}}$, where, recall, $\Delta = \theta_{\max} - \theta_{\min}$. However, a privacy level that diminishes with the size of the network is undesirable. Thus, resampling is incompatible with simultaneously achieving constant local differential privacy and Heisenberg scaling.

\subsection{Entanglement offers no privacy advantage when \texorpdfstring{$\alpha=1$}{at the standard quantum limit}}\label{app-subsec:no-privacy-advantage-entanglement}
Recall from the main text that when we take $\alpha = 1$ in the noisy Hamiltonian protocol (see~\cref{subsec:noisy-hamiltonian-protocol}), we drop the scaling of $\varepsilon_{\mathrm{MSE}}$ down to the standard quantum limit. As such, there is no benefit in using the entangled protocol, and instead we can just use an unentangled protocol while achieving the same level of privacy:

\begin{theorem}[Unentangled strategy]\label{thm:unentangled-strategy}
    Consider the  noisy Hamiltonian protocol from~\cref{subsec:noisy-hamiltonian-protocol} and the result in~\cref{thm:noisy-hamiltonian-optimality}. In particular, take $\alpha = 1$ such that $\varepsilon_{\mathrm{MSE}} = \bigOh{1/n}$, which scales according to the standard quantum limit, and $\varepsilon = \bigTheta{1}$. An unentangled protocol, where each node independently estimates their parameter and applies local Laplace noise, yields the same scaling for $\varepsilon_{\mathrm{MSE}}$ and $\varepsilon$. Thus, there is no scaling advantage to using an entangled protocol when $\alpha = 1$.
\end{theorem}

\begin{proof}
    If each node adds to their parameter noise sampled from a Laplace distribution with scale parameter $b$, then by the properties of the Laplace mechanism, we have
    \begin{equation}
        b \geq \frac{\theta_{\max} - \theta_{\min}}{\varepsilon},\label{eqn:expression-for-b}
    \end{equation}
    as we have seen before. The mean-squared error of the unentangled protocol is given by
    \begin{equation}
        \varepsilon_{\mathrm{MSE}}^{\mathrm{unent}} = \bigOh{\frac{1}{n}} + \frac{2b^2}{n}.
    \end{equation}
    Choosing a constant privacy budget $\varepsilon = \bigTheta{1}$, from~\cref{eqn:expression-for-b}, this gives $b = \bigTheta{1}$. Thus, we have
    \begin{equation}
        \varepsilon_{\mathrm{MSE}}^{\mathrm{unent}} = \bigOh{\frac{1}{n}}.
    \end{equation}
    In the entangled protocol, we have
    \begin{equation}
        \varepsilon_{\mathrm{MSE}}^{\mathrm{ent}} = \bigOh{\frac{1}{n^2}} + \frac{2b^2}{n} = \bigOh{\frac{1}{n}}
    \end{equation}
    for the same choice of $b = \bigTheta{1}$. Thus, although the entangled protocol has a smaller intrinsic sensing term (i.e., $\bigOh{1/n^2}$ versus $\bigOh{1/n}$), the privacy-noise contribution scales as $1/n$ and determines the overall $n$-scaling. Thus, once we choose $\alpha=1$, an unentangled protocol achieves the same asymptotic scaling in $n$ in both mean-squared error and privacy.
\end{proof}
\end{document}